\documentclass
[
    a4paper,
    DIV=10,
    abstracton
]
{scrartcl}

\usepackage{amsfonts,amssymb,amsmath,amsthm,mathtools}
\usepackage{subcaption,float}
\usepackage{mathabx}
\usepackage{bm}
\usepackage{bbm}
\usepackage{dsfont}
\usepackage{authblk}
\usepackage{paralist}
\usepackage{framed}
\usepackage{algpseudocode,algorithm}

\usepackage[pdftex]{graphicx}
\usepackage[usenames,dvipsnames,table]{xcolor}
\usepackage{epstopdf}
\usepackage{comment}
\usepackage{enumerate}  %for custom lists
\usepackage{tikz-cd} %for nice pictures
\usetikzlibrary{decorations.markings} %for halfway arrows
\usetikzlibrary{arrows,decorations.pathmorphing,backgrounds,fit,positioning,shapes.symbols,chains}
\usetikzlibrary{arrows.meta} % arbitrary arrowhead
\usepackage[
colorlinks,
pdffitwindow=true,
plainpages=false,
pdfpagelabels=true,
pdfpagemode=UseOutlines,
pdfpagelayout=SinglePage,
bookmarks=false,
colorlinks=true,
hyperfootnotes=false,
linkcolor=blue,
urlcolor=blue!50!black,
citecolor=green!50!black]{hyperref}
\usepackage{anyfontsize}

\usepackage{booktabs}
\usepackage{stackrel}
\usepackage{babel}

\usepackage{tikz}
\usetikzlibrary{shapes,arrows}
\usepackage{verbatim}
\usepackage{nicefrac}
\usepackage{bigdelim}

\graphicspath{{graphics/}}

\newtheoremstyle{custom}{3pt}{3pt}{}{}{\bfseries}{:}{.5em}{}
\theoremstyle{custom}
\newtheorem{example}    {Example}%[section]
\newtheorem{definition} [example]{Definition}
\newtheorem{theorem}    [example]{Theorem}
\newtheorem{lemma}      [example]{Lemma}

\newtheorem{corollary}  [example]{Corollary}
\newtheorem{remark}[example]{Remark}
\newtheorem{problem}[]{Problem}
\newtheorem{proposition}[example]{Proposition}

\newcommand{\lowrank}{r}

\newcommand{\tpi}{\widetilde{\Pi}}
\newcommand{\tP}{\smash{\widetilde{P}}}
\newcommand{\tLam}{\smash{\widetilde{\Lambda}}}

\newcommand*{\bP}{\mathbb{P}}
\newcommand*{\bR}{\mathbb{R}}
\newcommand*{\bN}{\mathbb{N}}

\newcommand*{\cC}{\mathcal{C}}
\newcommand*{\cE}{\mathcal{E}}
\newcommand*{\cF}{\mathcal{F}}

\newcommand*{\fB}{\mathfrak{B}}
\newcommand*{\fW}{\mathfrak{W}}
\newcommand*{\Id}{\mathrm{Id}}

\DeclareMathOperator*{\argmax}{arg\,max}
\DeclareMathOperator*{\argmin}{arg\,min}
\DeclareMathOperator*{\diag }{diag}

\DeclarePairedDelimiterX{\infdivx}[2]{(}{)}{#1\;\delimsize\|\;#2}
\newcommand{\KLD}{\mathrm{D_{KL}}\infdivx}

\DeclarePairedDelimiter{\norm}{\lVert}{\rVert}

\definecolor{darkgreen}{rgb}{0,0.7,0}

\newcommand{\rev}[1]{\textcolor{black}{#1}}

\begin{document}
\title{Coherent set identification via direct low rank maximum likelihood estimation}

% \author[1]{Ilja Klebanov}
% \author[2]{P\'eter Koltai}
% \author[3]{Nikolas Nüsken}
% \author[1]{Robert Polzin}

% \affil[1]{Institute of Mathematics, Freie Universit\"at Berlin, Germany}
% \affil[2]{Department of Mathematics, University of Bayreuth, Germany}
% \affil[3]{Department of Mathematics, King's College London, UK}
\author[1]{Robert M. Polzin}
\author[1]{Ilja Klebanov} 
\author[2]{Nikolas Nüsken}
\author[3]{P\'eter Koltai}
\renewcommand\Affilfont{\small}
\affil[1]{Institute of Mathematics, Freie Universit\"at Berlin, Germany}
\affil[2]{Department of Mathematics, King's College London, UK}
\affil[3]{Department of Mathematics, University of Bayreuth, Germany}

\date{}

\maketitle

\begin{abstract}

We analyze connections between two low rank modeling approaches from the last decade for treating dynamical data. The first one is the coherence problem (or coherent set approach), where groups of states are sought that evolve under the action of a stochastic transition matrix in a way maximally distinguishable from other groups. 
The second one is a low rank factorization approach for stochastic matrices, called Direct Bayesian Model Reduction (DBMR), which estimates the low rank factors directly from observed data.
We show that DBMR results in a low rank model that is a projection of the full model, and exploit this insight to infer bounds on a quantitative measure of coherence within the reduced model.
Both approaches can be formulated as optimization problems, and we also prove a bound between their respective objectives.
On a broader scope, this work relates the two classical loss functions of nonnegative matrix factorization, namely the Frobenius norm and the generalized Kullback--Leibler divergence, and suggests new links between likelihood-based and projection-based estimation of probabilistic models.
\end{abstract}


\medskip

\noindent \textbf{MSC classification:} 
65F55, %    Numerical methods for low-rank matrix approximation; matrix compression
62M05, %    Markov processes: estimation; hidden Markov models
37M10, %    Time series analysis of dynamical systems
15A23, % 	Factorization of matrices
% 15B48, %  Positive matrices and their generalizations; cones of matrices\\
% 15B51, %  	Stochastic matrices\\
% 41A29, % 	Approximation with constraints\\
60J22. % Computational methods in Markov chains [See also 65C40]
% 65C40. % Numerical analysis or methods applied to Markov chains [See also 60J22]
% 62M99  %    Inference from stochastic processes; ALTERNATIVE: 62F86 Parametric inference and fuzziness

\section{Introduction}
\label{section:Intro}

\subsection{Motivation and contributions}

One of the fundamental concepts in statistics, data science and machine learning is that seemingly complicated data has an underlying simpler structure; and filtering out this structure is the crucial task of \emph{model reduction}.
Apart from a simpler representation of the data that requires less storage, the main advantages include robustness when used for predictions as well as interpretability of the low-dimensional features.
Since data is often given in the form of matrices $A\in \smash{\bR^{m\times n}}$ \rev{(modeling, for instance, the dependencies or interactions between two variables)}, the above task typically boils down to matrix factorizations of the form $A \approx \smash{BC}$, where $B\in\smash{\bR^{m\times \lowrank}}$ and $C\in \smash{\bR^{\lowrank\times n}}$ are matrices of lower dimensionality, $\smash{ \lowrank\ll \min\{m,n\} }$; see~\cite{udell2016generalized} for a comprehensive overview.
In many applications, data is inherently nonnegative, \rev{for example when frequencies, temperatures or length measurements are involved.} Incorporating the nonnegativity constraint on the factors $B$ and $C$ can add to the interpretability of the low rank approximation, which explains the success of \emph{nonnegative} matrix factorization (NMF) over the past two decades~\cite{lee1999learning,singh2008unified,wang2012nonnegative,li2014nonnegative,gillis2020nonnegative}.
Stochastic matrices are a frequent occurrence of nonnegative matrices in applications.
Often they arise as (or from) data matrices, because the entities encoded by rows and columns of these matrices are in some probabilistic relationship.

One such application, which will be the guiding example of this paper, is determining coherent sets of a dynamical system~\cite{froyland2010transport,FrLlSa10}, as made precise in section~\ref{section:Coherent_sets}.
In a nutshell, for a left stochastic\footnote{Depending on the field, some use the term ``column stochastic''. Both mean that all entries of the matrix are nonnegative and every column sums to one.} ``transition'' matrix $P \in \bR^{m\times n}$, we seek partitions of $\{1,2,\ldots,n\}$ such that random transitions as described by $P$ from different partition elements remain ``maximally dis\-tin\-guish\-able'' in the sense defined in~\eqref{eq:coherence_set} below. We call these partition elements ``coherent sets''.
The concept arose from the desire to understand the transport properties of unsteady flows, see e.g.~\cite{aref1984,romkedar_etal_1990,haller1998finite,froyland_padberg_09,haller2013coherent,aref2017} and references in them. 
% Coherent sets as we use them here have been applied in atmospheric~\cite{froyland2010transport} and oceanographic~\cite{froyland2015studying} applications, and the concept was carried over to study nonequilibrium molecular dynamics~\cite{koltai2016metastability,KWNS18,wu2020variational}.

Approximating the transition matrix $P$ by a product $P \approx \lambda \Gamma$ of two left stochastic matrices $\lambda,\Gamma$ of lower dimension can be interpreted as a (soft) clustering of the cor\-re\-spond\-ing states into coherent sets, even more so if $\Gamma \in \{0,1\}^{\lowrank\times n}$ has binary entries, which corresponds to (hard) clustering.
Such a factorization is precisely what is provided by direct Bayesian model reduction (DBMR), a specific NMF algorithm proposed by \cite{gerber2017toward,gerber2018scalable} for the identification of reduced models directly from the data, i.e.\ without approximating the ``full'' transition matrix $P$ in the first place.
One of the main motivations for this paper is the application of DBMR to the coherence problem described above.
In comparison to the ``classical'' approach popularized by Froyland and others~\cite{froyland2010transport,FrLlSa10}, which relies on a truncated singular value decomposition (SVD) of the transition matrix $P$ (after a suitable rescaling, cf.\ Algorithm~\ref{alg:classical_approach_coherence}), the DBMR approach holds the following promises:
\begin{compactitem}
\item As alluded to by the name, DBMR aims to \emph{directly} infer $\lambda$ and $\Gamma$ from data, without estimating the full transition matrix $P$ in a preliminary step. The model complexity is therefore constrained from the outset, which is expected to lead to improved stability and generalization properties in high-dimensional settings where only relatively few samples are available~\cite{gerber2017toward}.   
%\item
%As mentioned above, the ``full'' transition %matrix $P$ need not be approximated.
%While such an ap\-prox\-i\-ma\-tion often relies on a tremendous amount of samples, \cite{gerber2017toward} argue that the (comparatively few) matrix entries of the low rank ap\-prox\-i\-ma\-tion require far less data. This is particularly appealing for applications in statistical physics, where $mn$ is prohibitively large~\cite{gerber2017toward}.
\item 
Whilst the classical approach requires an ambiguous post-processing step to identify coherent sets, often performed by $k$-means clustering \cite{denner2017coherent}, the DBMR output provides the coherent sets directly through the matrix $\Gamma$, while the matrix $\lambda$ acts as a reduced transition matrix acting on these compound states.
\rev{The left stochasticity of $\lambda$ ensures a form of structure preservation that the classical approach lacks. In the classical method, the `reduced transition matrix' is typically not a stochastic matrix, as it may contain negative entries and columns that do not sum to one.}
% In fact, the entries in each of its columns may not sum to one and some of them can even be negative.
\end{compactitem}
Conceptually speaking, truncated SVD provides an optimal low rank approximation with respect to the Frobenius norm as asserted by the Eckart--Young--Mirsky theorem \linebreak[4] \cite[Theorem~4.4.7]{HsEu15}, whilst DBMR corresponds to a (relaxed) maximum likelihood estimate of $(\lambda,\Gamma)$ and hence minimizes the Kullback--Leibler (KL) divergence between the full and the low rank model, see Remark~\ref{remark:Connection_ML_KL} in section \ref{ssec:DBMR}.

In other words, both SVD and DBMR can be viewed as providing a ``low-complexity'' approximation to the transition matrix $P$,
\begin{equation}
\label{eq:low compl approx}
A^* \in \argmin_{A \, \text{``low-complexity''}}  d(A,P),
\end{equation}
where $d$ alludes to a ``distance-like'' quantity between matrices, and the notion of low-complexity has to be made precise (in our context this will amount to $A$ being low rank).
Building on this parallel and following \cite{lee2000algorithms}, \cite[Section 1.2]{gillis2020nonnegative}, truncated SVD and DBMR can be succinctly formulated as solutions to \textbf{Problem~\ref{prob:Frobenius}} and \textbf{Problem~\ref{prob:DBMR}} below, respectively.
From the perspective of matrix factorization, the Frobenius norm $\| A - BC \|_F$ and the (generalized) Kullback--Leibler divergence $\KLD{A}{BC}$ (essentially applied to the vector consisting of matrix entries \cite{lee2000algorithms}) are two of the most fundamental distances minimized within NMF for the approximation $A \approx \smash{BC}$ discussed above~\cite{wang2012nonnegative}.
For this reason, our comparison of the classical approach to the coherence problem with the one by DBMR should be seen on a broader scale: We derive connections between the above two central objectives of matrix factorization and make the following two main contributions (with the terminology yet to be made precise).
\begin{enumerate}[(i)]
\item
\label{item:Intro_Contribution_Projection}
We prove that the DBMR output corresponds to the composition of the full model $P$ and an orthogonal projection $\Pi$; that is, $P \Pi = \lambda \Gamma$ (Theorem~\ref{thm:projection}).
Based on this insight we deduce that the ``degree of coherence'' contained in the low rank model bounds the degree of coherence contained in the full model from below (Proposition~\ref{prop:sval_comp}).
\item
\label{item:Intro_Contribution_Bounds}
We derive an inequality involving the two measures of distance between the full and the low rank model mentioned above---the Frobenius norm (for the SVD approach) on the one hand, and the Kullback--Leibler divergence (for DBMR) on the other hand (Theorem~\ref{thm:Frob_DBMR_bound}).
To our knowledge, this is the first quantitative relationship between these two classical objectives of matrix factorization.
To this end, we prove and utilize a novel \emph{Pinsker-type inequality}, which could be of independent interest (Proposition~\ref{prop:l2_Pinsker} in Appendix~\ref{section:Pinsker_l_2}).
\end{enumerate}
Next, we turn to the discussion of related work, and the remainder of the paper is structured as follows.
% The necessary notation and the coherence problem are introduced in section~\ref{section:Setup_and_Notation}, while the computationally relevant form of the coherence problem arises from a relaxation detailed in section~\ref{sec:coh_svd}.
Section~\ref{section:Setup_and_Notation} introduces basic notation, while the problem formulation of coherent set identification and the most common computational approach are discussed in section~\ref{section:Coherent_sets}.
In section~\ref{sec:Reducedmodels} we summarize the low rank modeling approach referred to as DBMR. The material in this first part of the paper is encapsulated in \textbf{Problem~\ref{prob:Frobenius}} and \textbf{Problem~\ref{prob:DBMR}} and those serve as the starting point for the analysis in the subsequent second part.
We emphasize that the detailed exposition of existing approaches in sections~\ref{sec:coh_svd} and~\ref{sec:Reducedmodels} is deliberate, in order to have a self-contained manuscript allowing for a direct comparison of these methods.
The novel contributions are derived in sections~\ref{sec:PTO} and~\ref{sec:relation_FrobKL}, then they are illustrated by numerical examples in section~\ref{sec:numerical_examples}, followed by a conclusion in section~\ref{sec:Outlook}.

\subsection{Related Work}
\label{section:Related_Work}

\paragraph{Coherent sets and canonical variables.}
The term ``coherent set'' stems from fluid dynamics and dynamical systems~\cite{froyland2010transport,FrLlSa10}, and the concept has been preceded by transport-related considerations around the term ``coherent structures'' (see the references in these papers). Coherent sets as we use them here have been applied in at\-mo\-spher\-ic~\cite{froyland2010transport} and oceanographic~\cite{froyland2015studying} applications.

The abstract linear-algebraic problem, that the coherence problem boils down to in our setting, is equivalent to Canonical Correlation Analysis~\cite{Hot36} (see \cite{klus2019kernel} for this observation) and it has also been transferred to other applications, e.g., nonequilibrium statistical physics~\cite{koltai2016metastability,KWNS18,wu2020variational}.

\paragraph{Orthogonal NMF and clustering.}
Coherent sets are a special form of clusters. Clus\-ter\-ing itself is strongly related to NMF, in particular, through a modification called Orthogonal NMF and a weighted version of $k$-means clustering \cite{pompili2014two}. Centered around this observation, a body of work on clustering and community detection via NMF has developed~\cite{ding2006orthogonal,yang2013overlapping,wu2018nonnegative,lu2020community,ortiz2022community}; see~\cite{li2014nonnegative} for a survey. From a broader perspective, the generality of the formulation in \eqref{eq:low compl approx} has been exploited to derive other modifications of NMF, for instance replacing the distance-like quantity $d$ \cite{fevotte2009nonnegative,fevotte2011algorithms}; see \cite{gillis2020nonnegative} for an overview. We would also like to point the reader to \cite{shashanka2008probabilistic} which directly links NMF and probabilistic modeling with latent variables.
The orthogonality constraint implicitly appears in DBMR as well: The rows of the DBMR factor $\Gamma$ turn out to be orthogonal, cf.\ Remark~\ref{rem:OrthNMF}, hence, after rescaling of the factors $\lambda$ and $\Gamma$, DBMR satisfies the constraints of orthogonal NMF.

\paragraph{Probabilistic models and estimation.}
With a probabilistic model in the background, the NMF \emph{approximation} problem can be phrased as an \emph{estimation} problem. 
This connection is used in Probabilistic Latent Semantic Analysis (PLSA)~\cite{hofmann1999probabilistic,hofmann2001unsupervised}, where in the original motivation rows and columns of the data matrix correspond to words and documents, respectively.
It has been shown in~\cite{ding2006nonnegative} that PLSA is equivalent to NMF if the latter is formulated in terms of the (generalized) Kullback--Leibler divergence.
As mentioned above as well as in Remark~\ref{remark:Connection_ML_KL}, the Kullback--Leibler divergence is associated with maximum likelihood estimation, hence both PLSA and DBMR essentially compute \emph{most likely low rank models}.
Indeed, Gerber and Horenko compare DBMR with PLSA in \cite[SI sec.~5]{gerber2017toward} and demonstrate superior scaling performance of DBMR for large problems.

\paragraph{Projection-based approximation.}
Considering our contributions \eqref{item:Intro_Contribution_Projection} and \eqref{item:Intro_Contribution_Bounds} mentioned above, we establish a new link between likelihood-based and pro\-jec\-tion-based ap\-prox\-i\-ma\-tion of probabilistic models.
The latter class has also been extensively studied~\cite{deuflhard2005robust,huisinga2005metastability,de2008shrinking,schutte2013metastability}, and provides bounds of the form where the eigenvalue error between original and projected model is bounded from above by the projection error of the associated eigenvectors.
Sharper bounds hold if the model is reversible, essentially meaning that the probability matrix that one seeks to approximate is self-adjoint with respect to a suitable inner product.
This applies to the classical approach to the coherence problem as well, since the singular value decomposition of $A$ is equivalent to the eigenvalue decomposition of $A^\top A$, which is self-adjoint by construction.

\section{Notation}
\label{section:Setup_and_Notation}

Throughout this paper, we denote by $\bR^{r}$ the Euclidean space of dimension $r\in\mathbb{N}$, equipped with the corresponding Euclidean norm $\norm{\cdot}_2$ and inner product $\langle \cdot, \cdot \rangle_2$. 
For a vector $w \in \bR_{>0}^{r}$ with positive entries and $x,y \in \mathbb{R}^r$, let
\begin{itemize}
        \item
        $w^{-1} \coloneqq (w_{j}^{-1})_{j=1,\dots,r}$ denote the componentwise inverse of $w$;
	\item 
	$D_{w} \coloneqq \mathrm{diag}(w) \coloneqq (w_{i}\delta_{ij})_{i,j=1}^r \in \bR^{r\times r}$ denote the corresponding diagonal matrix with entries $w_{i}$ on its diagonal and $\delta_{ij}$ being the Kronecker delta;
	\item
	$\langle x,y \rangle_w := x^\top D_w y = \langle D_w^{1/2}x, D_w^{1/2}y \rangle_2$ denote the $w$-weighted inner product and $\norm{\cdot}_w$ the associated norm.
    We call $x,y \in \bR^{r}$ $w$-orthogonal if $\langle x,y \rangle_w = 0$.
\end{itemize}
For a matrix $A \in \bR^{m \times n}$, $\smash{ \| A \|_F := \big( \sum_{i=1}^m \sum_{j=1}^n |a_{ij} |^2 \big)^{1/2} }$ denotes the Frobenius norm of $A$ and $\sigma_{k}(A)$ the $k$-th largest singular value. We abbreviate the $j$-th column of $A$ by $A_{\bullet j}$.
Throughout, a \emph{left stochastic matrix} $A \in \bR^{m\times n}$ will be a nonnegative matrix such that the entries in each of its columns sum to one.
We will \emph{not} require it to be square, slightly abusing standard terminology.

In the bulk of this work, we consider discrete state spaces modeled by finite sets of the form $[r] \coloneqq \{1,\ldots,r\}$, $r\in\mathbb{N}$. Probability measures on $[r]$ will be identified with probability vectors $p \in \mathbb{R}^r_{\ge 0}$, that is, with vectors having nonnegative entries that sum to one. 
For such probability vectors $u,v\in\bR_{\geq 0}^{m}$ we define the Kullback--Leibler divergence between $u$ and $v$ by
$\KLD{u}{v} = \sum_{i=1}^{m} u_{i} \log \tfrac{u_{i}}{v_{i}}$ if $u$ is absolutely continuous with respect to $v$ (interpreted as probability measures) and $\KLD{u}{v} = \infty$ otherwise.
% \end{itemize}
Here and in what follows, we use the conventions $\log 0 \coloneqq -\infty$, $\tfrac{0}{0} \coloneqq 0$, and set $c \log 0$ to $0, -\infty,+\infty$ for $c=0$, $c>0$, $c<0$, respectively.
Finally, we denote by $\mathds{1}_E \in \mathbb{R}^r$ the indicator vector associated to a subset $E\subseteq [r]$: its $i$-th entry is $1$ if $i\in E$ and $0$ otherwise.

\section{Coherent sets}
\label{section:Coherent_sets}

Following \cite{froyland2010transport}, we introduce the concept of coherent sets induced by a stochastic (or deterministic) transition. In this context, we adapt a space-discrete setting: This can either be viewed as an approximation to a continuous-space dynamics, or else as a genuinely discrete system. Historically, the former case motivated the construction, and a straightforward connection between the space-continuous and discrete settings is briefly summarized in Appendix~\ref{app:cont_discr}.

\subsection{Problem setup}

Let $(\Omega,\Sigma,\bP)$ be a probability space and consider two random variables $X:\Omega\to [n]$ and $Y:\Omega\to [m]$ with distributions $p \in \bR_{\geq 0}^{n}$ and $q \in \bR_{\geq 0}^{m}$, respectively, modeling the state of a random system at initial and final time. 
Here, $n,m\in\bN$ denote the sizes (cardinalities) of the respective discrete state spaces. Often, $X \sim p$ is considered to be an \emph{input} and $Y \sim q$ to be an \emph{output},\footnote{The present situation is sometimes called a \emph{Bayesian relation model}~\cite{gerber2017toward}.
It is a specific, ``two-layer'' instance of a \emph{Bayesian network} or \emph{decision network}, see, for instance, \cite{heckerman1998tutorial}.} which is why the elements of $[n]$ and $[m]$ will be called input and output categories/states, respectively.
We assume that $X$ and $Y$ are coupled through the left stochastic transition matrix $P\in\bR^{m\times n}$,
$
P_{ij} = \bP \left[ Y=i \mid X=j\right],
$
and, hence, follow the joint distribution $\bP \left[ Y=i, X=j\right] = P_{ij}p_j$. Using matrix notation we write $(X,Y) \sim PD_p$ and note the relation $q = Pp$.

We next recall the \emph{coherence problem} for the input-output pair~$(X,Y)$.
On an intuitive level, we would like to obtain a coarse-grained understanding of the situation, for example allowing us to forecast $Y$ given $X$, in a conceptually simple and computationally tractable, yet faithful way.
To this end, we seek nontrivial\footnote{We say that the partition $(E_k)_{k=1}^r$ is nontrivial if none of the sets $E_k$ are empty.} partitions $\mathcal{E} := (E_k)_{k=1}^r$ of $[n]$ and $\cF := (F_k)_{k=1}^r$ of~$[m]$  such that $X\in E_k$ implies $Y\in F_k$ with high probability; or, as we will say, $(E_k,F_k)$ form a \emph{coherent pair}.
The number of subsets $r$ is fixed for now and roughly corresponds to the complexity of the reduced model; typically we shall aim for
$r\ll \min(m,n)$. 
Following \cite{froyland2010transport}, we formulate the following two heuristic conditions for $(E_k,F_k)$ to form a coherent pair:
\begin{enumerate}
\setlength{\itemsep}{0pt}
    \item $\mathbb{P}\left[ Y\in F_k \mid X\in E_k \right] \approx 1$, and
    \item $\mathbb{P}[X\in E_k] \approx \mathbb{P}[Y\in F_k]$.
\end{enumerate}
The first condition demands that states from $E_k$ transition predominantly to~$F_k$. The second condition ensures that, in addition, \emph{exclusively} the states from $E_k$ transition to~$F_k$, up to a small error. Taken together, these two conditions describe the scenario that the pair  $(E_k,F_k)$ evolves coherently, approximately unaffected by the dynamics on the complements $[n] \setminus E_k$ and~$[m] \setminus F_k$. In the standard matrix notation this means that if we group the columns and rows of $P$ according to the $E_k$ and the $F_k$, respectively, then the resulting matrix would have a pronounced block structure: The blocks corresponding to $E_k, F_k$ with the same index $k$ dominate the associated rows and columns, see figures~\ref{fig:CSI_1} and~\ref{fig:CSI_2} below. 

An attempt to accordingly partition the system into a fixed number $r\in\bN$ of coherent pairs is to consider the maximization problem
\begin{equation}
    \label{eq:coherence_set}
    \max_{\substack{\cE,\ \cF \\ \bP[X\in E_k] \approx \bP[Y\in F_k]}} \sum_{k=1}^r \mathbb{P}\left[ Y\in F_k \mid X\in E_k \right],
\end{equation}
carried out over all (nontrivial) partitions $\cE = (E_k)_{k=1}^r$ of $[n]$ and $\cF = (F_k)_{k=1}^r$ of~$[m]$ that respect the second condition from above.
Here, no partition elements are allowed to be empty sets.
To make this a well-posed problem, the constraint $\mathbb{P}[X\in E_k] \approx \mathbb{P}[Y\in F_k]$ needs to be given a quantitative meaning. Note that simply requiring equality may easily render the set of admissible solutions empty. Irrespective of this choice, \eqref{eq:coherence_set} tends to be a computationally hard combinatorial optimization problem; thus  we will later discuss a numerically more approachable relaxation (see Problem~\ref{prob:Frobenius} below).

Any partition $\cE = (E_{1},\dots,E_{\lowrank})$ of $[n]$ can be encoded by an ``assignment'' $\gamma:[n] \to [\lowrank]$ or an ``affiliation matrix'' $\Gamma\in \{ 0, 1\}^{\lowrank\times n}$ via
\begin{equation}
\label{equ:correspondence_partition_affiliation_matrix}
\gamma(j)
\coloneqq
k \text{ with } j \in E_{k},
\qquad
\Gamma_{kj}
\coloneqq
\delta_{k \gamma(j)}
=
\begin{cases}
	1 &\text{if } j \in E_{k},
	\\
	0 &\text{if } j \notin E_{k}.
\end{cases}
\end{equation}
The partition $\cE$ can then be characterized by $E_{k} = \gamma^{-1}( \{k\} )$.
To fix the terminology, we introduce the following notions:
\begin{definition}[Affiliation matrix]
\label{def:hard_affiliation_matrix}
We call a left stochastic matrix with binary entries $\Gamma\in \{ 0, 1\}^{\lowrank\times n}$, $\lowrank,n \in \bN$, a \emph{(hard) affiliation matrix}.
We call the unique map $\gamma:[n] \to [\lowrank]$ satisfying $\Gamma_{kj} = \delta_{k \gamma(j)}$ the \emph{assignment} corresponding to $\Gamma$. \rev{The image set (or range) $\mathrm{Ran} \gamma := \gamma([n])$ will be called the set of \emph{active} latent states.}
\end{definition}

As explained above, the distribution of the pair $(X,Y)$ can be  described in terms of the initial distribution $p$ and the transition matrix~$P$. In practical settings,
these objects are typically ap\-prox\-imated by their empirical (maximum likelihood) estimates based on finitely many samples, $\bm{D} = (X(u), Y(u))_{u=1}^S$, $S \in \bN$, where the pairs $(X(u), Y(u))$ are assumed to be independent and identically distributed copies of~$(X,Y)$.
The data $\bm{D}$ leads us to the \emph{count matrix} $N \in \mathbb{N}_{0}^{m \times n}$ and the \emph{empirical frequency estimators} $\hat{p} \in \bR_{> 0}^{n}$ and $\hat{P} \in \bR_{\geq 0}^{m \times n}$ given by
\begin{equation}
\label{eq:N_and_pP_estimators}
N_{ij} \coloneqq \# \{ u \mid X(u) = j, \, Y(u) = i \},
\qquad
\hat{p}_j = \frac{1}{S} \sum_{i=1}^m N_{ij},
\qquad
\widehat{P}_{ij} = \frac{N_{ij}}{\sum_{i'=1}^m N_{i'j}}.
\end{equation}
Here and in the following we assume the row and column sums of $N$ to be strictly positive for each $i \in [m]$ and $j \in [n]$; otherwise, the associated input or output categories are removed and the sets $[m]$ and $[n]$ are restricted and relabeled accordingly.
Note that $\hat{p}$ and $\hat{P}$ are, in fact, maximum likelihood estimates, cf.\ section \ref{sec:Reducedmodels} and equation \eqref{eq:original_likelihood} in particular.
Note that the maximum likelihood estimate of $q$, $\hat{q} \coloneqq (\sum_{j=1}^{n} N_{ij}/S)_{i \in [m]} \in \bR_{>0}^{m}$ satisfies $\smash{ \hat{q} = \hat{P} \hat{p} }$, inheriting the above-mentioned relation $q = Pp$ of the exact quantities.

\medskip
{
\centering
\fbox{
\begin{minipage}{0.9\textwidth}
As the estimation problem (i.e., the comparison of $P$ with $\widehat{P}$, $p$ with $\hat{p}$, and $q$ with $\hat{q}$) is not the focus of the current work, we will not distinguish the exact quantities from their frequency estimators from now on, i.e.\ we assume them to coincide.
In particular, for simplicity of notation, we will use $P,p,q$ to denote the empirical estimators.
\end{minipage}
}

}
\medskip

\rev{In the case of DBMR, the transition matrix $P$ is approximated by a reduced (i.e., low-rank) matrix $\Lambda = \lambda \Gamma$, where $\Lambda \in \bR^{m\times n}$, $\lambda \in \bR^{m\times \lowrank}$, and $\Gamma \in \bR^{\lowrank\times n}$ are left stochastic matrices. This notation will be used consistently throughout this manuscript.}

\rev{Additionally, as discussed in section~\ref{sec:coh_svd}, it is beneficial to normalize the transition matrices $P$ and $\Lambda$ with respect to the reference distributions $p$ and~$q$:
\begin{equation}
\label{eq:P_tilde}
P' := D_q^{-1} P D_p \in \bR^{m\times n},
\qquad
\tP \coloneqq D_q^{-1/2} P D_p^{1/2} \in \bR^{m\times n},
\qquad
\tLam \coloneqq D_q^{-1/2} \Lambda D_p^{1/2}.
\end{equation}
Note that the normalized transition matrix $P'$ transports \emph{densities} with respect to the reference measure $p$ at the initial time to densities with respect to the reference measure $q$ at the final time. Given that $q = Pp$, it follows immediately that $\smash{ P'\mathds{1}_{[n]} = \mathds{1}_{[m]} }$.
As we will demonstrate, using $\tilde{P}$ instead of $P'$ converts $p$-orthogonality into $\ell^{2}$-orthogonality, simplifying the application of results based on singular value decomposition and, in particular, the Courant--Fischer theorem.}

% \begin{equation}
% \label{eq:P_tilde}
% P':= D_q^{-1}P D_p \in \bR^{m\times n},
% \qquad
% \tP \coloneqq D_q^{-1/2} P D_p^{1/2} \in \bR^{m\times n},
% \qquad
% \tLam \coloneqq D_q^{-1/2} \Lambda D_p^{1/2}.
% \end{equation}

% Further, throughout this manuscript, $\Lambda\in \bR^{m\times n}$, $\lambda \in \bR^{m\times \lowrank}$ and $\Gamma \in \bR^{\lowrank\times n}$ denote left stochastic matrices such that $\Lambda = \lambda \Gamma$ and, for reasons to be clarified in section~\ref{sec:coh_svd} \rev{There will be a revision text here justifying notation placement}, we introduce transition matrices that are normalized with respect to the reference distributions $p$ and~$q$,
% \begin{equation}
% \label{eq:P_tilde}
% P':= D_q^{-1}P D_p \in \bR^{m\times n},
% \qquad
% \tP \coloneqq D_q^{-1/2} P D_p^{1/2} \in \bR^{m\times n},
% \qquad
% \tLam \coloneqq D_q^{-1/2} \Lambda D_p^{1/2}.
% \end{equation}
% Note that the normalized transition matrix $P'$ transports \emph{densities} with respect to the reference measure $p$ at initial time to densities with respect to the reference measure $q$ at final times. By $q =Pp$ it is immediate that~$P'\mathds{1}_{[n]} = \mathds{1}_{[m]}$.

In this context, we also define our main measure of coherence within the pair $(P,p)$, the intuition behind which will be explained in detail in section~\ref{sec:coh_svd}.

\begin{definition}[Degree of coherence]
\label{def:coherence}
    We define the \emph{degree of $r$-coherence} $\cC_{r}(p,P)$ in the pair $(p,P)$ of input distribution and transition matrix as the sum $\sum_{i=1}^r \sigma_i(\tP)$ of the $r$ leading singular values of~$\smash{ \tP \coloneqq D_q^{-1/2} P D_p^{1/2} }$ with $q = Pp$.
    % , where $\tP$ is obtained from the pair as in~\eqref{eq:P_tilde}.
    We simply say that this is the \emph{degree of coherence} in $P$ and write $\cC(P)$, if the integer $r$ and the reference distribution $p$ are clear from the context.
\end{definition}

\subsection{Classical approach to coherent sets}
\label{sec:coh_svd}

The following relaxation of the coherence problem \eqref{eq:coherence_set} can be found in \cite{froyland2010transport} for a two-partition, and in \cite[sec.~3.3]{denner2017coherent} for an arbitrary number of coherent pairs. For details, the reader is referred to these works.
Using \eqref{eq:P_tilde} we obtain
\begin{equation}
    \label{eq:coherence_functional1}
    \mathbb{P}\left[ Y\in F \mid X\in E \right] = \frac{\sum_{i\in F}\sum_{j\in E} P_{ij}p_j}{\sum_{j\in E} p_j} = \frac{\langle \mathds{1}_F, PD_p \mathds{1}_E\rangle_2}{\langle \mathds{1}_E, D_p\mathds{1}_E \rangle_2} = \frac{\langle \mathds{1}_F, P' \mathds{1}_E\rangle_q}{\| \mathds{1}_E \|_p^2}.
\end{equation}
Note that, if $\cE = (E_k)_{k=1}^r$ is a partition of $[n]$, then $\mathds{1}_{E_k},\mathds{1}_{E_l}$ are $p$-orthogonal whenever $k \neq l$, i.e.\ $\langle \mathds{1}_{E_k}, \mathds{1}_{E_l} \rangle_p = 0$.
To obtain a computationally feasible relaxation of the coherent set problem~\eqref{eq:coherence_set}, we relax the condition that the vectors $\mathds{1}_{E_k}$ should be indicator vectors, but keep their $p$-orthogonality. 
We thus replace $\mathds{1}_{E_k}$ by vectors $e_k\in\mathbb{R}^n$ and $\mathds{1}_{F_k}$ by vectors~$f_k\in \mathbb{R}^m$ (how these new vectors can be related back to partition elements is explained below).
Note that, although we do not require the system $f_k$ to be $q$-orthogonal at this stage, this property will be a consequence of our analysis (see below). The constraint $\mathbb{P}[X\in E_k] \approx \mathbb{P}[Y\in F_k]$ can now be required with equality and translates to $\| e_k\|_p = \| f_k\|_q$, yielding the relaxed coherent-set maximization problem
\begin{equation}
    \label{eq:coherence_relax1}
    \max_{\substack{e_1,\ldots,e_r\\ f_1,\ldots,f_r}} \sum_{k=1}^r \frac{\langle f_k, P' e_k \rangle_q}{\|e_k\|_p \|f_k\|_q},
\end{equation}
subject to $e_1,\ldots,e_r$ being a $p$-orthogonal system in~$\mathbb{R}^n$. Since \eqref{eq:coherence_relax1} is invariant under (positive) scaling of $e_k$ and $f_k$, we can further restrict the optimization to unit vectors.
By noting that $f_k = P' e_k / \| P' e_k\|_q$ is a maximizer of the summands for fixed $e_{1},\dots,e_{r}$, this further reduces to
 
\begin{equation}
    \label{eq:coherence_relax2}
    \max_{(e_1,\ldots,e_r) \text{ $p$-orthonormal}} \sum_{k=1}^r \| P' e_k \|_q
    \quad
    \Longleftrightarrow
    \quad
    \max_{(\tilde{e}_1,\ldots,\tilde{e}_r) \text{ orthonormal}} \sum_{k=1}^r \| \tP \tilde{e}_k \|_2,
\end{equation}
after observing that any $p$-orthonormal system $(e_{1},\dots,e_{r})$ in $\bR^{n}$ can be written as $\smash{ e_{k} = D_{p}^{-1/2} \tilde{e}_{k} }$, $k \in [r]$,
% coincides with $(D_{p}^{-1/2} \tilde{e}_{1},\dots, D_{p}^{-1/2} \tilde{e}_{r})$
for some orthonormal system $(\tilde{e}_1,\ldots,\tilde{e}_r)$ in $\bR^{n}$ and that $\norm{P' e_{k}}_{q} = \smash{\norm{D_q^{1/2} P' D_p^{-1/2} \tilde{e}_k}_{2}} = \norm{\tP \tilde{e}_k}_{2}$.
By (the singular value version of) the Courant--Fischer theorem stated in Theorem \ref{thm:Courant_Fischer_singular_values_version}, the right- and left-hand side of \eqref{eq:coherence_relax2} are maximized by the $r$ leading right singular vectors $\tilde{e}_{k}$ of $\tP$ and by $\smash{ e_{k} = D_{p}^{-1/2}\tilde{e}_{k} }$, $k\in [r]$, respectively. The optimal value is equal to the sum of the leading $r$ singular values.

Via the Eckart--Young--Mirsky theorem \cite[Theorem~4.4.7]{HsEu15}, the task \eqref{eq:coherence_relax2} is equiv\-a\-lent to finding the best rank-$r$ ap\-prox\-i\-ma\-tion to $\tP$ with respect to the Frobenius norm and also the spectral norm:
\begin{equation}
\label{equ:leading_r_singular_modes}    
\arg\min_{\substack{\tP_{\textup{red}}\in \mathbb{R}^{m\times n} \\[3pt] \mathrm{rank}\,\tP_{\textup{red}} = r}} \big\| \tP - \tP_{\textup{red}}
\big\|_F = \sum_{k=1}^r \sigma_k(\tP) \widetilde{f}_k \widetilde{e}_k^\top.
\end{equation}
In this light, the relaxed coherence problem is equivalent to a low rank approximation problem of the weighted transition matrix $\tP$. To summarize the discussion so far, we formulate the relaxed coherence problem as follows:

\medskip
{
\centering
\fbox{
\begin{minipage}{0.9\textwidth}
\begin{problem}[Relaxed coherence] 
\label{prob:Frobenius}
Given a count matrix $N\in \mathbb{N}^{m \times n}$ and a fixed rank $r \le \min(m,n)$, find a rank-$r$ matrix $\tP_{\textup{red}} \in \mathbb{R}^{m \times n}$ that minimizes $\big\| \smash{\tP} - \tP_{\textup{red}}
\big\|_F$, where $\smash{\tP}$ is constructed from $N$ via \eqref{eq:N_and_pP_estimators} and \eqref{eq:P_tilde}.
\end{problem}
\end{minipage}
}

}
\medskip

We emphasize that the right-singular vectors of a rank-$r$ matrix $\tP_{\textup{red}} \in \mathbb{R}^{m \times n}$ solving Problem \ref{prob:Frobenius}\rev{, which coincide with the leading $r$ right singular vectors of $\tP$ by~\eqref{equ:leading_r_singular_modes},} satisfy the (right) optimality condition in~\eqref{eq:coherence_relax2}.
Due to the left stochasticity of $P$ the leading singular value of $\tP$ can be shown\footnote{Since $\smash{ q^{1/2} = \tP p^{1/2} }$ and $p,q$ are probability vectors, i.e.\ $\|p\|_2 = \|q\|_2 = 1$, we have~$\sigma_1 \ge 1$. To show $\sigma_1 \le 1$, note that $\sigma_1^2$ is the leading eigenvalue of $\smash{ \tP^{\top}\tP }$ and hence also of the similar matrix~$\smash{ D_p^{-1/2} \tP^{\top}\tP D_p^{1/2} = P^{\top} D_q^{-1} P D_p }$. It is a straightforward calculation that $\| P^{\top} D_q^{-1} P D_p \|_{\infty} = 1$, thus $\sigma_1^2 \le 1$.} to be~$\sigma_{1}=1$, with corresponding right singular vector~$\smash{ p^{1/2} := (\sqrt{p_1}, \ldots, \sqrt{p_n})^\top }$.
The maximizers $e_k$ of the relaxed coherence problem \eqref{eq:coherence_relax2} need not be approximate indicator vectors, but for well-pronounced coherent dynamics, their linear span (and likewise that of the $f_k$) is going to be close to the linear span of indicator vectors~\cite{koltai2016metastability}, see also~\cite{deuflhard2005robust}. This observation suggests, by viewing the singular vectors $e_k$ as \emph{features} of the states $j\in [n]$, various approaches to extract a coherent $r$-partition from the singular vectors: A number of algorithms exist, differing in how post-processing steps are handled, and whether hard or soft clusters are sought. One can use k-means clustering \cite[sec.~3.3]{denner2017coherent}, \cite{banisch2017understanding}, PCCA+ ~\cite{deuflhard2005robust,roblitz2013fuzzy}, or SEBA~\cite{froyland2019sparse}. The final-time members of the coherent pairs are obtained in a similar manner from the vectors $\smash{ f_k = D_q^{1/2}\tilde{f}_k }$, where $\tilde{f}_k$ are the left singular vectors of~$\tP$, and are matched to the initial-time members such that the objective in~\eqref{eq:coherence_set} is maximal.
We summarize this \emph{classical approach} to coherent pairs, using k-means clustering in the postprocessing step, in Algorithm~\ref{alg:classical_approach_coherence}.

\begin{algorithm}[htb]
\caption{Classical approach to coherent pairs.}
\label{alg:classical_approach_coherence}
\begin{algorithmic}[1]
    \State \textsc{Input}: Data subsumed into the count matrix $N$, number of coherent pairs $r\in\mathbb{N}$
    \State Compute $P,p,\tP$ via~\eqref{eq:N_and_pP_estimators} and~\eqref{eq:P_tilde}
    \State Compute SVD of $\tP$: $\tP = U\Sigma V^{\top}$ with orthogonal matrices $U\in \bR^{m \times m}$, $V\in \bR^{n \times n}$ and rectangular matrix
    \[
    \Sigma = \begin{pmatrix}
    \textrm{diag}(\sigma_{1},\dots,\sigma_{s}) & 0 \\ 0 & 0
    \end{pmatrix} \in \bR^{m \times n}
    \]
    with the singular values $\sigma_{1} \geq \cdots \geq \sigma_{s} > 0$ of $\tP$ on its diagonal
    \State Truncate to the $r$ leading singular values to obtain a low rank approximation $\tP_{\textup{red}}$ of $\tP$,
    \[
    \tP_{\textup{red}}
    =
    U \begin{pmatrix} \diag (\sigma_{1},\dots,\sigma_{r}) & 0 \\ 0 & 0
    \end{pmatrix} V^{\top},
    \]
    and transform back to obtain a low rank approximation $P_{\textup{red}} \coloneqq D_{q}^{-1/2} \tP_{\textup{red}} D_{p}^{1/2}$ of~$P$
    \State Use the first $r$ right singular vectors $\tilde{e}_{1},\dots,\tilde{e}_{r}$ of $\tP$ (the first $r$ columns of $V$) as features for a k-means-clustering of $[n]$ into $r$ clusters \cite[sec.~3.3]{denner2017coherent}, leading to the partition $\cE = (E_k)_{k=1}^r$ of~$[n]$; and obtain similarly the partition $\cF = (F_k)_{k=1}^r$ of~$[m]$ using the left singular vectors $\tilde{f}_{1},\dots,\tilde{f}_{r}$ of~$\tP$
    \State ``Match'' the partitions $\cE$ and $\cF$ by reordering $\cF$ such that the objective in~\eqref{eq:coherence_set} is maximized
    \State \textsc{Output}: $\tP_{\textup{red}}$, $P_{\textup{red}}$, $\cE$, $\cF$
\end{algorithmic}
\end{algorithm}

The relationship between \eqref{eq:coherence_functional1} and \eqref{eq:coherence_relax1} shows that the value of the latter is a measure for the coherence of an $r$-partition which motivates our usage of this value as the ``degree of coherence'' in Definition~\ref{def:coherence}.
Since the optimal value for \eqref{eq:coherence_relax1} is the sum $\sum_{i=1}^r \sigma_i(\tP)$ over the $r$ leading singular values of~$\tP$, it follows that the degree of coherence of an $r$-partition is bounded from above by~$r$.
In other words, tightness of the bound $\sum_{i=1}^r \sigma_i \le r$ indicates coherence of  
 the system at hand.
In the case of complete coherence ($\sum_{i=1}^r \sigma_i = r$), the transition matrix $P$ (and hence $\smash{\tP}$ as well) has the form where there are partitions $(E_k)_{k=1}^r$ of $[n]$ and $(F_k)_{k=1}^r$ of $[m]$ such that $P_{ij}>0$ implies $i\in E_k$ and~$j\in F_k$ for the same~$k$.

The question of how to choose the number $r$ of coherent sets can be answered by considering the singular spectrum of $\tP$ \cite{froyland2013analytic}:
The aim is to have the leading $r$ singular values close to one (and the corresponding singular vectors close to linear combinations of indicator vectors $\mathds{1}_{E_k}$ corresponding to some partition $\cE = (E_k)_{k=1}^r$), while the remaining singular values should, ideally, be substantially smaller than 1 (indicating no further coherence within the system).
Consequently, the choice of $r$ should be informed by the values and gaps in the spectrum.

\begin{remark}
\label{rem:CCA}
The observation, mentioned in section~\ref{section:Intro}, that the algebraic form of the relaxed coherence problem is equivalent to Canonical Correlation Analysis (CCA), has been made in~\cite{klus2019kernel}. CCA is commonly described as a method that finds bases of the input and output space with maximal correlation under an assumed probabilistic relationship.
\end{remark}

%%%
\section{Likelihood-based estimation from data}
\label{sec:Reducedmodels}

\subsection{Full versus low rank models}
\label{sec:Bayesian_model}

Solving the relaxed coherence problem as presented so far is a two-step procedure: First, $p$ and $P$ are estimated from observational data, and second, the dominant singular vectors are extracted from $\tP$ (see Problem \ref{prob:Frobenius}). Thus, it is natural to ask whether a low rank approximation of $\tP$ can be obtained directly, merging estimation and projection.

For this purpose, recall the empirical estimators $\widehat{p}$ and $\widehat{P}$ for $p$ and $P$ from \eqref{eq:N_and_pP_estimators}.
It is classical and straightforward to show that $\pi^{\ast} = \widehat{p}$, $\Lambda^{\ast} = \widehat{P}$ maximize the likelihood of the data $\bm{D}$,
\begin{equation}
\label{eq:original_likelihood}
\mathbb{P}\left[\bm{D} \mid \pi, \Lambda \right] = \prod_{i=1}^m \prod_{j=1}^n \pi_j^{N_{ij}} \Lambda_{ij} ^{N_{ij}}\,.
\end{equation}
Reiterating the discussion from section \ref{section:Coherent_sets}, we overload the notation, dropping the hats in $(\pi^*,\Lambda^*) = (\widehat{p},\widehat{P})$, and hence denoting the true objects $(p,P)$ and their empirical (maximum likelihood) estimators by the same symbols.
Traditionally, these maximum likelihood estimates are then used to approximate $\tP$ and its $r$ leading singular modes, from which a coherent $r$-partition can be extracted in various ways, see section~\ref{sec:coh_svd}.
Note that this approach requires the approximation of $n + mn$ probability values \rev{for the estimation of $(p,P)$ in \eqref{eq:N_and_pP_estimators}}, some of which might be very small, hence a large number of samples if often required.

However, if we are interested in a fixed number $\lowrank\ll \min ( m,n )$ of singular modes, the effective information of interest is already represented by $\mathcal{O}(\lowrank (n+m))$ quantities, and hence it might be expected that a direct approach (circumventing the estimation of $p$ and $P$) could provide accurate results based on a significantly reduced number of samples.
More specifically, we will contrast the traditional procedure with a direct estimation of $P$ by a low rank transition matrix $\Lambda = \lambda \Gamma$, where $\lambda \in \bR^{m\times \lowrank}$ and $\Gamma \in \bR^{\lowrank\times n}$ are left stochastic matrices still fulfilling~\rev{$q = \Lambda p$}, in the maximum likelihood framework of DBMR.
Indeed, DBMR requires the computation of only $\lowrank (n+m)$ matrix entries and thus its output promises to be of low variance, even in the regime where only a few samples are available~\cite[Theorem; in particular Eq.~7]{gerber2017toward}.

In terms of interpretability, our alternative approach has another crucial advantage:
DBMR maximizes a lower bound of the log-likelihood function, the optimum $\Lambda^{\ast} = \lambda^{\ast} \Gamma^{\ast}$ of which turns out to comprise a (hard) affiliation matrix $\Gamma^{\ast}$ in the sense of Definition~\ref{def:hard_affiliation_matrix} (though starting with the assumption on $\Gamma$ to be only left stochastic).
As mentioned in section \ref{section:Coherent_sets}, there is a one-to-one correspondence between such affiliation matrices and partitions of $[n]$, providing a meaningful partition $\cE := (E_k)_{k=1}^r$ of~$[n]$ \emph{without} any post-processing steps, while $\lambda$ corresponds to a `reduced transition matrix' on these compound states.
A natural choice for the output partition $\cF := (F_k)_{k=1}^r$ of~$[m]$ is given by
\begin{equation}
\label{eq:Sets_F_k}
F_{k} = \{ i \in [m] \mid \lambda_{ik} = \max_{k' \in [\lowrank]} \lambda_{ik'} \}
\end{equation}
(with arbitrary choice of category in case of non-uniqueness of the maximizer).

\subsection{Direct Bayesian model reduction (DBMR)}
\label{ssec:DBMR}

In the following, we will discuss how the low rank model $\Lambda = \lambda \Gamma \approx P$, where $\lambda \in \bR^{m\times \lowrank}$ and $\Gamma \in \bR^{\lowrank\times n}$ are left stochastic matrices, can be estimated from the data in a (relaxed) maximum likelihood fashion similar to the derivation of \eqref{eq:N_and_pP_estimators} from \eqref{eq:original_likelihood}.
This approach, proposed by Gerber and Horenko \cite{gerber2017toward}, achieves both estimation and model reduction simultaneously, without the need to estimate the full model $P$ in the first place.
In other words, we assume that the output depends on the input through a \emph{latent variable} $Z \in [\lowrank]$, illustrated by the graphical model $\smash{ X \stackrel{\Gamma}{\longrightarrow} Z \stackrel{\lambda}{\longrightarrow} Y }$, encapsulating the conditional independence assumption $\mathrm{Law}[Y \mid X,Z] = \mathrm{Law}[Y \mid Z]$.
In this case, $\Gamma$ and $\lambda$ correspond to the transition matrices to and  from the latent state, respectively:
\begin{equation} \label{eq:cond_prob}
\Gamma_{kj}=\mathbb{P}\big[Z=k \mid X=j\big],
\qquad
\lambda_{ik}=\mathbb{P}\big[Y = i \mid Z = k \big].
\end{equation}
Note that we can interpret $\Gamma_{kj} \in [0,1]$ as a (soft) \emph{affiliation} of input category~$j$ to the latent state~$k$. 
As we will see below, the DBMR solution in fact yields binary estimates $\Gamma_{kj}\in \{ 0,1 \}$, interpreted as hard affiliations as in Definition~\ref{def:hard_affiliation_matrix}.

Since \eqref{eq:original_likelihood} can be split into two optimization problems, one for $\pi$ and one for $\Lambda$, estimating the factors $\lambda$ and $\Gamma$ from the observation data~$\bm{D}$ via maximum likelihood estimation reduces to maximizing 
\begin{equation}
\label{eq:new_prob}
\ell(\lambda,\Gamma)
\coloneqq
\log \mathbb{P}[\bm{D}\mid\lambda\Gamma]
=
\sum_{i=1}^{m} \sum_{j=1}^{n} N_{ij}\log(\lambda\Gamma)_{ij},
\end{equation} 
over all pairs $(\lambda,\Gamma)$ of left stochastic matrices.
Since the full model $P = \widehat{P}$ maximizes \eqref{eq:original_likelihood} without the low rank constraint, we obtain the natural bound~$\ell(\lambda,\Gamma) \le \ell(P,\Id_n)$.
\rev{Since \eqref{eq:new_prob} might be challenging to maximize computationally, \cite{gerber2017toward} suggest to relax the problem and maximize a lower bound of $\ell$ instead,
\begin{equation}
\label{eq:DBMR_objective}
\hat\ell(\lambda,\Gamma)
:=
\sum_{i=1}^{m} \sum_{j=1}^{n} \sum_{k=1}^{\lowrank}
N_{ij}\Gamma{}_{kj}\log \lambda_{ik}
\leq
\ell(\lambda,\Gamma),
\end{equation}
where we have applied Jensen's inequality. We note that both \eqref{eq:new_prob} and \eqref{eq:DBMR_objective} do not admit closed-form maximizers, but the fact that $\Gamma_{kj}$ and $\log \lambda_{ik}$ split multiplicatively in \eqref{eq:DBMR_objective} gives rise to an elegant two-step alternating procedure (see \cite{gerber2017toward} and the discussion below).
We summarize the relaxed coherence objective as follows:}

\medskip
{
\centering
\fbox{
\begin{minipage}{0.9\textwidth}
\begin{problem}[DBMR]
\label{prob:DBMR}
Given a count matrix $N\in \mathbb{N}^{m \times n}$ and a fixed rank $\lowrank \le \min(m,n)$, find left stochastic matrices $\lambda \in \bR^{m\times \lowrank}$ and $\Gamma \in \bR^{\lowrank\times n}$ that maximize $\hat\ell$ given by \eqref{eq:DBMR_objective}.
\end{problem}
\end{minipage}
}

}
\medskip

The DBMR algorithm \ref{alg:the_alg} suggested by \cite{gerber2017toward} maximizes $\hat\ell(\lambda,\Gamma)$ by an alternating optimization over $\lambda$ and~$\Gamma$ with a computational cost that is linear in~$m$ and $n$:
Maximizing $\hat\ell(\lambda,\Gamma)$ for fixed $\Gamma$ yields a unique optimum
\begin{equation}
\label{eq:lamda_hat}
    {\lambda}_{ik}=\frac{\sum_{{j=1}}^{n}{\Gamma}_{kj}N_{ij}}{\sum_{{i'=1}}^{m}\sum_{{j'=1}}^{n}{\Gamma}_{kj'}N_{i'j'}}.
\end{equation}
On the other hand, for any fixed left stochastic matrix $\lambda$, maximizing $\hat\ell(\lambda,\Gamma)$ with respect to $\Gamma$ decouples into $n$ separate linear programs~\cite[Suppl.\ p.~19]{gerber2017toward} solved by
\begin{equation}
\label{eq:solution_wrt_gamma}
\Gamma_{kj} = \begin{cases}
1, & k=\arg\max_{k'}  \sum_{{i=1}}^{m}N_{ij}\log\lambda_{ik^{'}} \\
0, & \text{else.}
\end{cases}
\end{equation}
Possible non-uniqueness of the $\arg \max$ is resolved such that there is only one nonzero entry in every column of~$\Gamma$.

In particular, while $\Gamma \in \bR^{\lowrank\times n}$ was only assumed to be left stochastic in Problem~\ref{prob:DBMR}, it is solved by an affiliation matrix $\Gamma \in \{ 0,1 \}^{\lowrank\times n}$, describing \emph{binary} (or hard) affiliations $\Gamma_{kj} \in \{ 0,1 \}$ of the input states $j\in [n]$ to the latent states~$k\in [\lowrank]$ as in Definition~\ref{def:hard_affiliation_matrix}.
The fact that the partial updates \eqref{eq:lamda_hat} and \eqref{eq:solution_wrt_gamma} are available in closed form motivates the alternating procedure described in Algorithm \ref{alg:the_alg}, introducing the iteration counter $h$.  

\begin{algorithm}[htb]
\caption{DBMR algorithm from \cite{gerber2017toward}.}
\label{alg:the_alg}
\begin{algorithmic}[1]
    \State \textsc{Input}: Data subsumed into the count matrix $N$, number of latent states $\lowrank\in\mathbb{N}$, and maximum iteration number~$h_{\max}$
    \State Set random stochastic matrix $\Gamma^{(0)} \in \{0,1\}^{\lowrank\times n}$ and~$h=0$
    \State Set $\lambda^{(0)}$ by evaluating \eqref{eq:lamda_hat} for $\Gamma = \Gamma^{(0)}$
    \While{$\hat\ell(\lambda^{(h)},\Gamma^{(h)}) \neq \hat\ell(\lambda^{(h-1)},\Gamma^{(h-1)})$ and $h<h_{\max}$}
    \State Set $\Gamma^{(h+1)}$ by evaluating \eqref{eq:solution_wrt_gamma} for $\lambda = \lambda^{(h)}$
    \State \rev{Set $\lambda^{(h+1)}$} by evaluating \eqref{eq:lamda_hat} for $\Gamma = \Gamma^{(h+1)}$
    \State $h \leftarrow h+1$
    \EndWhile
    \State \textsc{Output:} $\lambda=\lambda^{(h)}$ and $\Gamma = \Gamma^{(h)}$
\end{algorithmic}
\end{algorithm}

We note that $\hat\ell(\lambda,\Gamma)$ is concave with respect to both $\lambda$ and $\Gamma$, individually. Thus, $\hat\ell(\lambda^{(h)},\Gamma^{(h)})$ increases in~$h$.
Since there are only finitely many values that $\Gamma$ can take, the algorithm converges, but possibly to a maximum of $\hat{\ell}$ that is only local (with respect to the updates \eqref{eq:lamda_hat} and \eqref{eq:solution_wrt_gamma}). A practical alternative stopping criterion is to stop if the relaxed likelihood $\hat\ell$ shows small improvements that are below a given threshold.
Since the algorithm might only find a locally optimal solution, it is usually run several times with independent random initializations $\Gamma^{(0)}$, and the result with the highest relaxed likelihood value is taken.

\begin{remark}
\label{rem:OrthNMF}
Note that~\eqref{eq:solution_wrt_gamma} implies that $\Gamma \Gamma^\top$ is a diagonal matrix: the rows of $\Gamma$ are orthogonal, but typically not orthonormal.
In Orthogonal NMF, as mentioned in section~\ref{section:Related_Work}, the orthogonality requirement would translate to~$\Gamma \Gamma^\top = \Id_K$. 
If $\Gamma$ has full rank, this can be achieved by the replacement $\Gamma \to D\Gamma$, using a diagonal scaling $D$; cf.~\cite[below equation~(10)]{ding2006orthogonal}. In this case we also have $\lambda \Gamma = \big(\lambda D^{-1}\big) (D\Gamma)$, where the factors in the parentheses fulfill the requirements of Orthogonal NMF.
DBMR hence yields a particular form of Orthogonal NMF.
\end{remark}

\paragraph{DBMR as maximum likelihood estimate on a constraint set.}
For fixed $m,n\in\bN$ and for $\lowrank \in \mathbb{N}$,  
$\lowrank \le \min\{m,n\}$, denote
\begin{equation}
\label{equ:DBMR_constraint_set}
\mathcal{D}^\lowrank_{\Lambda}
\coloneqq
\left\{
\Lambda = \lambda \Gamma \in \mathbb{R}^{m \times n} \mid
\lambda \in \bR_{\geq 0}^{m \times \lowrank}, \, \Gamma \in \{ 0,1 \}^{\lowrank \times n} \text{ both left stochastic}
\right\}.    
\end{equation}
The set $\smash{ \mathcal{D}^\lowrank_{\Lambda} }$ comprises the set of  transition matrices that we consider as low rank approximations (more precisely, as approximations of rank at most $\lowrank$) to the full-rank transition matrix~$P$. The salient feature (in addition to the low rank constraint) is the sparsity assumption $\Gamma \in \{ 0,1 \}^{\lowrank \times n}$ which, by the left stochasticity of $\Gamma$, leads to the interpretation of $\Gamma$ as a (hard) affiliation matrix, see Definition~\ref{def:hard_affiliation_matrix}.

If we restrict the maximum likelihood estimation from section~\ref{sec:Bayesian_model} to the set~$\mathcal{D}_{\Lambda}^\lowrank$,
\begin{equation}
\label{eq:loglikeproblem restricted}
 \Lambda^* = \underset{\Lambda \in \mathcal{D}_{\Lambda}^\lowrank}{\arg \max} \left\{ {\displaystyle \sum_{i=1}^{m}}{\displaystyle \sum_{j=1}^{n}}N_{ij}\log\Lambda_{ij}\right\},
\end{equation}
then we obtain an alternative characterization of the DBMR problem (Problem~\ref{prob:DBMR}):
\begin{lemma}
\label{lemma:dbmr constrained}
The functions $\ell$ and $\hat \ell$ coincide on $\smash{\mathcal{D}_{\Lambda}^\lowrank}$ (viewed as the set of admissible pairs~$(\lambda,\Gamma)$ in~\eqref{equ:DBMR_constraint_set}).
Further, every solution of \eqref{eq:loglikeproblem restricted} is a \rev{maximizer of} $\hat\ell$ in \eqref{eq:DBMR_objective}.
% The maximizers of $\hat\ell$ in \eqref{eq:DBMR_objective}, setting $\Lambda^* = \lambda^* \Gamma^*$, coincide with the solutions to \eqref{eq:loglikeproblem restricted}.
% In other words, the maximizers of $\ell$ and $\hat \ell$ coincide if we restrict the considerations to binary~$\Gamma$.
\end{lemma}
\begin{proof}
Note that the inequality in \eqref{eq:DBMR_objective} follows from Jensen's inequality, $\sum_{k=1}^{\lowrank} \Gamma_{kj}\log\lambda_{ik} \leq \log \sum_{k=1}^{\lowrank} \lambda_{ik}\Gamma_{kj}$, and is, in fact, an equality whenever $\Gamma$ is a \emph{binary} left stochastic matrix.
Hence, $\ell$ and $\smash{\hat{\ell}}$ coincide on $\smash{\mathcal{D}_{\Lambda}^\lowrank}$.
Further, there exists a maximizer $\Lambda^*= \lambda^* \Gamma^*$ of $\smash{\hat\ell}$ that lies in $\mathcal{D}_{\Lambda}^\lowrank$ \cite[Theorem]{gerber2017toward}.
Since $\ell$ and $\hat{\ell}$ coincide on $\smash{\mathcal{D}_{\Lambda}^\lowrank}$, this proves the second claim.
\end{proof}

It is well known that maximum likelihood estimation is inherently related to Kullback--Leibler minimization \cite[Section 9.5]{wasserman2004all}.
The following remark establishes this connection in the specific context of DBMR:
\begin{remark}[Connection to Kullback--Leibler divergences]
\label{remark:Connection_ML_KL}
As mentioned in section~\ref{section:Coherent_sets}, the joint distribution of the data $(X,Y)$ is described by the probability values $P_{ij} p_{j} = \bP \left[ Y=i, X=j\right]$ (which are proportional to $N_{ij}$). We can likewise describe the joint distribution of the DBMR output by $\Lambda_{ij} p_{j}$, where $\Lambda \in \mathcal{D}_{\Lambda}^\lowrank$.
The DBMR objective \eqref{eq:loglikeproblem restricted} can be rewritten in the form
\begin{equation}
 \Lambda^* = \underset{\Lambda \in \mathcal{D}_{\Lambda}^\lowrank}{\arg \min} \, \KLD{\left\{P_{ij} p_{j} \right\}}{ \left\{\Lambda_{ij} p_{j}\right\}},
\end{equation}
where, with slight abuse of notation, we naturally extend the definition of $\mathrm{D}_{\mathrm{KL}}$ from section~\ref{section:Setup_and_Notation} to matrices associated with joint distributions.
In other words, DBMR attempts to match the true joint distribution in the Kullback--Leibler sense, while obeying the rank and sparsity constraints imposed by~$\mathcal{D}_{\Lambda}^\lowrank$. Indeed, 
\begin{equation}
\label{eq:KL P Lambda}
\KLD{ \left\{P_{ij} p_{j} \right\}}{ \left\{\Lambda_{ij} p_{j}\right\} } = \sum_{i,j} \log\left( \frac{P_{ij} p_{j}} {\Lambda_{ij} p_{j}}\right) P_{ij} p_{j} \, \propto \, \sum_{i,j} N_{ij} \left( \log P_{ij} - \log \Lambda_{ij}\right).
\end{equation}
Clearly, minimizing \eqref{eq:KL P Lambda} is equivalent to maximizing \eqref{eq:loglikeproblem restricted}, since $\sum_{i,j} N_{ij} \log P_{ij}$ does not depend on~$\Lambda$.
\end{remark}

\begin{remark}
\label{remark:No_low_rank_structure_identification}
In the case when the count matrix $N$ is of low effective dimensionality, that is, when $\mathrm{rank} \, N = \widetilde{\lowrank} < \lowrank$, a natural question is whether the DBMR output $\Lambda^{\ast} \in \smash{ \mathcal{D}_{\Lambda}^{\lowrank} }$ lies in the smaller set $\smash{ \mathcal{D}_{\Lambda}^{\widetilde{\lowrank}} }$, i.e., whether DBMR automatically identifies the low rank structure of~$N$.
In general, this is not the case:
Consider two linearly independent vectors $a,b \in \bR^{m}$, $m\geq 3$, and $\smash{ P = [a, \tfrac{a}{2} + \tfrac{b}{2}, b] \in \mathbb{R}^{m\times 3} }$ (recall that $N$ is simply a scaled version of $P$, implying here that $\smash{ \mathrm{rank} \, N = \mathrm{rank} \, P =  \widetilde{\lowrank} = 2 }$).
For $\lowrank=3$, the best approximation of $P$ within $\smash{ \mathcal{D}_{\Lambda}^{\lowrank} }$ is clearly $P$ itself, by choosing $\lambda^{\ast} = P$ and $\Gamma^{\ast} = \Id_{3}$.
This solution is unique up to permutations of the columns of $\lambda^{\ast}$ and rows of $\Gamma^{\ast}$, respectively, since, for any affiliation matrix $\Gamma \in \bR^{3\times 3}$ with $\mathrm{rank}\, \Gamma < 3$ and any $\lambda \in \bR^{m\times 3}$, the product $\lambda\Gamma$ has at least two identical columns and cannot coincide with~$P$.
Clearly, the above solution $\Lambda^{\ast} = \lambda^{\ast}\Gamma^{\ast}$ maximizes the likelihood in \eqref{eq:new_prob} and coincides with the DMBR output, at least for appropriate initializations (e.g.\ $\Gamma^{(0)} = \Id_{3}$).
\end{remark}

% \begin{remark}[Choice of $\lowrank$]
% \label{remark:AIC}
% Statistical approaches towards model selection may be applied to the problem of choosing an appropriate latent dimensionality $r$, see, for instance, \cite[Section 5.2]{gillis2020nonnegative}.
% In the context of DBMR, Akaike's information criterion (AIC) \cite{akaike1974new} as an estimator of the relative expectation of Kullback--Leibler distance based on maximized log-likelihood has been suggested by \cite{gerber2017toward}. For the this setting, the AIC for the reduced model is defined by
% \begin{equation}
% \label{eq:Akaike_information_criterion}
%     \mathrm{AIC}(\lambda,\Gamma):=2 \lowrank (n+m)-2\hat\ell(\lambda,\Gamma),
% \end{equation}
% where $\lowrank(n+m)$ is the number of estimated parameters in $\lambda$
% and $\Gamma$ and $\hat\ell$ is the DBMR objective defined by~\eqref{eq:DBMR_objective}. The term $2 \lowrank (n+m)$ in \eqref{eq:Akaike_information_criterion} is a bias-correction term depending on the number of latent states $\lowrank$. The optimal integer value of $\lowrank$ can be obtained by performing the DBMR algorithm with different numbers of $\lowrank$ (i.e., $\lowrank= 1,2,3,...$) and then selecting the reduced model with the minimal AIC value in~\eqref{eq:Akaike_information_criterion}.
% \end{remark}

\section{Relations between the full and the reduced model}
\label{sec:PTO}

\subsection{DBMR as a projection}

In order to analyze the approximation properties of the reduced model $\Lambda = \lambda \Gamma$ provided by DBMR, we will consider the (hard) affiliation matrix $\Gamma$ fixed throughout this section (assume e.g.\ that it has already been computed).
Then, in view of \eqref{eq:lamda_hat}, DBMR constructs the low rank approximation $\Lambda$ of $P$ as follows:
For each $k=1,\dots,\lowrank$,
\begin{itemize}
\item 
compute the column $\lambda_{\bullet k}$ as the weighted average of the columns of $P$ associated with $k$:
$ \lambda_{\bullet k} \, \propto\, \sum_{j=1}^{n} \Gamma_{kj} p_{j} P_{\bullet j} $,
\item 
replace the columns of $P$ associated with $k$ by this average:
%\newline
$\Lambda_{\bullet j} = \lambda_{\bullet k}$ if $\Gamma_{kj} = 1$.
\end{itemize}
The crucial observation in this section is that this process can be rewritten as a composition of $P$ with a projection $\Pi$, 
\begin{equation}
\label{eq:DMBR decomp}
\Lambda = \lambda \Gamma = P \Pi,
\end{equation}
as illustrated by the following example:
\begin{example}
\label{example:DBMR_as_projection}
Let $n = 5$ and $\lowrank=2$ and assume that we have already computed
\[ \Gamma=\left(\begin{array}{ccccc}
\cellcolor{yellow!20} 1 & \cellcolor{yellow!20} 1 & 0 & 0 & 0\\
0 & 0 & \cellcolor{green!10} 1 & \cellcolor{green!10} 1 & \cellcolor{green!10} 1
\end{array}\right). \]
Then the observation described above concerning the $\lambda$-update-step~\eqref{eq:lamda_hat}
 within DBMR can be illustrated as follows:
\begin{align*}
&\xrightarrow[\text{according to $\Gamma$}]{\text{group columns of $P$}}&
P
&=
\left(\begin{array}{>{\columncolor{yellow!20}}c>{\columncolor{yellow!20}}c>{\columncolor{green!10}}c>{\columncolor{green!10}}c>{\columncolor{green!10}}c}
	P_{\bullet 1} & P_{\bullet 2} & P_{\bullet 3} & P_{\bullet 4} & P_{\bullet 5}
\end{array}\right)
\\
&\xrightarrow[\text{columns}]{\text{average}}&
\lambda
&=
\left(\begin{array}{>{\columncolor{yellow!20}}c>{\columncolor{green!10}}c}
	\lambda_{\bullet 1}
    =\tfrac{p_{1}P_{\bullet 1}+p_{2}P_{\bullet 2}}{p_{1}+p_{2}}
    &
    \lambda_{\bullet 2}
    =\tfrac{p_{3}P_{\bullet 3}+p_{4}P_{\bullet 4}+p_{5}P_{\bullet 5}}{p_{3}+p_{4}+p_{5}}
\end{array}\right)
\\
&\xrightarrow[]{\text{augment}}&
\Lambda
= \lambda \Gamma
&=
\left(\begin{array}{>{\columncolor{yellow!20}}c>{\columncolor{yellow!20}}c>{\columncolor{green!10}}c>{\columncolor{green!10}}c>{\columncolor{green!10}}c}
	\lambda_{\bullet 1} &\lambda_{\bullet 1}  & \lambda_{\bullet 2} & \lambda_{\bullet 2} & \lambda_{\bullet 2}
\end{array}\right).
\end{align*}
Clearly, this procedure of obtaining $\Lambda = \lambda \Gamma$ can be written as a matrix product $\Lambda = P \Pi$ for
\[
\Pi
=
\begin{pmatrix}
	\cellcolor{yellow!20} \frac{p_{1}}{p_{1}+p_{2}} & \cellcolor{yellow!20} \frac{p_{1}}{p_{1}+p_{2}} & 0 & 0 & 0\\[1ex]
	\cellcolor{yellow!20} \frac{p_{2}}{p_{1}+p_{2}} & \cellcolor{yellow!20} \frac{p_{2}}{p_{1}+p_{2}} & 0 & 0 & 0\\[1ex]
	0 & 0 & \cellcolor{green!10} \frac{p_{3}}{p_{3}+p_{4}+p_{5}} & \cellcolor{green!10} \frac{p_{3}}{p_{3}+p_{4}+p_{5}} &  \cellcolor{green!10} \frac{p_{3}}{p_{3}+p_{4}+p_{5}}\\[1ex]
	0 & 0 & \cellcolor{green!10} \frac{p_{4}}{p_{3}+p_{4}+p_{5}} & \cellcolor{green!10} \frac{p_{4}}{p_{3}+p_{4}+p_{5}} &  \cellcolor{green!10} \frac{p_{4}}{p_{3}+p_{4}+p_{5}}\\[1ex]
	0 & 0 & \cellcolor{green!10} \frac{p_{5}}{p_{3}+p_{4}+p_{5}} & \cellcolor{green!10} \frac{p_{5}}{p_{3}+p_{4}+p_{5}} &  \cellcolor{green!10} \frac{p_{5}}{p_{3}+p_{4}+p_{5}}
\end{pmatrix}.
\]
This motivates the definition of $\Pi$ in \eqref{equ:definition_Pi_and_tpi} below.
\end{example}

\begin{remark}
The procedure and example above explain why DBMR provides ``good'' low rank approximations $\Lambda$ of $P$ and why it identifies coherent pairs:
$P$ is a maximum likelihood estimate and $\Lambda$ maximizes the same likelihood, but within the set $\mathcal{D}_{\Lambda}^\lowrank$ of low rank matrices, cf.~\eqref{eq:loglikeproblem restricted}.
Hence, indirectly, $\Lambda$ results to be as ``close'' to $P$ as possible.
In view of the discussion above, the best way to be ``close'' to $P$ is to group its columns by \emph{similarity} (and this clustering then defines $\Gamma$).
This way, each column in $P$ does not differ too much from the corresponding average column in $\Lambda$.
As a consequence, in terms of coherence, the states within each group (i.e.\ partition element) evolve with similar probability vectors, hence ``coherently''.
\end{remark}

After establishing equation~\eqref{eq:DMBR decomp} in Theorem~\ref{thm:Frob_DBMR_bound} below, together with several properties of $\Pi$, we leverage this result to draw conclusions about the relationship between the full model $P$ and the low rank model~$\Lambda = \lambda\Gamma$, in particular, in the context of coherence.
More precisely,	we show that the degree of coherence (measured by the sum of leading singular values of the corresponding matrix as in section~\ref{sec:coh_svd}) associated to the DBMR approximation is bounded from above by the one of the full model (cf.\ Proposition~\ref{prop:sval_comp}), $\cC(\Lambda)\leq \cC(P)$.

In what follows, we work with the full and low rank transition matrices $P$ and $\Lambda$ as well as with their rescaled versions $\smash{ \tP = D_q^{-1/2} P D_p^{1/2} \in \bR^{m\times n},\,	\tLam = D_q^{-1/2} \Lambda D_p^{1/2} }$, see \eqref{eq:P_tilde} and section~\ref{sec:coh_svd}. To identify the projection $\Pi$ in \eqref{eq:DMBR decomp}, recall the assignment $\gamma:[n] \rightarrow [\lowrank]$ associated to a fixed (hard) affiliation matrix $\Gamma \in \mathbb{R}^{\lowrank \times n}$ from Definition~\ref{def:hard_affiliation_matrix}. 
Motivated by Example~\ref{example:DBMR_as_projection}, we define $\Pi \in \bR_{\geq 0}^{n\times n}$ and $\smash{ \tpi = D_p^{-1/2} \Pi D_p^{1/2} }$ by
\begin{equation}
\label{equ:definition_Pi_and_tpi}
\Pi_{ij}
=
\frac{p_i \, \delta_{\gamma(i) \gamma(j)}}{\sum_{l} p_l \delta_{\gamma(l) \gamma(i)}}, 
\qquad
\tpi_{ij} = \frac{\sqrt{p_i p_j} \, \delta_{\gamma(i)\gamma(j)}}{\sum_l p_l \delta_{\gamma(l)\gamma(i)}},
\end{equation}
noting that here and in subsequent sections, $\widetilde{\Pi}$ serves as an auxiliary object that facilitates computations. Indeed, it is straightforward to verify that $\tpi$ is an orthogonal projection (cf.\ Lemma~\ref{lemma:DBMR_projection_tpi_properties} in Appendix~\ref{app:proof_projection_theorem}). Consequently, $\Pi$ is a $p^{-1}$-orthogonal projection, meaning that $\Pi^2 =  \Pi$ and that $\Pi$ is $p^{-1}$-symmetric:
\begin{equation}
\label{eq:p symmetric}
\langle \Pi x,y\rangle_{p^{-1}} = \langle x,\Pi y\rangle_{p^{-1}}, \qquad \text{for all} \quad x,y \in \mathbb{R}^n.   
\end{equation}
The following result confirms the relation \eqref{eq:DMBR decomp} and clarifies the structure of $\Pi$ in terms of its eigenvector decomposition. \rev{An important role is played by the set 
\begin{equation}
\label{eq:range gamma}
\mathrm{Ran} \gamma := \gamma([n])  \subset [\lowrank],
\end{equation}
 which we refer to as the set of \emph{active} latent states (recall Definition \ref{def:hard_affiliation_matrix}). In \eqref{eq:range gamma}, the notation $\gamma([n])$ refers to the image set (or range) of $\gamma$, that is, $k \in \gamma([n])$ if and only if there exists $i \in [n]$ such that $\gamma(i) = k$.}

\begin{theorem}[DBMR as a projection]
\label{thm:projection}
Let $\Gamma \in \{0,1\}^{\lowrank\times n}$ be a hard affiliation matrix (according to Definition \ref{def:hard_affiliation_matrix}) and $\lambda$ be given by~\eqref{eq:lamda_hat}.
    Then $\Pi$ as defined in \eqref{equ:definition_Pi_and_tpi} is a left~stochastic and $p^{-1}$-orthogonal projection which satisfies $\lambda \Gamma = P \Pi$.
    Moreover, $\Pi$ has the following properties:
	\begin{enumerate}[(a)]
\item
\label{item:projection_thm_rank_bound}
The rank of $\Pi$  coincides with the number of active latent states, $\mathrm{rank} \,  \Pi = \# \mathrm{Ran} \gamma$. In particular, ~$\mathrm{rank}\,  \Pi \leq \lowrank$.
\item
\label{item:projection_thm_eigenvectors}
The vectors $a^{(k)} \in \mathbb{R}^n$, associated to active latent states $k \in \mathrm{Ran}\gamma \subseteq [\lowrank]$ and defined by
\begin{equation}
\label{eq:ak}
a_i^{(k)} := p_i \delta_{\gamma(i)k}\,, \qquad i =1, \ldots,n,
\end{equation}
are eigenvectors of $\Pi$ with eigenvalue $1$, i.e.,~$\Pi a^{(k)} = a^{(k)}$. They span the image space of $\Pi$, that is, 
\begin{equation}
\label{eq:eigenspace}
\mathrm{Span} \{a^{k}: \,\, k \in \mathrm{Ran} \, \gamma\} = \mathrm{Ran} \, \Pi.
\end{equation}

The supports of the vectors $a^{(k)}$ are disjoint, that is, $a_i^{(k)} a_i^{(l)} = 0$, for all $i$ and all $k \neq l $.
In particular, these vectors are orthogonal as well as $p$-orthogonal.

Hence, the vector $\smash{ p = \sum_{k \in \mathrm{Ran}\gamma} a^{(k)} }$ is also an eigenvector of $\Pi$ with eigenvalue $1$.
\item
\label{item:projection_thm_DBMR_structure_preserving}
The DBMR output $\Lambda = \lambda \Gamma$ respects the final distribution in the sense that~$\Lambda p = q$. 
\end{enumerate}
\end{theorem}
\begin{proof}
The proof is given in Appendix~\ref{app:proof_projection_theorem}.
\end{proof}

\begin{remark}
\label{rem:comments_on_projection_theorem}
Let us comment on the previous result.
\begin{enumerate}[(a)]
    \item Note that the assumptions in Theorem \ref{thm:projection} on  $\Gamma$ and $\lambda$ are satisfied at any (completed) iteration of the DBMR algorithm; convergence is not required.
    \item For a fixed (active) latent state $k \in [\lowrank]$, it is natural to consider the corresponding set $\gamma^{-1}(\{k\})$ comprised of the states in $[n]$ that are mapped to $k$. The disjoint union $\bigcup_{k} \gamma^{-1}(\{k\}) = [n]$ can then be viewed as a coarse-graining of the input set $[n]$ induced by the assignment $\gamma$. Theorem \ref{thm:projection} shows that this interpretation persists at the level of the projection $\Pi$, and that its range encodes the same information. Indeed, the eigenvectors in  \eqref{eq:ak} can be obtained as the restriction of $p$ to $\gamma^{-1}(\{k\})$,   
    \begin{equation}
    a_i^{(k)} = \begin{cases}
    p_i \quad \text{if } i \in  \gamma^{-1}(\{k\}),  \\
    0  \quad \text{otherwise.}
    \end{cases}   
    \end{equation}
\end{enumerate}    
\end{remark}

\begin{corollary}
\label{cor:DBMR_as_best_approximation}
Let $\Gamma \in \{0,1\}^{\lowrank\times n}$ be a hard affiliation matrix (according to Definition \ref{def:hard_affiliation_matrix}), $\lambda$ be given by~\eqref{eq:lamda_hat} and $\Lambda = \lambda \Gamma$.
Further, let $\tP$ and $\tLam$ be given by \eqref{eq:P_tilde} and $\tpi$ by \eqref{equ:definition_Pi_and_tpi}.
Then $\tP-\tLam$ is orthogonal to any matrix of the form $A\tpi$, $A\in\bR^{m\times n}$ with respect to the Frobenius norm.
In particular, since $\tLam = \tP\tpi$,
\begin{equation}
\label{eq:Pythagoras}
\| \tP-\tLam\| _{F}^{2}=\| \tP\| _{F}^{2} - \| \tLam\| _{F}^{2}.
\end{equation}
Further, for a fixed (hard) affiliation matrix $\Gamma\in \{0,1\}^{\lowrank\times n}$ of rank $r$, the choice of $\lambda$ in~\eqref{eq:lamda_hat} results in a matrix $\tLam$ that is the best approximation of $\tP$ in Frobenius norm:
\begin{equation}
    \label{eq:bestapprox}
    \tLam = D_q^{-1/2} \lambda \Gamma D_p^{1/2} \in \argmax_{\lambda' \in \bR^{m \times \lowrank}}\, \big\| \tP - D_q^{-1/2} \lambda' \Gamma D_p^{1/2} \big\|_{F}.
\end{equation}
\end{corollary}
\begin{proof}
Note that for real matrices $A$ and $B$ of the same dimension, the Frobenius norm is induced by the inner product $\langle A, B\rangle_F := \operatorname{tr} (A^\top B)$, where $\operatorname{tr}(\cdot)$ denotes the trace.
By \eqref{eq:P_tilde} and \eqref{equ:definition_Pi_and_tpi} the identity $\Lambda = P \Pi$ from Theorem~\ref{thm:projection} implies $\tLam = \tP\tpi$.
Therefore, using the properties~$\tpi^{\top} = \tpi$ and $\smash{ \tpi^2 = \tpi }$ from Lemma~\ref{lemma:DBMR_projection_tpi_properties} as well as invariance of the trace under cyclic permutation of factors, we obtain, for each $A\in\bR^{m\times n}$
\begin{equation}
    \label{eq:FrobOrthogonality}
    \langle \tP-\tLam, A\tpi \rangle_F = \operatorname{tr} \left ( \big(\tP-\tLam\big)^\top A\tpi \right) = \operatorname{tr} \left( \tpi \big(\Id-\tpi\big) \tP^\top A \right) = 0.
\end{equation}
This readily implies~\eqref{eq:Pythagoras}.

For the last statement, let us define the diagonal matrices $D_N = \diag (\mathds{1}^{\top}N)^{-1}$, $D_1 = \diag (\mathds{1}^{\top}N\Gamma)^{-1}$, and~$D_2 = \diag( \mathds{1}^{\top}N \Gamma^{\top}\Gamma)^{-1}$. It is a straightforward calculation to see that $P = ND_N$, $D_1\Gamma = \Gamma D_2$, and~$\Pi = D_N^{-1}\Gamma^{\top}\Gamma D_2$. Since $\Gamma$ is assumed to have rank~$r$, $\Gamma\Gamma^{\top}$ is invertible, yielding that for an arbitrary~$\lambda' \in \mathbb{R}^{m\times r}$ we have
\begin{align*}
    \lambda'\Gamma &= \lambda' D_1^{-1} D_1\Gamma = \lambda' D_1^{-1} \Gamma D_2
    = \lambda' D_1^{-1} (\Gamma\Gamma^\top)^{-1} \Gamma \Gamma^\top \Gamma D_2 \\
    &= \lambda' D_1^{-1} (\Gamma\Gamma^\top)^{-1} \Gamma D_N \ D_N^{-1}\Gamma^\top \Gamma D_2 =: A' \Pi
\end{align*}
This means that $\lambda' \Gamma$ is of the form $A' \Pi$ for a suitable $A' \in \mathbb{R}^{m\times n}$, and~\eqref{eq:bestapprox} follows from~\eqref{eq:FrobOrthogonality}.
\end{proof}
Hence, \eqref{eq:lamda_hat} is optimal in the sense of~\eqref{eq:bestapprox} in addition to being optimal in terms of the (relaxed) maximum likelihood from section~\ref{ssec:DBMR}.

Corollary~\ref{cor:DBMR_as_best_approximation} will be used in Corollary~\ref{cor:relating_objectives} below to relate the objectives degree of coherence $\cC_{\lowrank}$ and relaxed likelihood~$\hat{\ell}$.

% \ik{
% Corollary~\ref{cor:DBMR_as_best_approximation} shows that, for a fixed (hard) affiliation matrix $\Gamma\in \{0,1\}^{\lowrank\times n}$, the choice of $\lambda$ in~\eqref{eq:lamda_hat} results in a matrix $\tLam$ that is the best approximation of $\tP$ in the Frobenius norm:
% \[
% \tLam
% =
% D_q^{-1/2} \lambda \Gamma D_p^{1/2}
% \in
% \argmax_{\lambda' \in \bR^{m \times \lowrank}} \norm{\tP - D_q^{-1/2} \lambda' \Gamma D_p^{1/2} }_{F}.
% \]

\subsection{Pointwise singular value bounds}

Recall that $\sigma_i(M)$, $i \ge 1$, denotes the $i$-th singular value of a matrix $M$, in descending order.
By Theorem~\ref{thm:projection}\eqref{item:projection_thm_eigenvectors}, $\Pi p = p$ and hence $p^{1/2}$ is a right eigenvector of $\tpi$ with eigenvalue~$1$:
\[
\tpi p^{1/2} = D_p^{-1/2} \Pi D_p^{1/2} p^{1/2} = D_p^{-1/2} \Pi p = D_p^{-1/2} p = p^{1/2}.
\]
As discussed in section~\ref{sec:coh_svd}, we have that~$\sigma_1(\tP)=1$ with corresponding right singular vector~$p^{1/2}$. Thus, we also have that~$\sigma_1(\tP\tpi)=1$ with corresponding right singular vector~$p^{1/2}$, since $\tpi$ is an orthogonal projection and hence $\tP\tpi$ cannot have singular values larger than~$1$. Thus, $\tP$ and $\tP\tpi$ share the same leading singular value with the same left and right singular vector pair.
As for the comparison of the other singular values, the following holds true.

\begin{proposition}
\label{prop:sval_comp}
With $\tP$ defined in \eqref{eq:P_tilde} and $\tpi$ defined in \eqref{equ:definition_Pi_and_tpi}, we have that
\begin{equation}
    \label{eq:sval_ineq}
    \sigma_i(\tP\tpi) \le  \sigma_i(\tP),\quad i\le \min\{m,n\}.
\end{equation}
\end{proposition}

\begin{proof}
The claim follows from \cite[(2.3) on pp.~27]{gohberg1978introduction} by noting that $\|\tpi\|_2=1$. For the readers' convenience, we give a proof based on the Courant--Fischer theorem:
Let us fix $i\le \min\{m,n\}$ such that $\sigma_i(\tP\tpi) > 0$ (otherwise, inequality~\eqref{eq:sval_ineq} holds trivially), in particular, $i\le \mathrm{rank}(\tpi)$.
Let $F_i$ denote the subspace spanned by the first $i$ right singular vectors of $\tP\tpi$.
Since the associated singular values are all nonzero, $\mathrm{ker}\,\tpi \cap F_i = \{0\}$ and thereby
\begin{equation}
	\label{eq:dim_Fi}
	\dim F_i
	=
	\dim \{ \tpi f \mid f \in F_i \}
	\eqqcolon
	\dim \tpi F_i.
\end{equation}
Since $\tpi$ is an orthogonal projection, $\|\tpi h\|_2 \le \|h\|_2$ for each $h \in \bR^{n}$, and Theorem \ref{thm:Courant_Fischer_singular_values_version} (Courant--Fischer) implies
\[
\sigma_i(\tP\tpi)
=
\min_{\substack{h\in F_i\\ \|h\|_2=1}} \norm{\tP\tpi h}_{2}
\leq 
\min_{\substack{\tilde{h}\in \tpi F_i\\ \|\tilde{h}\|_2=1}} \norm{\tP \tilde{h}}_{2}
\leq
\max_{W \in \fW_{i}}  \min_{\substack{\tilde{h}\in W\\ \|\tilde{h}\|_2=1}}  \norm{\tP \tilde{h}}_{2}
=
\sigma_i(\tP),
\]		
where $\fW_{i}$ denotes the set of $i$-dimensional subspaces of $\bR^{n}$.
This proves the claim.
\end{proof}

\section{Relations between the Frobenius norm and the relaxed likelihood objectives}
\label{sec:relation_FrobKL}

Recall that the degree of coherence $\cC_{r}$ of the full matrix $P$ and the reduced matrix $\Lambda$ is defined via the singular values of the scaled transition matrices $\tP = \linebreak[3] \smash{ D_q^{-1/2} P D_p^{1/2} }$ and $\smash{ \tLam = D_q^{-1/2} \Lambda D_p^{1/2} }$, respectively, and that $\sigma_i(\tLam) \leq  \sigma_i(\tP)$ for each $i\le \min\{m,n\}$ by Proposition~\ref{prop:sval_comp}.

Noting that the squared Frobenius norm $\|\cdot\|_F^{2}$ of a matrix equals the sum of its squared singular values, it is therefore natural to measure the discrepancy between full and reduced models by the corresponding difference in Frobenius norm.
Theorem~\ref{thm:Frob_DBMR_bound} below relates $\| \tP - \tLam \|_{F}^{2}$ to the relaxed likelihood $\hat{\ell}$ from~\eqref{eq:DBMR_objective} that is maximized by DBMR, indicating that DBMR provides a quasi-optimal solution of the (relaxed) coherence problem~\eqref{eq:coherence_relax1}.

\begin{remark}
    In the context of Nonlinear Matrix Factorization, \cite[Equations (8)--(10)]{ding2006nonnegative} show that, assuming small errors and linearizing the objective around the optimum, a maximum-likelihood estimation of nonnegative factor matrices can be connected to $\chi^2$ statistics.
    This leads them to the minimization of the Frobenius-norm difference of an empirical frequency matrix and its factorized approximation as well as to a connection to the maximum likelihood setting. They do not elaborate this any further, eventually.
\end{remark}

More generally, Theorem~\ref{thm:Frob_DBMR_bound} establishes an a posteriori bound between the two main NMF objectives discussed in section~\ref{section:Intro}, namely the Frobenius norm $\| A - BC \|_F$ and the (generalized\footnote{The Kullback--Leibler divergence $\KLD{A}{BC}$ is generalized in the sense that $A$ as well as $BC$ represent \emph{unnormalized} probability distributions.}) Kullback--Leibler divergence~$\KLD{A}{BC}$.
For this purpose, we require Pinsker-like inequalities for the weighted $\ell^{2}$ norm in Appendix \ref{section:Pinsker_l_2}.
These are based on the concept of \emph{balancedness} of a vector $x\in\bR^{m}$ that we introduce in the following.
Roughly speaking, we call a vector $x\in\bR^{m}$ balanced if $\norm{x}_{\infty}/\norm{x}_{1} \ll 1$.
Note that the inequality $\norm{x}_{\infty} \leq \norm{x}_{1}$ holds true in general and equality only holds for (multiples of) standard unit vectors. % $x = e_{k}$, $k=1,\dots,m$.
On the other hand, the above ratio is minimal if all entries of $x$ have the same modulus, i.e.\ $x_{i} = \pm \norm{x}_{\infty}$ for each $i = 1,\dots,m$.
In other words, for $x$ to be balanced, the ``mass'' of the vector (measured by $\norm{x}_{1}$) should not be attributed to one or just a few entries, with the others being zero or close to zero, but should be distributed rather evenly among the entries.
More generally, $q$-balancedness of $x$ indicates that the ratio $|x_{i}|/q_{i}$ is close to constant in $i$, with $q\in\bR_{>0}^{m}$ being a strictly positive probability vector:

\begin{definition}
\label{def:balancedness}
For $m\in\bN$ and a strictly positive probability vector $q\in\bR_{>0}^{m}$, we define the \emph{balancedness} and the \emph{$q$-balancedness} of a vector $x\in\bR^{m}$ by
\[
\fB(x)
\coloneqq
\frac{\norm{x}_{1}}{m \norm{x}_{\infty}}
\in
[\tfrac{1}{m}, 1],
\qquad
\fB_{q}(x)
\coloneqq
\frac{\norm{x}_{1}}{\max_{i} \frac{|x_{i}|}{q_{i}}}
\in
[\min_{i} q_{i}, 1],
\qquad
\text{if }
x\neq 0,
\]
and by $\fB(0) \coloneqq \fB_{q}(0) \coloneqq 1$.
\end{definition}

\begin{remark}
Note that $\fB(x) = \fB_{q}(x)$ for the vector $q = (1/m,\dots,1/m)$.
\end{remark}

\begin{theorem}
\label{thm:Frob_DBMR_bound}
Let $\Lambda\in \bR^{m\times n}$, $\lambda \in \bR^{m\times \lowrank}$ and $\Gamma \in \{0,1\}^{\lowrank\times n}$ be left stochastic matrices such that $\Lambda = \lambda \Gamma$ and let $\tP$ and $\tLam$ be given by \eqref{eq:P_tilde}.
% Let $\alpha_{j} \coloneqq \alpha(P_{\bullet j},\Lambda_{\bullet j}) \coloneqq \frac{2}{3} \max_{i} \frac{|P_{ij}-\Lambda_{ij}|}{P_{ij}} < 1$
% $\alpha(u,v) \coloneqq \frac{2}{3} \max_{i} \frac{|u_{i}-v_{i}|}{u_{i}} < 1$.
Further, let $\alpha_{j} \coloneqq \alpha(P_{\bullet j},\Lambda_{\bullet j}) \coloneqq \smash{\frac{2}{3} \max_{i} \frac{|P_{ij}-\Lambda_{ij}|}{P_{ij}}} \in [0,\infty]$, $j=1,\dots,n$,
%%%%%%%%%%%%%%%%%%%%%%%%%%%%%%%%
\begin{align*}
% \kappa^{(1)}
% &\coloneqq
% \tfrac{m}{2}
% \min_{i=1,\dots,m} q_{i} \, 
% \min_{j=1,\dots,n} \fB(P_{\bullet j} - \Lambda_{\bullet j}),
% \\
\kappa_{q}^{(1)}
&\coloneqq
\tfrac{1}{2}
\min_{j=1,\dots,n} \fB_{q}(P_{\bullet j} - \Lambda_{\bullet j}),
\\
% \kappa^{(2)}
% &\coloneqq
% \tfrac{m}{2}
% \min_{i=1,\dots,m} q_{i} \, 
% \min_{j=1,\dots,n} \fB(P_{\bullet j}) (1 - \alpha_{j}),
% \\
\kappa_{q}^{(2)}
&\coloneqq
\tfrac{1}{2}
\min_{j=1,\dots,n} \fB_{q}(P_{\bullet j}) (1 - \alpha_{j}),
\end{align*}
$\kappa^{\mathrm{pr}} \coloneqq \min_{i} q_{i} /2$ and $\kappa^{\mathrm{post}} \coloneqq \max(\kappa_{q}^{(1)},\kappa_{q}^{(2)})$.
Then, $\kappa^{\mathrm{post}} \geq \kappa^{\mathrm{pr}}$ and, for $\kappa \in \{ \kappa^{\mathrm{pr}} , \kappa^{\mathrm{post}} \}$,
\begin{equation}
\label{eq:Frob_DBMR_bound}
\| \tP - \tLam \|_{F}^{2}
\ \leq\ 
\kappa^{-1} \sum_{j=1}^n p_j \KLD{P_{\bullet j}}{\Lambda_{\bullet j}}
\ =\ 
\frac{1}{\kappa S}
\left( \hat{\ell}(P,\Id_n) -
\hat{\ell}(\lambda,\Gamma) \right),
\end{equation}
where $\hat{\ell}$ is the DBMR objective given by~\eqref{eq:DBMR_objective}. 
\end{theorem}

\begin{proof}%[Proof of Theorem \ref{thm:Frob_DBMR_bound}]
Note that $\kappa^{\mathrm{post}} \geq \kappa^{\mathrm{pr}}$ by the definition of $\kappa_{q}^{(1)}$ and $\fB_{q}$, so it suffices to to prove \eqref{eq:Frob_DBMR_bound} for $\kappa = \kappa^{\mathrm{post}}$.
Observe that $P_{ij} p_{j} = N_{ij}/S$ for any $i=1,\dots,m$ and $j=1,\dots,n$, and that $\ell(\lambda,\Gamma) = \hat{\ell}(\lambda,\Gamma)$, since $\Gamma$ is an affiliation matrix, cf.\ \eqref{eq:new_prob}--\eqref{eq:DBMR_objective}.
Hence, Proposition~\ref{prop:l2_Pinsker} \rev{in Appendix~\ref{section:Pinsker_l_2}} implies
(note that we do not need to verify the condition $\alpha_{j} < 1$ for each $j=1,\dots,n$, since, in this case, $\smash{ \kappa_{q}^{(2)} < 0 }$ and $\kappa = \smash{ \kappa_{q}^{(1)} } $)
\begin{align*}
\| \tP - \tLam \|_F^2
&=
\sum_{i=1}^m \sum_{j=1}^n \frac{1}{q_i} (P_{ij} - \Lambda_{ij} )^2 p_j
\\
&\leq
\min\bigg(
\sum_{j=1}^n p_j \sum_{i=1}^m \frac{(P_{ij} - \Lambda_{ij} )^2}{q_i}
\ , \ 
\frac{1}{\min_{i} q_i} \sum_{j=1}^n p_j \| P_{\bullet j} - \Lambda_{\bullet j} \|_2^2
\bigg)
\\
&\leq
\kappa^{-1} 
\sum_{j=1}^n p_j \KLD{P_{\bullet j}}{\Lambda_{\bullet j}}
\\
&=
\kappa^{-1} \sum_{j=1}^n p_j \sum_{i=1}^{m} P_{ij} \log \frac{P_{ij}}{\Lambda_{ij}}
\\
&=
(\kappa S)^{-1} \left( \sum_{j=1}^n \sum_{i=1}^{m} N_{ij} \log P_{ij} -
\sum_{i=1}^{m} \sum_{j=1}^n N_{ij} \log \sum_{k=1}^{\lowrank} \lambda_{ik} \Gamma_{kj} \right)
\\
&=
(\kappa S)^{-1} \left( \sum_{j=1}^n \sum_{i=1}^{m} N_{ij} \sum_{k=1}^n \delta_{kj} \log P_{ik} -
\sum_{i=1}^{m} \sum_{j=1}^n N_{ij} \sum_{k=1}^{\lowrank} \Gamma_{kj}\log \lambda_{ik} \right)
\\
&=
(\kappa S)^{-1} 
\bigg( \hat{\ell}(P,\Id_n) - \hat{\ell}(\lambda,\Gamma) \bigg).
\end{align*}
\end{proof}

For the interpretation of this result, a few remarks are in order.

\begin{remark}
\label{rem:Frob_DBMR_bound}
\quad
\begin{enumerate}[(a)]
    \item 
    Note that \eqref{eq:Frob_DBMR_bound} provides a (weaker) \emph{a priori} bound for $\kappa = \kappa^{\mathrm{pr}}$ and a (sharper) \emph{a posteriori} estimate for $\kappa = \kappa^{\mathrm{post}}$ due to its dependence on the solution $\Lambda$ of the DBMR problem.    

    \item
    In the proof of Theorem~\ref{thm:Frob_DBMR_bound} we used only two of the four inequalities established in Proposition~\ref{prop:l2_Pinsker} \rev{in Appendix~\ref{section:Pinsker_l_2}}.
    Using all four inequalities, and defining $\kappa \coloneqq \max(\kappa^{(1)},\kappa_{q}^{(1)},\kappa^{(2)},\kappa_{q}^{(2)})$ with
    \begin{align*}
    \kappa^{(1)}
    &\coloneqq
    \tfrac{m}{2}
    \min_{i=1,\dots,m} q_{i} \, 
    \min_{j=1,\dots,n} \fB(P_{\bullet j} - \Lambda_{\bullet j}),
    \\    
    \kappa^{(2)}
    &\coloneqq
    \tfrac{m}{2}
    \min_{i=1,\dots,m} q_{i} \, 
    \min_{j=1,\dots,n} \fB(P_{\bullet j}) (1 - \alpha_{j}),
    \end{align*}    
    one would obtain a seemingly sharper bound in \eqref{eq:Frob_DBMR_bound}.
    However,
    \begin{alignat*}{3}
    \kappa_{q}^{(1)}
    &=
    \tfrac{1}{2} \min_{j=1,\dots,n}
    \frac{\norm{P_{\bullet j} - \Lambda_{\bullet j}}_{1}}
    {\max_{i} \frac{|P_{i j} - \Lambda_{i j}|}{q_{i}}}
    && \geq
    \frac{\min_{i} q_{i}}{2}\, 
    \min_{j}
    \frac{\norm{P_{\bullet j} - \Lambda_{\bullet j}}_{1}}{\norm{P_{\bullet j} - \Lambda_{\bullet j}}_{\infty}}
    &&
    =
    \kappa^{(1)},
    \\
    \kappa_{q}^{(2)}
    &=
    \tfrac{1}{2} \min_{j}
    \frac{1 - \alpha_{j}}
    {\max_{i} \frac{P_{ij}}{q_{i}}}
    && \geq
    \frac{\min_{i} q_{i}}{2}\, 
    \min_{j}
    \frac{1 - \alpha_{j}}
    {\norm{P_{\bullet j}}_{\infty}}
    &&
    =
    \kappa^{(2)},
    \end{alignat*}
    so this would not lead to an improvement over Theorem~\ref{thm:Frob_DBMR_bound}.
    
    \item
    Clearly, the higher the $q$-balancedness $\fB_{q}(P_{\bullet j})$ of $P_{\bullet j}$ in the formula of $\kappa_{q}^{(2)}$ is for each $j=1,\dots,n$, the sharper the inequality \eqref{eq:Frob_DBMR_bound} becomes.
    Note that this balancedness is large (for fixed $j$) if $P_{\bullet j} \approx q$.
    The dynamic interpretation of $P_{\bullet j} \approx q$ is that the state $j$ is mapped to a distribution that is close to the final distribution $q$. If that is true for every $j$, then there is little coherence in the system, as $\smash{ P \approx q \mathds{1}_{[n]}^{\top} }$, with singular values $\sigma_1 = 1$ and~$\sigma_n \approx 0$, $n\ge 2$.
    In contrast, $\smash{ \kappa_{q}^{(1)} }$ is large if the $q$-balancedness of the \emph{difference}, $\fB_{q}(P_{\bullet j} - \Lambda_{\bullet j})$, is large for each $j=1,\dots,n$.
    On the one hand, this seems to be a less restrictive requirement than the previous one. On the other hand, it is harder to characterize a priori, as all we know about the columns of $P$ and $\Lambda$ is that they are probability vectors, hence their difference has zero mean.
    
    \item DBMR maximizes $\hat{\ell}(\lambda,\Gamma)$ over all pairs of stochastic matrices of given fixed di\-men\-sions, cf.\ Problem~\ref{prob:DBMR}. Thus, within the bound given in Theorem~\ref{thm:Frob_DBMR_bound}, DBMR minimizes the Frobenius norm of the difference between the full model and the low rank model.
    By the Eckart--Young--Mirsky theorem \cite[Theorem~4.4.7]{HsEu15}, the best rank-$\lowrank$ approximation of a matrix with respect to the Frobenius norm is given by the composition of the leading $\lowrank$ singular modes of the matrix, cf.\ \eqref{equ:leading_r_singular_modes}.
    Theorem~\ref{thm:Frob_DBMR_bound} thus states that the optimal DBMR solution is a quasi-optimal approximation of the leading $\lowrank$ singular modes of~$\tP$, and hence to the coherence problem, as discussed in section~\ref{sec:coh_svd}.
    In particular, the bound in \eqref{eq:Frob_DBMR_bound} is zero if $P=\Lambda$, as any reasonably tight bound of $\|\tP-\tLam\|_F$ should be.
    
    \item The seeming dependence of \eqref{eq:Frob_DBMR_bound} on $S$ is deceptive, since $\hat \ell$ itself scales ``linearly'' with~$S$.
    More precisely, the right-hand side of the bound converges almost surely as $S\to \infty$, if the data acquisition procedure is such that $\frac1S N$ converges almost surely for~$S\to\infty$. This is obvious from~\eqref{eq:DBMR_objective}.
    The i.i.d.\ sampling procedure assumed in section~\ref{section:Coherent_sets} satisfies this condition by the law of large numbers.

    \item We note that a related ``balancedness'' concept plays a role in a different a posteriori refinement of Pinsker's inequality~\cite[Theorem~2.1]{ordentlich2005distribution}, relating the total variation distance and the Kullback--Leibler divergence.
\end{enumerate}
\end{remark}

An alternative interpretation to Theorem~\ref{thm:Frob_DBMR_bound} arises by invoking Corollary~\ref{cor:DBMR_as_best_approximation}. 

\begin{corollary}
\label{cor:relating_objectives}
Under the assumptions of Theorem~\ref{thm:Frob_DBMR_bound},
\begin{equation}
% \label{equ:bound_on_degree_of_coherence}
\cC_{\lowrank}(\Lambda)
\ge
\frac{1}{\kappa S}
\Big( \hat{\ell}(\lambda,\Gamma) - \hat{\ell}(P,\Id_n) \Big) + \big\| \tP \big\|_{F}^{2} \, .
\end{equation}
\end{corollary}
\begin{proof}
Using Corollary~\ref{cor:DBMR_as_best_approximation} and noting that all singular values of $\tLam$ satisfy $\sigma_i(\tLam) \in [0,1]$, we obtain
\[
\cC_{\lowrank}(\Lambda)
=
\sum_{i=1}^\lowrank \sigma_i(\tLam)
\ge
\sum_{i=1}^\lowrank \sigma_i^2(\tLam)
=
\big\| \tLam \big\|_F^2
=
\| \tP\| _{F}^{2} - \| \tP-\tLam\| _{F}^{2}.
\]
The claim follows directly from Theorem~\ref{thm:Frob_DBMR_bound}.
\end{proof}
Therefore, relying on the a priori error bound in \eqref{eq:Frob_DBMR_bound} with $\kappa = \kappa^{\mathrm{pr}}$, an increase of the DBMR objective $\hat{\ell}(\lambda,\Gamma)$ results in a sharper lower bound on the degree of coherence in~$\Lambda = \lambda\Gamma$.
% Therefore, if we were to neglect the dependence of $\kappa$ on $\Lambda$, the degree of coherence in~$\Lambda = \lambda\Gamma$ would be monotonic in $\hat{\ell}(\lambda,\Gamma)$. That this is almost the case, we show in Appendix~\ref{sec:DBMR_stat} for an example, and we validate the expressivity of the bound in the next section.
% % Therefore, the degree of coherence in~$\Lambda = \lambda\Gamma$ is quasi-monotonic in $\hat{\ell}(\lambda,\Gamma)$, meaning that $\kappa$ still depends on $\Lambda$.
% A natural next question would be whether $\frac{1}{\kappa S} \big( \hat{\ell}(\lambda,\Gamma) - \hat{\ell}(P,\Id_n) \big)$ is monotonic in~$\hat{\ell}(\lambda,\Gamma)$.

\section{Numerical examples}
\label{sec:numerical_examples}

We will consider three examples and compare the performance of Algorithm~\ref{alg:the_alg} im\-ple\-menting DBMR \cite{gerber2017toward} which is available as open access, see section \ref{sec:code_availability}, with Algorithm~\ref{alg:classical_approach_coherence} implementing the classical approach to coherence.
In the first example we model a transition matrix with two perfectly coherent partition elements where one of these elements can again be subdivided into two strongly, but not perfectly coherent sets.
The second example is a discrete version of a map with three large (and several small) coherent sets; see~\cite[Example 1]{FrLlSa10}.
The third one, a benchmark fluid-dynamical system, will be elaborated in Example~\ref{example:periodically_perturbed_double_gyre} below and the subsequent discussion here only refers to the first two examples.

In each example we consider three different perturbations of the transition matrix: unperturbed ($\varepsilon = 0$), slightly perturbed ($\varepsilon = 2$ or~1) and strongly perturbed ($\varepsilon = 10$ or~4) in the following sense:
Each data point $(x,y) \in \bm{D}$ is replaced by a uniform random point from the set
\begin{equation}
\label{equ:perturbation_of_data}
\big( (x + \{ -\varepsilon,\dots,\varepsilon \}) \, \mathrm{mod}\, n \big) \times \big( (y + \{ -\varepsilon,\dots,\varepsilon \}) \, \mathrm{mod}\, n \big).
\end{equation}
Note that we assume the states to be ordered periodically, i.e., states $1$ and $n$ are adjacent.
For DBMR we perform $100$ independent runs with randomly generated initial affiliation matrices $\Gamma^{(0)}$ (i.e., the columns of this matrix are independent uniform random samples of the $\lowrank$ canonical unit vectors) and the best result in terms of the DBMR objective~\eqref{eq:DBMR_objective} is taken.
The following criteria for coherence are considered for the comparison between DBMR and the classical approach with $\lowrank=3$ latent states:
\begin{enumerate}[(a)]
\item%[(a)]
\label{coherence_criteria:singular_values}
The second and third singular values $\sigma_{2},\sigma_{3}$ (note that $\sigma_{1}=1$ by construction) of the low rank projected (and reweighted) transition matrices $\tLam$ (of DBMR) and of $\tP_{\textup{red}}$ of the classical approach (note that, by construction, the latter coincide with the ones of the full rank transition matrix $\tP$) are presented in Tables~\ref{Tab:sv_model_case} and \ref{Tab:sv_tripling}.
\item%[(b)]
\label{coherence_criteria:DBMR_objective}
The DBMR objective~\eqref{eq:DBMR_objective} is evaluated for the resulting affiliation matrices `DBMR-$\Gamma$' and `SVD-$\Gamma$' (which naturally correspond to partitions by~\eqref{equ:correspondence_partition_affiliation_matrix}) of Algorithms~\ref{alg:classical_approach_coherence} and~\ref{alg:the_alg}, and compared to the `default' affiliation matrix `default-$\Gamma$' given by our construction of the example (see the descriptions below).
For this purpose, the corresponding matrix $\lambda$ is chosen to maximize $\hat{\ell}(\cdot,\Gamma)$ in~\eqref{eq:DBMR_objective} and, hence, is given by~\eqref{eq:lamda_hat}.
In addition, we compare these values to the `reference value' $\hat\ell(P,\Id_{n})$ of the unreduced model.
The corresponding values are presented in Tables~\ref{Tab:sv_model_case} and~\ref{Tab:sv_tripling}.
\item%[(c)]
\label{coherence_criteria:bound_tightness}
Finally, we compare the objectives of the different model reduction tools by con\-si\-dering the tightness of the bound in Theorem~\ref{thm:Frob_DBMR_bound}. Recall
% from Remark~\ref{rem:Frob_DBMR_bound} (d) and 
from Corollary~\ref{cor:relating_objectives}
% that by the Eckart--Young--Mirsky theorem 
that the quantity $\|\tP - \tLam\|_F^2$ is approximately monotonic in the degree of coherence~$\cC$ introduced in Definition \ref{def:coherence}. More precisely, the larger $\cC(\Lambda)$, the smaller $\|\tP - \tLam\|_F^2$. 
Tables~\ref{Tab:sv_model_case} and~\ref{Tab:sv_tripling} show both sides of the inequality~\eqref{eq:Frob_DBMR_bound} for $\Lambda$ obtained in the best DBMR run.
% \ik{I think we should add that the objectives themselves and their correlation is visualized in Figure~\ref{fig:DBMRtraj} in a really convincing way:}
In addition, Figure~\ref{fig:DBMRtraj} in Appendix~\ref{sec:DBMR_stat} illustrates in how far the two objectives $\|\tP - \tLam\|_F^2$ and $\hat{\ell}$ are in line by comparing their values for a large number of corresponding pairs $(\lambda,\Gamma)$.
\item%[(a')]
\label{coherence_criteria:visual_comparison}
A visual comparison is performed in Figures~\ref{fig:CSI_1} and~\ref{fig:CSI_2} by plotting the transition matrix $P$ and its reduced versions $P_{\textup{red}}$ (SVD)\footnote{Note that $P_{\textup{red}}$ can have negative entries, so it need not be a transition matrix.} and $\Lambda = \lambda \Gamma$ (DBMR), where larger transition probabilities correspond to darker shades of gray.
In addition, the partitions of the input states corresponding to the respective affiliation matrices (default-$\Gamma$, SVD-$\Gamma$ and DBMR-$\Gamma$) are color-coded in yellow, green and red on the bottom line of the matrix images.
\end{enumerate}

%%%%%%%%%%%% ab hier erstes Beispiel

\begin{figure}[htbp]
	\centering
	\begin{subfigure}[b]{0.015\textwidth}
		\centering
		\raisebox{4.0em}{\rotatebox{90}{\rev{Full matrix $P$}}}
	\end{subfigure}
	\hfill
	\begin{subfigure}[b]{0.32\textwidth}
		\centering
		\caption*{\hspace{1.3em} \rev{Unperturbed}}
		\includegraphics[width=\textwidth]{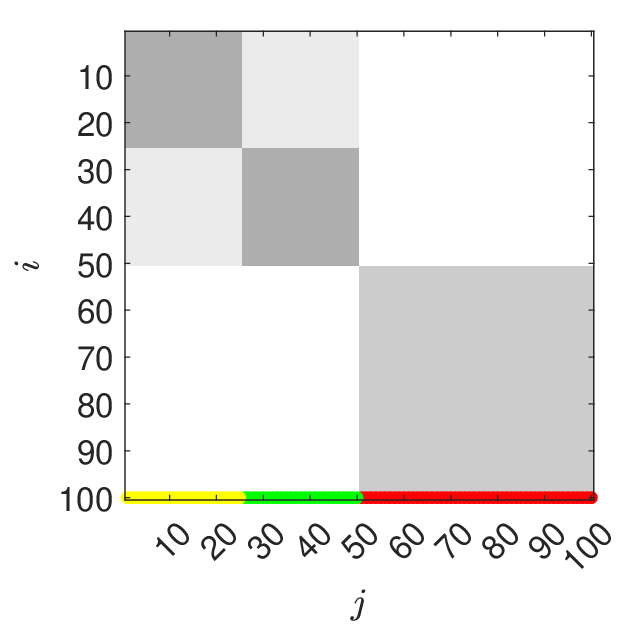}
		\phantomsubcaption
	\end{subfigure}
	\hfill
	\begin{subfigure}[b]{0.32\textwidth}
		\centering
		\caption*{\hspace{1.3em} \rev{Slightly Perturbed}}
		\includegraphics[width=\textwidth]{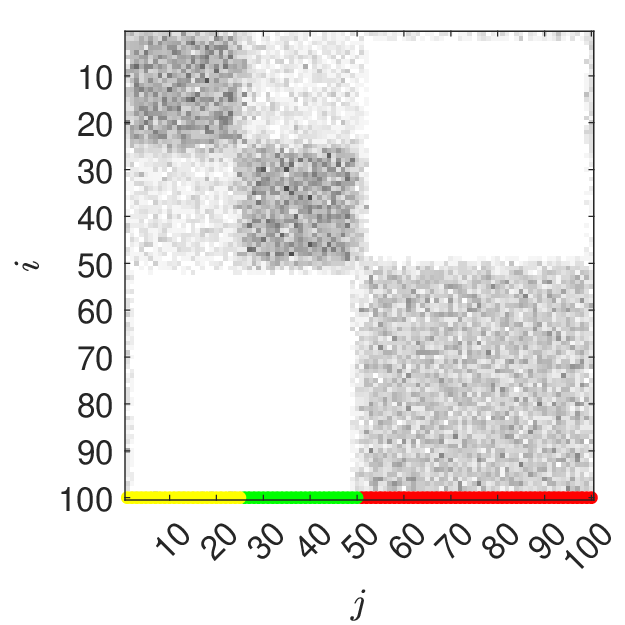}
		\phantomsubcaption
	\end{subfigure}
	\hfill
	\begin{subfigure}[b]{0.32\textwidth}
		\centering
		\caption*{\hspace{1.3em} \rev{Strongly Perturbed}}
		\includegraphics[width=\textwidth]{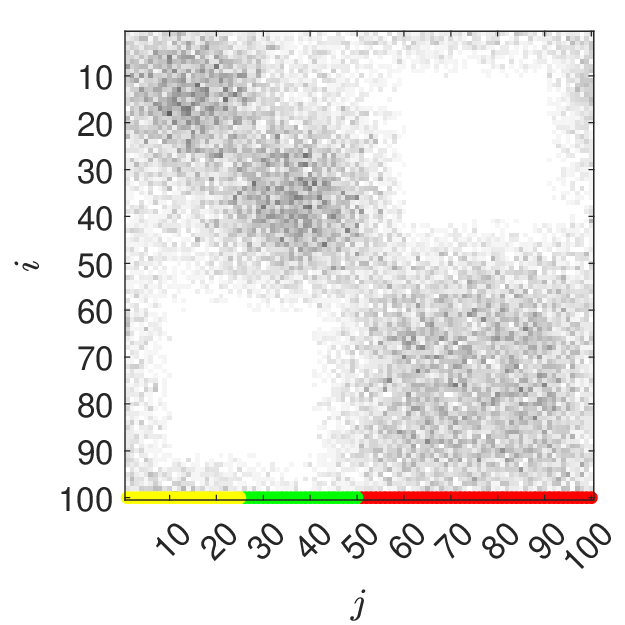}
		\phantomsubcaption
	\end{subfigure} 
	\vfill
	\begin{subfigure}[b]{0.015\textwidth}
		\centering
		\raisebox{3.4em}{\rotatebox{90}{\rev{Red.\ matrix $P_{\textup{red}}$}}}
	\end{subfigure}
	\hfill
	\begin{subfigure}[b]{0.32\textwidth}
		\centering
		\includegraphics[width=\textwidth]{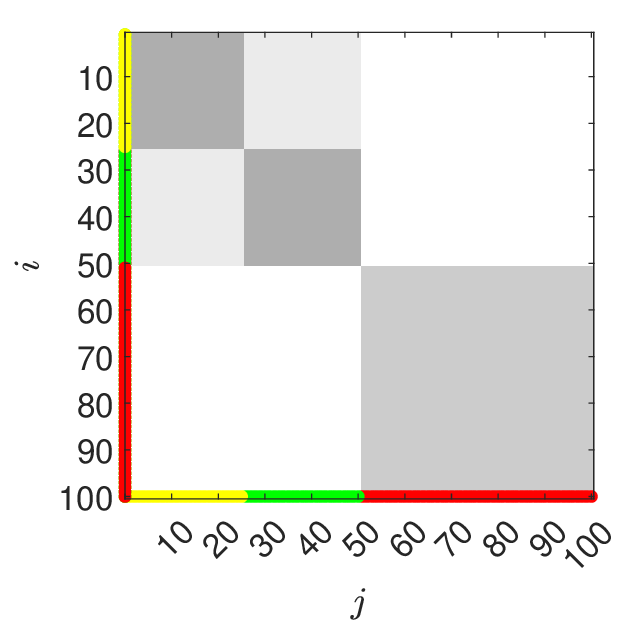}
	\end{subfigure}
	\hfill
	\begin{subfigure}[b]{0.32\textwidth}
		\centering
		\includegraphics[width=\textwidth]{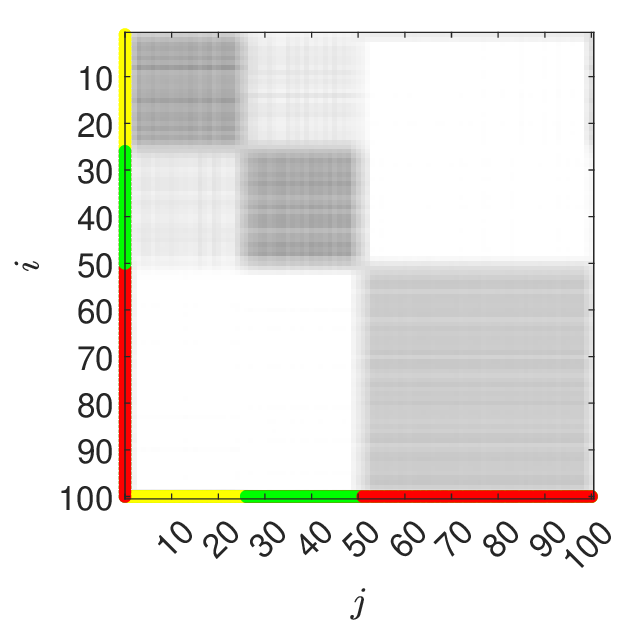}
	\end{subfigure}
	\hfill
	\begin{subfigure}[b]{0.32\textwidth}
		\centering
		\includegraphics[width=\textwidth]{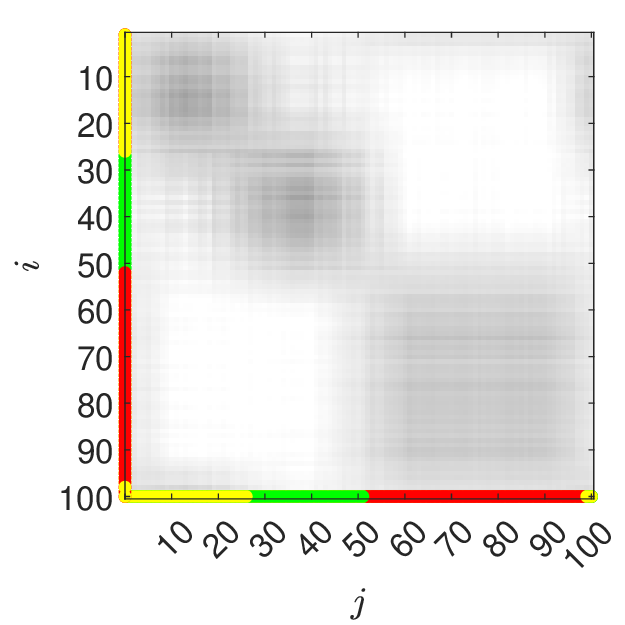}
	\end{subfigure}
	\vfill
	\begin{subfigure}[b]{0.015\textwidth}
		\centering
		\raisebox{2.6em}{\rotatebox{90}{\rev{Red.\ matrix $\Lambda = \lambda \Gamma$}}}
	\end{subfigure}
	\hfill
	\begin{subfigure}[b]{0.32\textwidth}
		\centering
		\includegraphics[width=\textwidth]{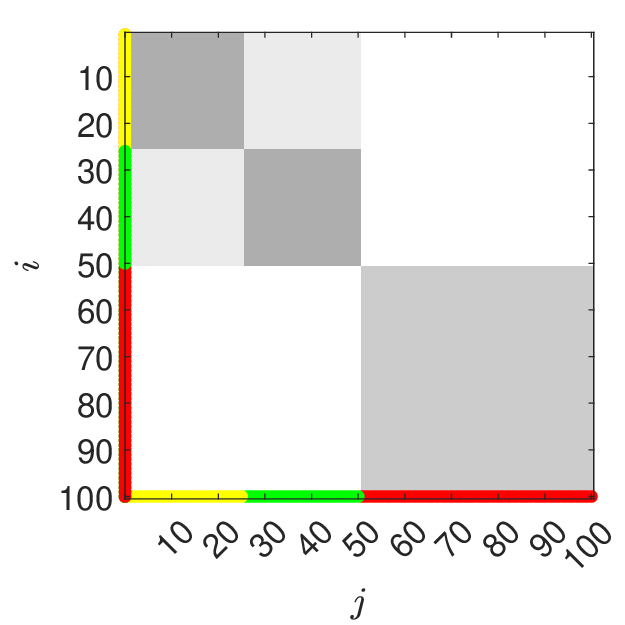}
	\end{subfigure}
	\hfill
	\begin{subfigure}[b]{0.32\textwidth}
		\centering
		\includegraphics[width=\textwidth]{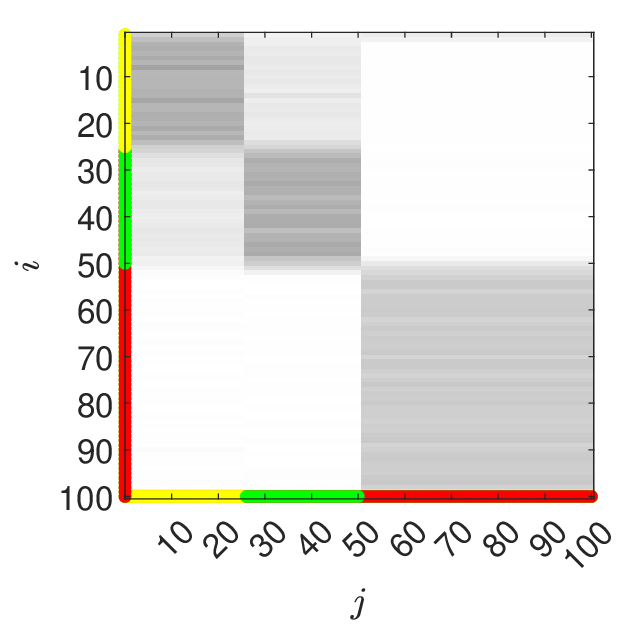}
	\end{subfigure}
	\hfill
	\begin{subfigure}[b]{0.32\textwidth}
		\centering
		\includegraphics[width=\textwidth]{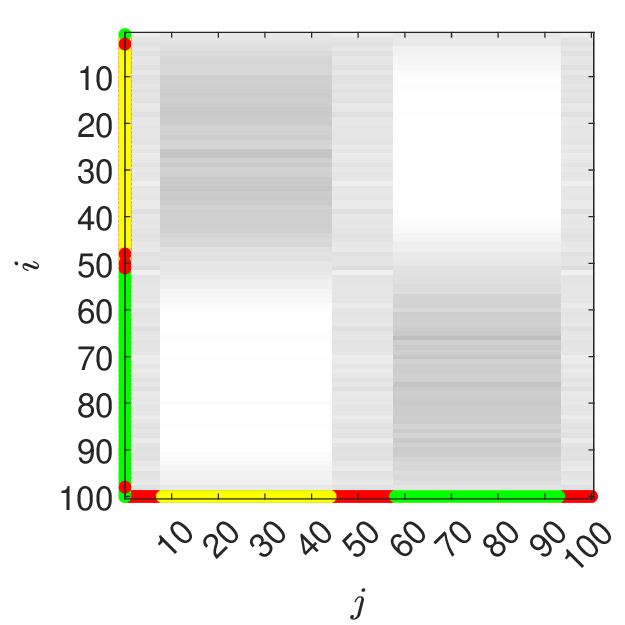}
	\end{subfigure}
	\caption{Coherent set identification for Example~\ref{example:1_three_coherent_sets} with $\lowrank=3$ clusters and $3$ different levels of perturbation:
    \textit{Top:} Full transition matrix $P$.
        \textit{Middle:} Reduced transition matrix $P_{\textup{red}}$ obtained within the classical Algorithm~\ref{alg:classical_approach_coherence}.
        \textit{Bottom:} Reduced transition matrix $\Lambda = \lambda \Gamma$ of the DBMR Algorithm~\ref{alg:the_alg}.
        The coloring at the bottom line of each plot corresponds to the clustering given by the associated affiliation matrix $\Gamma$ (or partition $\cE$):
        default-$\Gamma$ (top), SVD-$\Gamma$ (middle), and DBMR-$\Gamma$ (bottom). \rev{The coloring on the left margin of the plots represents the corresponding partitions~$\cF$.
        For the middle row, $\cF$ is obtained by clustering and matching (lines 5--6 of Algorithm~\ref{alg:classical_approach_coherence}), and for the bottom row by~\eqref{eq:Sets_F_k}.}
        }
	\label{fig:CSI_1}
\end{figure}

\begin{example}[three coherent sets]
\label{example:1_three_coherent_sets}
Our first example is an idealized dynamics having two perfectly coherent sets, one of which can further be subdivided into two less coherent sets.
We take $n=m=100$ input and output states and define the coherent sets $E_1 = \{1,\ldots,25\}$, $E_2 = \{26,\ldots,50\}$, and $E_3 = \{51,\ldots,100\}$ which partition both $[n]$ and~$[m]$. The data set $\bm{D}$ consists of $S=25000$ pairs $(X_u,Y_u)$, $u=1,\ldots,S$, and is constructed such that
\[
N_{ij} = \left\{
\begin{array}{ll}
	8, & i,j \in E_1 \text{ or }i,j\in E_2, \\
	2, & i\in E_1, j\in E_2 \text{ or } i\in E_2,j\in E_1, \\
	5, & i,j \in E_3, \\
	0, & \text{otherwise.}
\end{array}
\right.
\]
Hence, there are $\sum_{i=1}^m N_{ij} = 250$ transitions out of every state.
As discussed above, we also consider two perturbed version of the above data given by \eqref{equ:perturbation_of_data} for $\varepsilon = 2, 10$.

The resulting coherence criteria \eqref{coherence_criteria:singular_values}, \eqref{coherence_criteria:DBMR_objective}, \eqref{coherence_criteria:bound_tightness} and \eqref{coherence_criteria:visual_comparison} described above are summarized in Table~\ref{Tab:sv_model_case} and visualized in Figure~\ref{fig:CSI_1}.
Note that the rank of the unperturbed transition matrix $P$ is $\lowrank=3$, allowing the truncated SVD to match the exact transition matrix.
Also, since $P$ has only three different columns, the DBMR result with $\lowrank=3$ latent states coincides with $P$ (cf.\ Example~\ref{example:DBMR_as_projection}).
As expected, in both the full and the reduced models, coherence (measured by the singular values) as well as the values of $\hat{\ell}$, are decreasing with increasing perturbation strength.
We observe that, for unperturbed and slightly perturbed data, the reduced models as well as the partitions visualized in Figure~\ref{fig:CSI_1} are rather similar, as are the singular values $\sigma_{2}, \sigma_{3}$ and the values of $\hat{\ell}$ in Table~\ref{Tab:sv_model_case}.
On the other hand, for strong perturbations ($\varepsilon=10$), we report larger differences in all of the above criteria, suggesting that the two objectives $\cC_{3}$ and $\hat{\ell}$ are not entirely aligned.
\rev{In particular, the DBMR objective prefers merging the two small coherent sets into one and assigning a `mixing zone' on the interface between the two halves of the state space as a third cluster.}
Finally, we consider the tightness of the inequality~\eqref{eq:Frob_DBMR_bound} in Table~\ref{Tab:sv_model_case}~(c).
Since the transition matrix $P$ in the unperturbed case has rank three, DBMR is exact and both values are below machine precision. We observe for the perturbed cases that there is a factor of 10 between the left-hand side and right-hand side of the inequality, which could be indicative of the different nature of the objectives that are optimized in Problem~\ref{prob:Frobenius} and Problem~\ref{prob:DBMR}. 
\end{example}

\begin{table}[!h]
	\centering
	\renewcommand{\arraystretch}{1.5}
    \begin{tabular}{p{2.3em}<{\centering} p{5.3em}<{\centering} p{5.3em}<{\centering} p{2.5em}<{\centering} p{2.5em}<{\centering} p{2.5em}<{\centering} p{2.5em}<{\centering} p{2.5em}<{\centering} p{2.5em}<{\centering}}	
		\toprule
		& \multicolumn{2}{c}{Perturbation} & \multicolumn{2}{c}{$\varepsilon=0$} & \multicolumn{2}{c}{$\varepsilon=2$} & \multicolumn{2}{c}{$\varepsilon=10$}
		\\
		\midrule
		\multirow{2}{*}{\eqref{coherence_criteria:singular_values}} &
		\cellcolor{black!10} $\sigma_{2}(\tP)$ & $\sigma_{3}(\tP)$ & \cellcolor{black!10} $1.000$ & $0.600$ & \cellcolor{black!10} $0.939$ & $0.545$ & \cellcolor{black!10} $0.725$ & $0.362$
		\\ %\hline
		& \cellcolor{black!10} $\sigma_{2}(\tLam)$ & $\sigma_{3}(\tLam)$  & \cellcolor{black!10}$1.000$ & $0.600$  & \cellcolor{black!10}$0.918$ & $0.528$ & \cellcolor{black!10}$0.702$ & $0.071$
		\\
		\midrule		
		\multirow{4}{*}{\eqref{coherence_criteria:DBMR_objective}} &
		\multicolumn{2}{c}{$\hat{\ell}(\lambda,\Gamma)$ for default-$\Gamma$} & \multicolumn{2}{c}{$-0.954\cdot 10^5$} & \multicolumn{2}{c}{$-0.997\cdot 10^5$} & \multicolumn{2}{c}{$-1.077\cdot 10^5$}
		\\ %\hline
		& \multicolumn{2}{c}{$\hat{\ell}(\lambda,\Gamma)$ for SVD-$\Gamma$} & \multicolumn{2}{c}{$-0.954\cdot 10^5$} & \multicolumn{2}{c}{$-0.997\cdot 10^5$} & \multicolumn{2}{c}{$-1.076\cdot 10^5$}
		\\ %\hline
		& \multicolumn{2}{c}{$\hat{\ell}(\lambda,\Gamma)$ for DBMR-$\Gamma$} & \multicolumn{2}{c}{$-0.954\cdot 10^5$} & \multicolumn{2}{c}{$-0.997\cdot 10^5$} & \multicolumn{2}{c}{$-1.072\cdot 10^5$}
		\\ %\hline
		& \multicolumn{2}{c}{reference value $\hat\ell(P,\Id_{n})$} & \multicolumn{2}{c}{$-0.954\cdot 10^5$} & \multicolumn{2}{c}{$-0.951\cdot 10^5$} & \multicolumn{2}{c}{$-1.012\cdot 10^5$}   
		\\
		\midrule		
		\multirow{3}{*}{\eqref{coherence_criteria:bound_tightness}} &
		\multicolumn{2}{c}{$\| \tP - \tLam \|_F^2$} & \multicolumn{2}{c}{$5.2083\cdot 10^{-31}$} & \multicolumn{2}{c}{$0.4123$} & \multicolumn{2}{c}{$0.5443$} \\ 	
		& \multicolumn{2}{c}{$\kappa^{-1} \sum_{j=1}^n p_j \KLD{P_{\bullet j}}{\Lambda_{\bullet j}}$} & \multicolumn{2}{c}{$8.2991\cdot 10^{-17}$}
		& \multicolumn{2}{c}{$4.3351$} & \multicolumn{2}{c}{$3.9866$} \\ %\hline
		& \multicolumn{2}{c}{$\kappa$} & \multicolumn{2}{c}{$0.1562^{(2)}$} 
		& \multicolumn{2}{c}{$0.0424^{(1)}$} & \multicolumn{2}{c}{$0.0621^{(1)}$} \\
		\bottomrule
	\end{tabular}
	\caption{Coherence criteria \eqref{coherence_criteria:singular_values},
		\eqref{coherence_criteria:DBMR_objective} and
		\eqref{coherence_criteria:bound_tightness} discussed above for the comparison of the classical approach to coherence (Algorithm~\ref{alg:classical_approach_coherence}) and DBMR (Algorithm~\ref{alg:the_alg}) \rev{for Example~\ref{example:1_three_coherent_sets}}.
		In the last row, the superscript 1 or 2 indicates whether $\smash{ \kappa=\kappa^{(1)}_q }$ or $\smash{ \kappa=\kappa^{(2)}_q }$.
	}
	\label{Tab:sv_model_case}
\end{table}

%%%%%%%%%%%% ab hier zweites Beispiel

\begin{figure}[htbp]
	\centering
        \begin{subfigure}[b]{0.015\textwidth}
		\centering
		\raisebox{4.0em}{\rotatebox{90}{\rev{Full matrix $P$}}}
	\end{subfigure}
	\hfill
	\begin{subfigure}[b]{0.32\textwidth}
		\centering
		\caption*{\hspace{1.3em} \rev{Unperturbed}}
		\includegraphics[width=\textwidth]{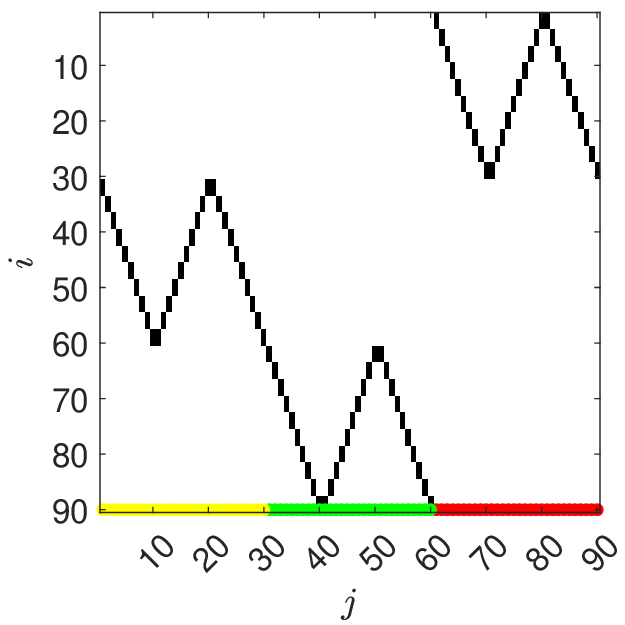}
		\phantomsubcaption
	\end{subfigure}
	\hfill
	\begin{subfigure}[b]{0.32\textwidth}
		\centering
		\caption*{\hspace{1.3em} \rev{Slightly Perturbed}}
		\includegraphics[width=\textwidth]{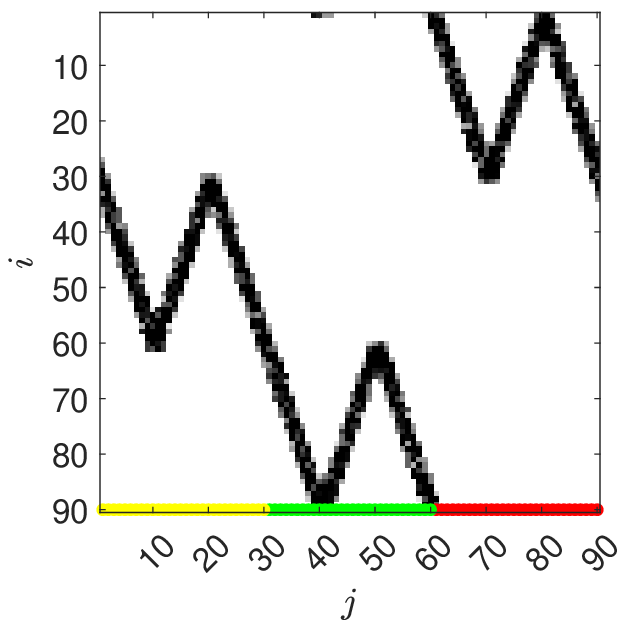}
		\phantomsubcaption
	\end{subfigure}
	\hfill
	\begin{subfigure}[b]{0.32\textwidth}
		\centering
		\caption*{\hspace{1.3em} \rev{Strongly Perturbed}}
		\includegraphics[width=\textwidth]{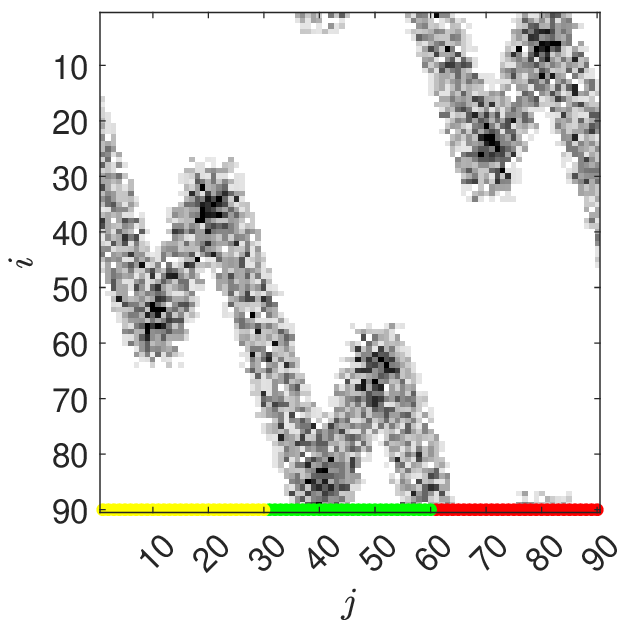}
		\phantomsubcaption
	\end{subfigure}	
    \vfill
        \begin{subfigure}[b]{0.015\textwidth}
		\centering
		\raisebox{3.4em}{\rotatebox{90}{\rev{Red.\ matrix $P_{\textup{red}}$}}}
	\end{subfigure}
	\hfill
    \begin{subfigure}[b]{0.32\textwidth}
		\centering
		\includegraphics[width=\textwidth]{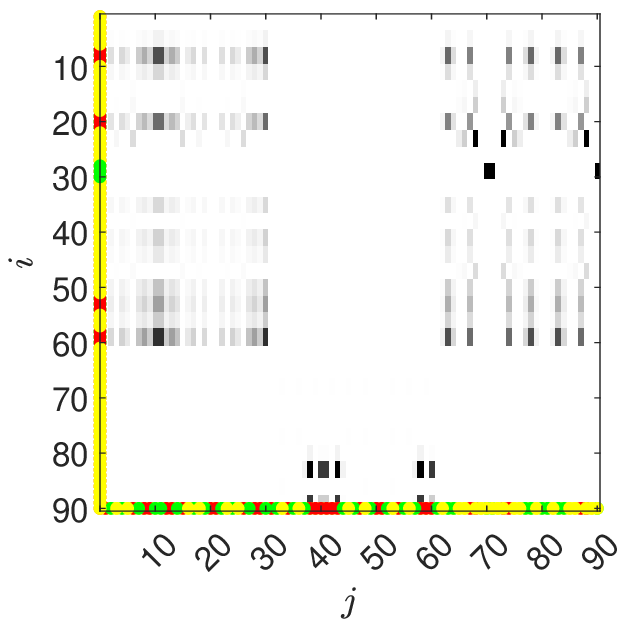}
	\end{subfigure}
	\hfill
	\begin{subfigure}[b]{0.32\textwidth}
		\centering
		\includegraphics[width=\textwidth]{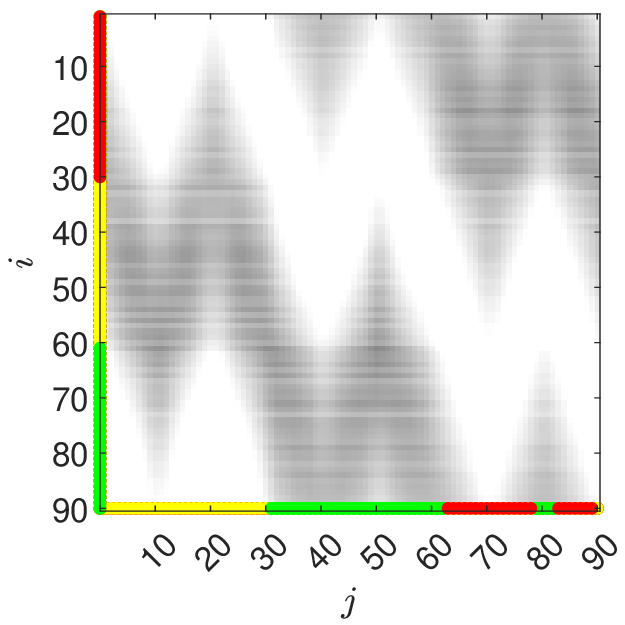}
	\end{subfigure}
 \hfill
	\begin{subfigure}[b]{0.32\textwidth}
		\centering
		\includegraphics[width=\textwidth]{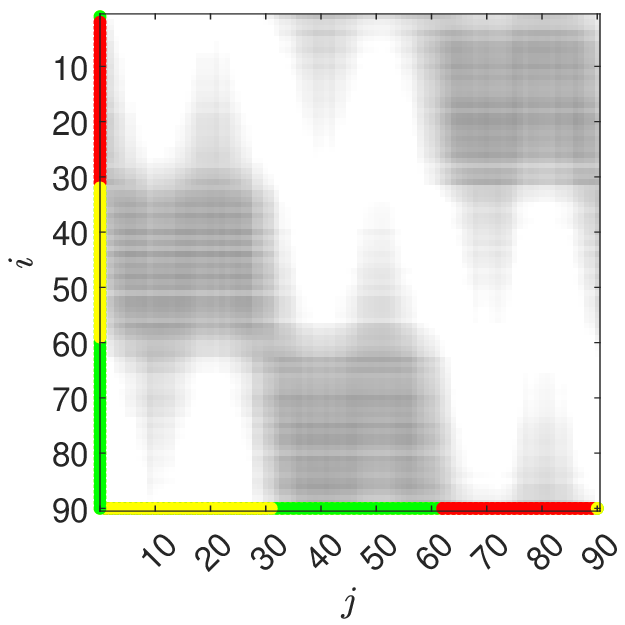}
	\end{subfigure}
    \vfill
        \begin{subfigure}[b]{0.015\textwidth}
		\centering
		\raisebox{2.6em}{\rotatebox{90}{\rev{Red.\ matrix $\Lambda = \lambda \Gamma$}}}
	\end{subfigure}
	\hfill
    \begin{subfigure}[b]{0.32\textwidth}
		\centering
		\includegraphics[width=\textwidth]{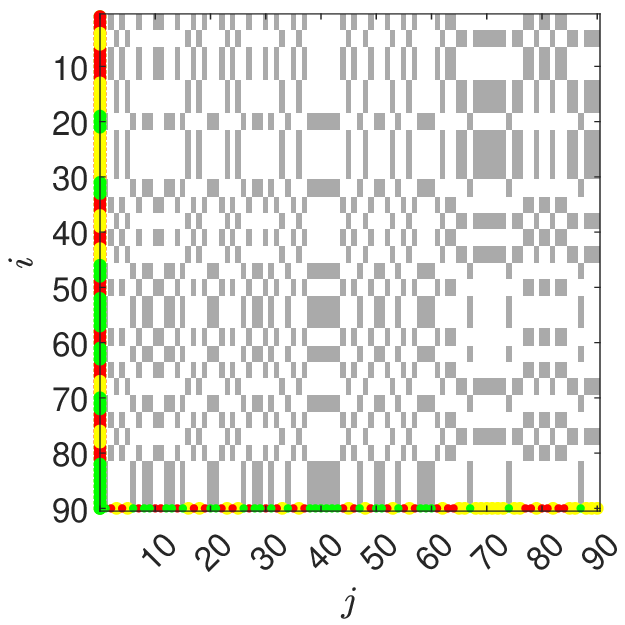}
	\end{subfigure}
	\hfill
	\begin{subfigure}[b]{0.32\textwidth}
		\centering
		\includegraphics[width=\textwidth]{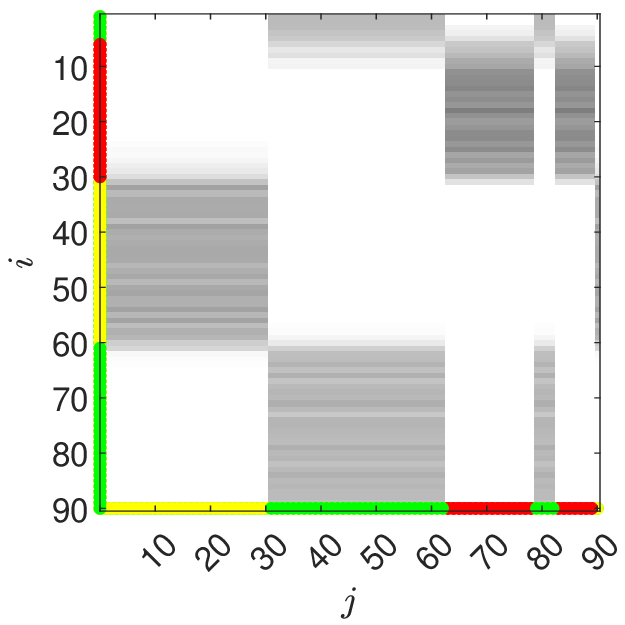}
	\end{subfigure}
 \hfill
	\begin{subfigure}[b]{0.32\textwidth}
		\centering
		\includegraphics[width=\textwidth]{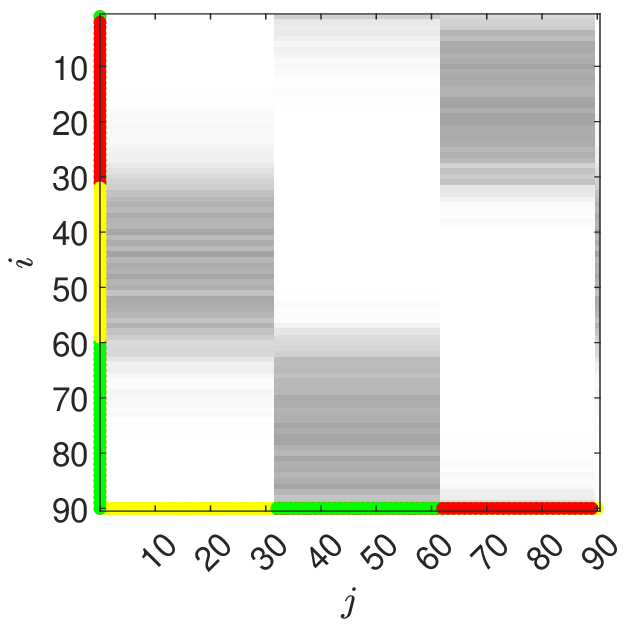}
	\end{subfigure}
	\caption{Coherent set identification for Example~\ref{example:piecewise_expanding_interval_map} with $\lowrank=3$ clusters and $3$ different levels of perturbation:
     \textit{Top:} Full transition matrix $P$.
        \textit{Middle:} Reduced transition matrix $P_{\textup{red}}$ obtained within the classical Algorithm~\ref{alg:classical_approach_coherence}.
        \textit{Bottom:} Reduced transition matrix $\Lambda = \lambda \Gamma$ of the DBMR Algorithm~\ref{alg:the_alg}.
        The coloring at the bottom line of each plot corresponds to the clustering given by the associated affiliation matrix $\Gamma$ (or partition $\cE$):
        default-$\Gamma$ (top), SVD-$\Gamma$ (middle) and DBMR-$\Gamma$ (bottom). \rev{The coloring on the left margin of the plots represents the corresponding partitions~$\cF$.
        For the middle row, $\cF$ is obtained by clustering and matching (lines 5--6 of Algorithm~\ref{alg:classical_approach_coherence}), and for the bottom row by~\eqref{eq:Sets_F_k}.}
        }         
	\label{fig:CSI_2}
\end{figure}

\begin{example}[piecewise expanding interval map]
\label{example:piecewise_expanding_interval_map}
In our second example, the number of input and output states equals $n=m=90$ with 90 transitions out of every state, totaling~$S=8100$ data points.
Each input state $j$ is paired, with equal frequencies, with three output states $i$ such that $N_{ij}=30$ for the corresponding pairs.
We do not explicitly write down how these three output states are chosen, but instead refer the reader to Figure~\ref{fig:CSI_2} (top left) as well as to \cite[Example 1]{FrLlSa10}, which this example was inspired by.
The sets $E_1 = \{1,\ldots,30\}$, $E_2 = \{31,\ldots,60\}$, and $E_3 = \{61,\ldots,90\}$ are perfectly coherent (their output `partners' being $E_2,E_3$, and $E_1$, respectively). There are also smaller perfectly coherent sets, for instance $\{61,80,81\}$, of which the output `partner' is~$\{1,2,3\}$.
There are, in fact, 30 such 3-element coherent sets, and arbitrary unions of them are also perfectly coherent. Note that the smaller a coherent set, the more its coherence will be affected by the perturbations. For the small perturbation we use $\varepsilon=1$ and for the large one we use~$\varepsilon=4$, cf.\ Figure~\ref{fig:CSI_2} (top middle and right).

The coherence criteria \eqref{coherence_criteria:singular_values}, \eqref{coherence_criteria:DBMR_objective}, \eqref{coherence_criteria:bound_tightness} and \eqref{coherence_criteria:visual_comparison} described above are reported in Table~\ref{Tab:sv_tripling} and visualized in Figure~\ref{fig:CSI_2}.
In the unperturbed case, Algorithms~\ref{alg:classical_approach_coherence} and~\ref{alg:the_alg} both identified perfectly coherent sets that we comment on in Remark~\ref{remark:coherent_set_sizes} below.
In the slightly perturbed case, the $\hat\ell$-value of DBMR-$\Gamma$ was worse than the ones of default-$\Gamma$ and SVD-$\Gamma$.
This shows that, even with a large number of 100 independent runs, DBMR was incapable of identifying the global optimum of $\hat{\ell}$, which we attribute to the large number of small coherent sets (in the unperturbed case), presumably resulting in a large number of local optima.
We did not observe this issue in the strongly perturbed case, where both Algorithms~\ref{alg:classical_approach_coherence} and~\ref{alg:the_alg} identified the default partition.
This suggests that the perturbation of $\varepsilon = 4$ was sufficient to `smoothen out' many of the local optima.
We also point out that, compared to Example~\ref{example:1_three_coherent_sets}, the inequality \eqref{eq:Frob_DBMR_bound} is sharper --- the deviating factor is between 2 and 8 rather than~10.
\end{example}

\begin{table}[!h]
	\centering
	\renewcommand{\arraystretch}{1.5}
%	\begin{tabular}{|c|c|c|c|c|c|c|c|c|}
	\begin{tabular}{p{2.3em}<{\centering} p{5.3em}<{\centering} p{5.3em}<{\centering} p{2.5em}<{\centering} p{2.5em}<{\centering} p{2.5em}<{\centering} p{2.5em}<{\centering} p{2.5em}<{\centering} p{2.5em}<{\centering}}
		\toprule
		& \multicolumn{2}{c}{Perturbation} & \multicolumn{2}{c}{$\varepsilon=0$} & \multicolumn{2}{c}{$\varepsilon=1$} & \multicolumn{2}{c}{$\varepsilon=4$}
		\\
		\midrule
		\multirow{2}{*}{\eqref{coherence_criteria:singular_values}} &
		\cellcolor{black!10} $\sigma_{2}(\tP)$ & $\sigma_{3}(\tP)$ & \cellcolor{black!10} $1.0000$ & $1.0000$ & \cellcolor{black!10} $0.9849$ & $0.9846$ & \cellcolor{black!10} $0.8961$ & $0.8948$
		\\ %\hline
		& \cellcolor{black!10} $\sigma_{2}(\tLam)$ & $\sigma_{3}(\tLam)$  & \cellcolor{black!10} $1.0000$ & $1.0000$ & \cellcolor{black!10} $0.9548$ & $0.9234$ & \cellcolor{black!10} $0.8518$ & $0.8400$
		\\
		\midrule
  %%%%%%%%%%%%%%%%%%%%%%%
		\multirow{4}{*}{\eqref{coherence_criteria:DBMR_objective}} &
		\multicolumn{2}{c}{$\hat{\ell}(\lambda,\Gamma)$ for default-$\Gamma$} & \multicolumn{2}{c}{$-0.2755\cdot 10^5$} & \multicolumn{2}{c}{$-0.2828\cdot 10^5$} & \multicolumn{2}{c}{$-0.3002\cdot 10^5$}
		\\ %\hline
		& \multicolumn{2}{c}{$\hat{\ell}(\lambda,\Gamma)$ for SVD-$\Gamma$} & \multicolumn{2}{c}{$-0.3212\cdot 10^5$} & \multicolumn{2}{c}{$-0.2828\cdot 10^5$} & \multicolumn{2}{c}{$-0.3002\cdot 10^5$}
		\\ %\hline
		& \multicolumn{2}{c}{$\hat{\ell}(\lambda,\Gamma)$ for DBMR-$\Gamma$} & \multicolumn{2}{c}{$-0.2755\cdot 10^5$} & \multicolumn{2}{c}{$-0.2861\cdot 10^5$} & \multicolumn{2}{c}{$-0.3002\cdot 10^5$}
		\\ %\hline
		& \multicolumn{2}{c}{reference value $\hat\ell(P,\Id_{n})$} & \multicolumn{2}{c}{$-0.0890\cdot 10^5$} & \multicolumn{2}{c}{$-0.1817\cdot 10^5$} & \multicolumn{2}{c}{$-0.2531\cdot 10^5$}   
		\\
		\midrule		
		\multirow{3}{*}{\eqref{coherence_criteria:bound_tightness}} &
		\multicolumn{2}{c}{$\| \tP - \tLam \|_F^2$} & \multicolumn{2}{c}{$27.000$} & \multicolumn{2}{c}{$7.5525$} & \multicolumn{2}{c}{$2.0543$} \\ 	
		& \multicolumn{2}{c}{$\kappa^{-1} \sum_{j=1}^n p_j \KLD{P_{\bullet j}}{\Lambda_{\bullet j}}$} & \multicolumn{2}{c}{$69.0776$}
		& \multicolumn{2}{c}{$40.3973$} & \multicolumn{2}{c}{$15.6709$} \\ %\hline
		& \multicolumn{2}{c}{$\kappa$} & \multicolumn{2}{c}{$0.0333^{(1)}$} 
		& \multicolumn{2}{c}{$0.0319^{(1)}$} & \multicolumn{2}{c}{$0.0370^{(1)}$} \\
		\bottomrule
	\end{tabular}
	\caption{Coherence criteria \eqref{coherence_criteria:singular_values},
		\eqref{coherence_criteria:DBMR_objective} and
		\eqref{coherence_criteria:bound_tightness} discussed above for the comparison of the classical approach to coherence (Algorithm~\ref{alg:classical_approach_coherence}) and DBMR (Algorithm~\ref{alg:the_alg})  \rev{for Example~\ref{example:piecewise_expanding_interval_map}}.
		In the last row, the superscript 1 or 2 indicates whether $\smash{ \kappa=\kappa^{(1)}_q }$ or $\smash{ \kappa=\kappa^{(2)}_q }$.
	}
	\label{Tab:sv_tripling}
\end{table}

\begin{remark}
\label{remark:coherent_set_sizes}
In Example~\ref{example:piecewise_expanding_interval_map} with no perturbation, visualized in Figure \ref{fig:CSI_2} (left), there is a large number of small perfectly coherent sets.
Hence, each partitioning of these sets into three groups will again produce perfectly coherent sets.
In that sense, both the classical Algorithm~\ref{alg:classical_approach_coherence} and the DBMR Algorithm~\ref{alg:the_alg} identify perfectly coherent sets (we verified that the sum of the leading three singular values of $\smash{ \tLam = D_q^{-1/2} \lambda\Gamma D_p^{1/2} }$, i.e.\ the degree of coherence $\cC_{3}(p,\lambda \Gamma)$, has the maximal possible value of $\lowrank=3$, $(\lambda,\Gamma)$ being the DBMR output), showing that this DBMR result is not inferior to the classical one with respect to our measure of coherence. Furthermore, DBMR performs a partitioning into groups of equal size (30 states each), cf.\ Figure~\ref{fig:CSI_2} (bottom left), while the group sizes resulting from the classical approach (namely $84/3/3$) strongly differ, cf.\ Figure~\ref{fig:CSI_2} (middle left).
As argued by \cite{froyland2010transport} in the context of $\lowrank=2$ coherent sets, equal group size is a preferable property in terms of coherence.
In fact, \cite[Section~III.A]{froyland2010transport} \emph{imposes} the two coherent sets to have approximately the same mass.
In that sense, the DBMR result is preferable to the one of the classical approach.
Note that this preferable property of coherent sets having large size is not reflected by our measures of coherence, namely the objective in~\eqref{eq:coherence_set} and the degree of coherence in Definition~\ref{def:coherence}.
\end{remark}

\begin{example}[Periodically perturbed double gyre]
\label{example:periodically_perturbed_double_gyre}

\rev{
As the concept of coherent sets arose from fluid-dynamical problems, we briefly consider a benchmark example from this field. More precisely, we consider a two-dimensional phenomeno\-logical model from \cite[sec.~6]{shadden2005definition} which contains two counter-rotating gyres next to one another, with a vertical separatrix between them undergoing a periodic oscillation in the horizontal direction.
This ``double gyre'' flow is governed by the stream function
% The third example is a phenomeno\-logical model which contains two gyres that start next to each other and rotate in opposite directions. This model should not be see as an approximate solution of real fluid flows. Over time, the two gyres rotate together within the domain so that one lies above the other. This $2$-dimensional system is a benchmark example for identifying finite-time coherent sets \cite[Sec. 3.5]{denner2017coherent}. The periodically varying double-gyre pattern occurs frequently in geophysical flows \cite{shadden2005definition,coulliette2001intergyre}. Examples include recirculation cells adjacent to the North Atlantic Current \cite[Fig. 10]{rossby1996north}. The rotating double gyre flow is described by the following stream-function 
\begin{equation} \label{eq:stream_fct_double_gyre}
	\psi(x,y,t) := A\sin(\pi f(x,t))\sin(\pi y),
\end{equation} where the function $f(x,t)$ is given by $f(x,t) = \delta\sin(\omega t) x^2+ (1-2\delta\sin(\omega t)) x$ 
% with parameters \begin{align}
% 		\label{eq:gyre_a}
% 		a(t)  & =\delta\sin(\omega t),\\
% 		\label{eq:gyre_b}
% 		b(t) & =1-2\delta\sin(\omega t)
% \end{align}
for $(x,y)$ in the domain $\Omega=[0,2]\times[0,1]\subseteq \mathbb{R}^2$. 
% The parameters $a(t)$ and $b(t)$ in Eq. \eqref{eq:gyre_a} and Eq. \eqref{eq:gyre_b} describe how far the line which separates the two rotating gyres moves to the left or right. The velocity field of the double gyre flow $b(x,y,t)=(u,v)^T\in\mathbb{R}^2$ is given by the vertical and horizontal velocities $u$ and $v$, i.e.
The flow itself is the solution of the ordinary differential equation $(\dot{x}, \dot{y}) = (-{\partial\psi}/{\partial y}, {\partial\psi}/{\partial x})$.
% \begin{align*} %\label{eq:gyre_u}
% 		\dot{x} & =-\frac{\partial\psi}{\partial y}=-\pi A\sin(\pi f(x,t))\cos(\pi y),\\ \label{eq:gyre_v}
%   \dot{y} & =\frac{\partial\psi}{\partial x}=\pi A\cos(\pi f(x,t))\sin(\pi y) \dfrac{\mathrm{d}f}{\mathrm{d}x}.
% \end{align*} 
% For $\delta=0$ the flow is time-independent and has the same pattern as in Fig. \ref{fig:velocity_field_gyre} in the left panel. For $\delta>0$ the flow is time-dependent and the gyres conversely expand and contract periodically in the $x$-direction such that the rectangle enclosing the gyres remains invariant. In Eqs. \eqref{eq:gyre_u} and \eqref{eq:gyre_v}, $A$ is the magnitude of the velocity vector field and $\omega/2\pi$ the frequency of the oscillation. 
% \begin{figure}[htbp]
% 	\centering
% \begin{subfigure}[b]{0.49\textwidth}
% 	\centering
% 	\includegraphics[width=\textwidth]{quiver_initial.eps}
% \end{subfigure}
% \hfill
% \begin{subfigure}[b]{0.49\textwidth}
% 	\centering
% 	\includegraphics[width=\textwidth]{quiver_final.eps}
% \end{subfigure} 			 
% \caption{Vector plot of velocity field of the double gyre for $t=0.25$ (left) and $t=40.75$ (right). The plotted blue arrows have the velocity components $u$ and $v$ from Eqs. \eqref{eq:gyre_u} and \eqref{eq:gyre_v}.}
% \label{fig:velocity_field_gyre}
% \end{figure}
We set the parameter values to $A=0.25$,  $\delta=0.25$ and $\omega=2\pi$, as in~\cite{froyland_padberg_09}, implying a flow map with period $\frac{2\pi}{\omega}=1$. Contour plots of the stream function half a period apart (at maximal displacements during the separatrix' oscillation) are shown in Fig.~\ref{fig:contour_DG}.
}
\begin{figure}[htbp]
	\centering
\begin{subfigure}[b]{0.49\textwidth}
	\centering
	\includegraphics[width=\textwidth]{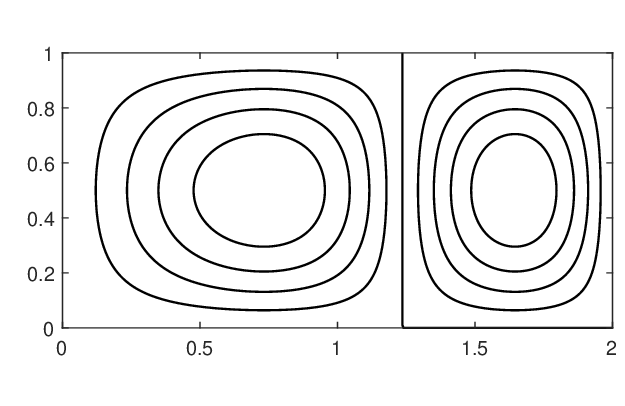}
\end{subfigure}
\hfill
\begin{subfigure}[b]{0.49\textwidth}
	\centering
	\includegraphics[width=\textwidth]{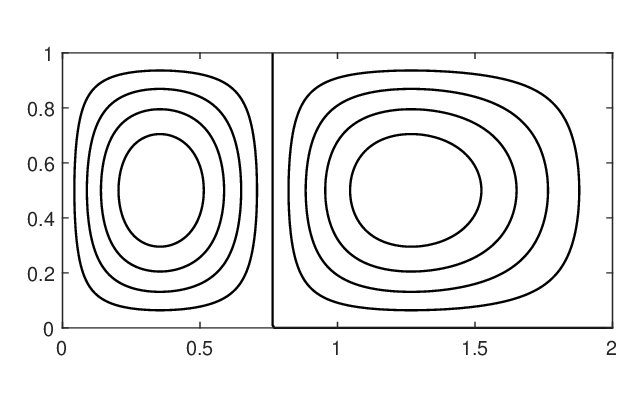}
\end{subfigure} 			 
\caption{Contour plots of the stream function $\psi$ from~\eqref{eq:stream_fct_double_gyre} at time $t = 1/4$ (left) and time $t = 3/4$ (right). At any time, the velocity field is tangential to level sets of the stream function.
% \ik{Are the time values ``$t = 1/4$ (left) and time $t = 3/4$ (right)'' consistent to the running text, where we write ``initial time $t_0=0$ to final time~$t_1=40$''?
% Something feels off, no?} \pk{We run the whole ODE for 40 periods, and the period is 1. This is fine here.}
% \ik{I see. Thank you for the explanation!}
}
\label{fig:contour_DG}
\end{figure}

\rev{
We define categories by discretizing the domain $\Omega$ uniformly into $64 \times 32 = 2048$ square boxes of edge length~$\frac{1}{32}$. For the input $X$ (initial time), we populate each box with 100 random points, which yields a sample size of~$S = 204\,800$. For the output $Y$ (final time), we consider the time-40-flow map of the double gyre system from initial time $t_0=0$ to final time~$t_1=40$ (implemented by a fourth-order Runge--Kutta scheme with constant stepsize $h=0.01$). As in the previous examples, and similarly to \cite{froyland2014almost}, we apply an additive i.i.d.\ uniform perturbation to every input and output point coordinate-wise from~$[-\rho,\rho]$, reflecting back any point that has been propelled out of $\Omega$ by this perturbation. We choose $\rho=\frac{1}{32}$. The input and output categories are then determined by assigning the (perturbed) points to the boxes they lie in.
}

% We compute an approximation with GAIO using $2^{11}=2^{6}\times2^{5}=2048$ categorical boxes and $100$ trajectories per box. The domain $\Omega=[0,2]\times[0,1]$ is discretized by $204800$ trajectories. For the time integration, we form a transition matrices $\tilde{P}\in\mathbb{R}^{2048 \times 2048}$ by integrating a fourth-order Runge Kutta scheme with constant stepsize $h=0.01$. As initial and final time, we choose $t_0=0$ and $t_1=40$. Following \cite{froyland2014almost}, a uniform perturbation at initial and final time is applied with a range of $\rho\sim\frac{1}{2^{6}}$.
% Moreover, we implement reflecting boundary conditions for $\Omega$. The norm of the velocity field of range $0$ to $1$ is plotted in Fig. \ref{fig:norm_velocity_field} at times $t=0$ and $t=40$. 
% \begin{figure}[htbp]
% 	\centering
% 	\begin{subfigure}[b]{0.49\textwidth}
% 		\centering
% 		\includegraphics[width=\textwidth]{norm_b_initial.eps}
% 	\end{subfigure}
% 	\hfill
% 	\begin{subfigure}[b]{0.49\textwidth}
% 		\centering
% 		\includegraphics[width=\textwidth]{norm_b_final.eps}
% 	\end{subfigure} 			 
% 	\caption{The norm of the velocity field $\left\Vert b\right\Vert _{2}=\sqrt{u^{2}+v^{2}}$ for $t=0$ (left) and $t=40$ (right).}
% 	\label{fig:norm_velocity_field}
% \end{figure} 

\rev{
We compare the performance of DBMR (Algorithm~\ref{alg:the_alg}) with the classical approach (Algorithm~\ref{alg:classical_approach_coherence}) when identifying $r=3$ and $r = 5$ coherent sets.
We recall that the flow is 1-periodic and we chose 40 iterates of the time-1-flow map to generate our input-output pairs. Thus, any invariant sets of the time-1-flow map that are reasonably large compared with the perturbation size $\rho$ are expected to be found coherent. The structure of these invariant sets if fairly intricate, as can be seen, e.g., in \cite[Fig.~12]{BaKoPG19}.
We expect to observe two coherent gyre-like structures centered approximately around the points $(0.5,0.5)$ and~$(1.5,0.5)$, respectively. These two gyres can be ``stratified'' further into ring-like coherent sets. The two coherent gyres and some smaller, ``satellite'' invariant sets around them are surrounded and separated by a third large coherent set that is internally mixing and which hence cannot be separated any further into reasonably coherent subsets.
In Fig.~\ref{fig:Double_gyre} (top two rows) we show the results of the classical approach for $r=3$ and $r=5$ coherent pairs from Algorithm~\ref{alg:classical_approach_coherence}, respectively. In comparison to that, the DBMR result of Algorithm~\ref{alg:the_alg} (best of $100$ runs) are shown in Fig.~\ref{fig:Double_gyre} (bottom two rows), for $r=3$ and~$r=5$, respectively.
We observe that DBMR tends to identify a larger collection of the above-mentioned invariant sets as coherent, while the classical approach focuses on the gyre cores and their stratification.
% We observe that DBMR here shows a tendency to overestimate the gyres.
% We observe that the gyres found by DBMR are wider than the ones determined by the classical approach.
}
\begin{figure}[htbp]
	\centering
	\begin{subfigure}[b]{0.49\textwidth}
		\centering
		\includegraphics[width=\textwidth]{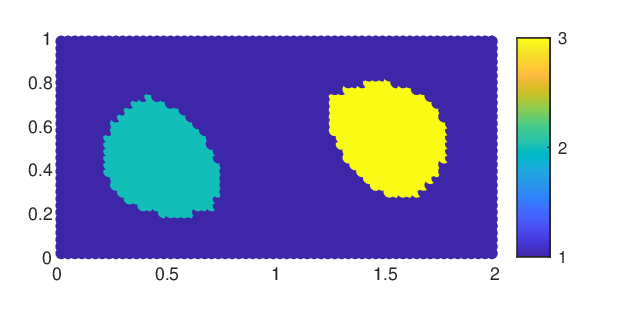}
	\end{subfigure}
	\hfill
	\begin{subfigure}[b]{0.49\textwidth}
		\centering
		\includegraphics[width=\textwidth]{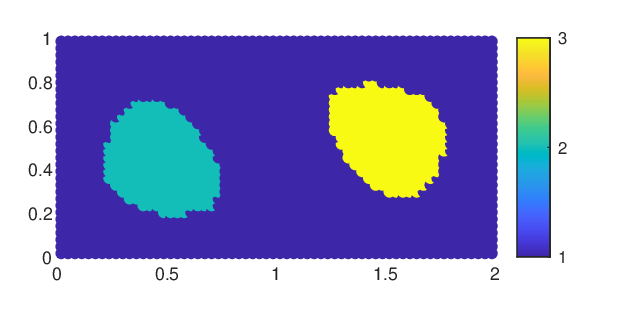}
	\end{subfigure} 	

 \begin{subfigure}[b]{0.49\textwidth}
		\centering
		\includegraphics[width=\textwidth]{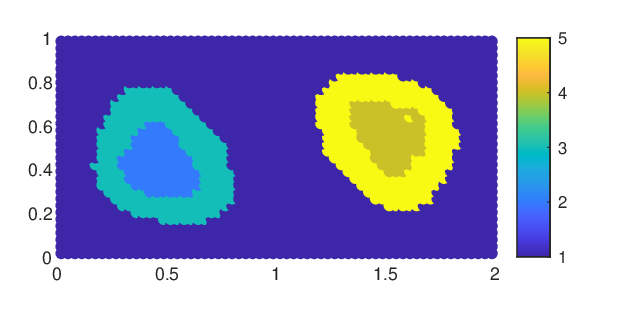}
	\end{subfigure}
	\hfill
	\begin{subfigure}[b]{0.49\textwidth}
		\centering
		\includegraphics[width=\textwidth]{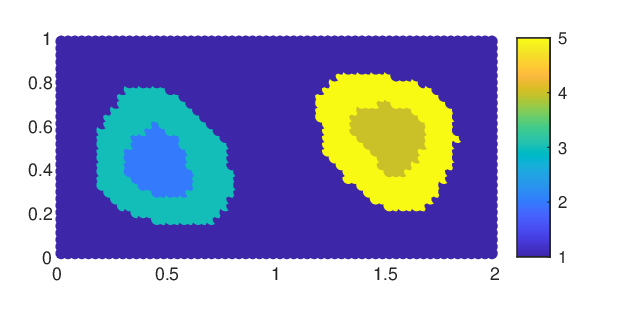}
	\end{subfigure} 

    \begin{subfigure}[b]{0.49\textwidth}
		\centering
		\includegraphics[width=\textwidth]{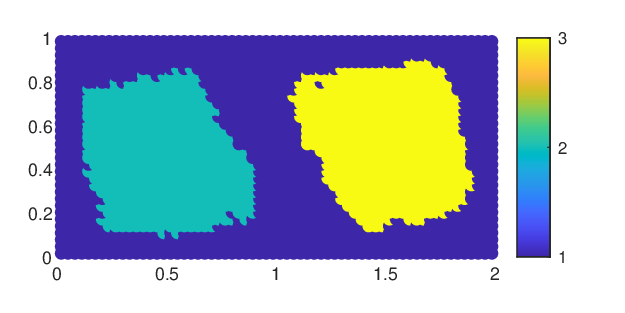}
	\end{subfigure}
	\hfill
	\begin{subfigure}[b]{0.49\textwidth}
		\centering
		\includegraphics[width=\textwidth]{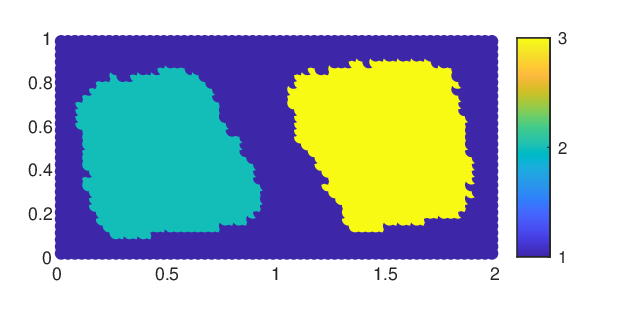}
	\end{subfigure}

    \begin{subfigure}[b]{0.49\textwidth}
		\centering
		\includegraphics[width=\textwidth]{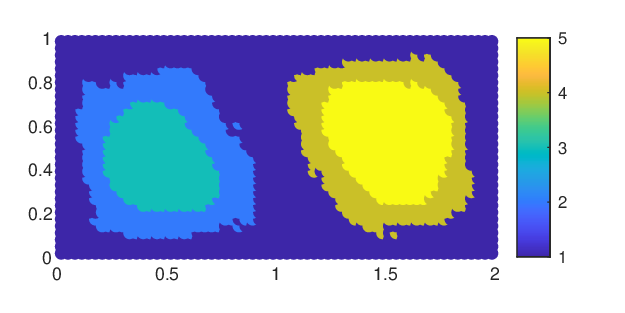}
	\end{subfigure}
	\hfill
	\begin{subfigure}[b]{0.49\textwidth}
		\centering
		\includegraphics[width=\textwidth]{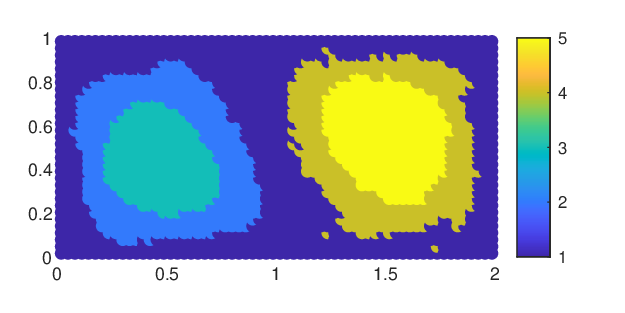}
	\end{subfigure} 
	\caption{\rev{Coherent sets for the double gyre flow from Example~\ref{example:periodically_perturbed_double_gyre}. Top two rows: The classical approach to coherence, clustering of singular vectors into $r=3$ (first row) and $r=5$ (second row) clusters to yield coherent sets at initial time (left) and at final time (right). Bottom two rows: DBMR with $r=3$ (third row) and $r=5$ (fourth row), coloring by the latent states from $\Gamma$ (left) and the sets $F_k$ given by~\eqref{eq:Sets_F_k} (right).}}
	\label{fig:Double_gyre}
\end{figure}
% \begin{figure}[htbp]
% 	\centering
% 	\begin{subfigure}[b]{0.49\textwidth}
% 		\centering
% 		\includegraphics[width=\textwidth]{graphics/dbmr_affiliation.eps}
% 	\end{subfigure}
% 	\hfill
% 	\begin{subfigure}[b]{0.49\textwidth}
% 		\centering
% 		\includegraphics[width=\textwidth]{graphics/dbmr_sets_F_k.eps}
% 	\end{subfigure} 			 
% 	\caption{Plot of optimal $\Gamma$ (left) and sets $F_k$ from Eq. \eqref{eq:Sets_F_k} with respect to optimal $\lambda$ for $r=3$ latent states and $100$ runs of DBMR.}
% 	\label{fig:DBMR_gyre}
% \end{figure}
% \begin{figure}[htbp]
% 	\centering
% 	\begin{subfigure}[b]{0.49\textwidth}
% 		\centering
% 		\includegraphics[width=\textwidth]{graphics/dbmr_affiliation_r5.eps}
% 	\end{subfigure}
% 	\hfill
% 	\begin{subfigure}[b]{0.49\textwidth}
% 		\centering
% 		\includegraphics[width=\textwidth]{graphics/dbmr_sets_F_k_r5.eps}
% 	\end{subfigure} 			 
% 	\caption{Plot of optimal $\Gamma$ (left) and sets $F_k$ from Eq. \eqref{eq:Sets_F_k} with respect to optimal $\lambda$ for $r=5$ latent states and $100$ runs of DBMR.}
% 	\label{fig:DBMR_gyre}
% \end{figure}

\end{example}

\paragraph{Classical approach versus DBMR for coherent set identification.}

In this manuscript, we have compared analytically as well as empirically the performance of two approaches to identify coherent sets of dynamical systems --- the classical approach (Algorithm~\ref{alg:classical_approach_coherence}) and DBMR (Algorithm~\ref{alg:the_alg}).
Let us briefly summarize the advantages and disadvantages of using DBMR for this task:

On the one hand, DBMR performs worse in terms of coherence \emph{given our measure of coherence} (see Definition \ref{def:coherence}).
This is almost a tautology since the classical approach by construction maximizes the degree of coherence $\cC_{\lowrank}$, whilst DBMR optimizes a different objective.
Furthermore, DBMR comes with the risk of running into local maxima of $\hat{\ell}$, cf.\ Example~\ref{example:piecewise_expanding_interval_map} with slight perturbation.

On the other hand, DBMR seems to promote large coherent sets (an attractive property in many settings of applied interest), whilst the classical approach might identify coherent sets that are very small, cf.\ Remark~\ref{remark:coherent_set_sizes} and Example~\ref{example:piecewise_expanding_interval_map} with no perturbation. In this light, the DBMR objective might be preferred, but future work developing a systematic comparison between optimization objectives for coherent sets is needed. In this direction, we conjecture that the entropic characterization of DBMR (see Remark \ref{remark:Connection_ML_KL}) can be shown to provide a theoretical foundation for the observed size-sensitive properties of DBMR. 

The `reduced transition matrix' $\lambda$ (operating on the compound input states grouped by $\Gamma$) of DBMR is a left stochastic matrix.
Therefore, DBMR is structure preserving in this sense, while the `reduced transition matrix' $\tP_{\textup{red}}$ of the classical approach can have negative entries and the entries in each of its columns typically do not sum to one.
A further advantage of DBMR might be that it does not require the approximation of the entire $mn$ entries of the ``full'' transition matrix $P$; instead it estimates the $r(n+m)$ entries (essentially, even only $rm+n$ entries, since $\Gamma$ has only one nonzero entry per column) of the factors $\lambda$ and $\Gamma$ directly from the data.
\cite{gerber2017toward} argue that the (comparatively few) matrix entries of the low rank ap\-prox\-i\-ma\-tion require far less data.
However, our experiments have neither verified nor refuted this intuition.

\section{Conclusion and Outlook}
\label{sec:Outlook}

In this paper, we have suggested and analyzed the application of Direct Bayesian Model Reduction (DBMR; Algorithm~\ref{alg:the_alg}, \cite{gerber2017toward}) for the identification of coherent sets and compared it with the classical approach based on truncated singular value decomposition (Algorithm \ref{alg:classical_approach_coherence}).
Both approaches perform a certain factorization of a matrix $A \approx BC$ into low rank matrices $B,C$, but maximize two different objectives, namely the `degree of coherence' as the sum of the leading singular values, corresponding to the minimization of the Frobenius norm $\| A - BC \|_F$, and the relaxed likelihood $\hat{\ell}$ from \eqref{eq:DBMR_objective}, connected to maximum likelihood estimation and minimization of the (generalized) Kullback--Leibler divergence $\KLD{A}{BC}$.
Therefore, on a broader scale, our contributions also establish connections between these two central minimization problems for matrix factorization.

The above-mentioned comparison is based on two central results, Theorems~\ref{thm:projection} and~\ref{thm:Frob_DBMR_bound}.
The first shows that the DBMR output $\Lambda = \lambda \Gamma$ can be written as a composition of the full model $P$ with an orthogonal projection $\Pi$, $\Lambda = P\Pi$.
While this is insightful in its own right, it also gives us the necessary tools to derive bounds on the degree of coherence of the reduced model~$\Lambda$ in Proposition~\ref{prop:sval_comp}.
The second theorem establishes a connection between the Frobenius norm distance and the Kullback--Leibler divergence mentioned above, which, to the best of our knowledge, is the first relationship of this kind.
For this purpose, we have derived certain Pinsker-type inequalities for the (weighted) $\ell^{2}$ norm in Appendix~\ref{section:Pinsker_l_2}, which might be of independent interest.

In our numerical experiments, DBMR was able to identify meaningful coherent sets.
It is well known that DBMR can get stuck in local maxima of its objective function, which we also observed (Example~\ref{example:piecewise_expanding_interval_map} with slight perturbation $\varepsilon=1$) even though we used a large number 100 of independent runs of DBMR.
The singular values, and thereby the degree of coherence of the corresponding reduced models was slightly inferior to the classical approach, which is hardly surprising since the classical approach optimizes precisely this objective.
However, the additional computations in Appendix~\ref{sec:DBMR_stat}, and in particular Figure~\ref{fig:DBMRtraj}, show that the two objectives are mostly aligned, backing up our theoretical findings from Theorem~\ref{thm:Frob_DBMR_bound}.

An important advantage of DBMR over the classical approach is that its low rank model $\Lambda = \lambda \Gamma$ is a product of a left stochastic matrix $\lambda$ and an affiliation matrix $\Gamma$, which has a clear probabilistic interpretation of a reduced transition matrix $\lambda$ that operates on compound states clustered by $\Gamma$.
This structure preservation is missed by the classical approach, where the `reduced transition matrix' $\tP_{\textup{red}}$ can have negative entries and the entries in each of its columns typically do not sum to one.

The connection between coherence---understood widely as subgroups of states subject to similar evolution---and matrix factorization remains an active field of research and various future directions are imaginable.
It is arguable whether the sum of the leading singular values of $\tP$ is a good measure for the `degree of coherence', see the discussion towards the end of section \ref{sec:numerical_examples}.
Establishing and optimizing other objectives, such as the DBMR objective $\hat{\ell}$ from \eqref{eq:DBMR_objective}, and analyzing connections between these objectives would deepen our theoretical understanding of both matrix factorization and the study of coherent structures.
The coherence problem~\eqref{eq:coherence_set} has a symmetry in the sense that matching partitions of \emph{both} input and output space are sought. In contrast, DBMR in its current form gives merely partitions of the input space. Future research could hence address the development of efficient `symmetrized' versions of DBMR.

%%%
\section*{Code availability}
\label{sec:code_availability}
MATLAB code for the Bayesian-Model-Reduction-Toolkit~\cite{gerber2017toward} is available at 
\begin{center}
\href{https://github.com/SusanneGerber/Bayesian-Model-Reduction-Toolkit}{https://github.com/SusanneGerber/Bayesian-Model-Reduction-Toolkit},
\end{center}
and the adaptation for this paper is available at 
\begin{center}
\href{https://github.com/RobertMaltePolzin}{https://github.com/RobertMaltePolzin/DBMR\_Coherence}.
\end{center}

\section*{Acknowledgment}

The authors thank Robin Chemnitz, Nicolas Gillis, Rupert Klein, Mattes Mollenhauer, Peter N\'evir, and Niklas Wulkow for discussions and helpful suggestions.
This research has been partially funded by the Deutsche Forschungsgemeinschaft (DFG, German Research Foun\-dation) through the grant CRC 1114 `Scaling Cascades in Complex Systems', Project Number 235221301, projects A01 ``Coupling a multiscale stochastic precipitation model to large scale atmospheric flow dynamics'', A02 ``Multiscale data and asymptotic model assimilation for atmospheric flows''  and A08 ``Characterization and prediction of quasi-stationary atmospheric states''.
The research of IK was funded by the DFG under Germany's Excellence Strategy (EXC-2046/1, project 390685689) through the project EF1-10 of the Berlin Mathematics Research Center MATH+.

%%% - Appendix - %%%
\appendix

%%%
\section{The coherence problem: From continuous to discrete space}
\label{app:cont_discr}

The coherence problem that we consider here has one of its roots in fluid dynamics. There, the (Lagrangian) evolution of passive tracers advected by the flow field is described by a flow $\phi^{t,\tau}:\Omega\to\Omega$ on some (mostly two or three-dimensional) spatial domain~$\Omega$. The flow is nonautonomous, and $\phi^{t,\tau}$ denotes the dynamical evolution from time $t$ to~$t+\tau$. Of particular interest are non-trivial subsets $A\subset \Omega$ that ``evolve coherently'' under the flow on some time interval $[t,t+\tau]$, meaning that sets $\phi^{t,s}(A)$, $0\le s\le \tau$, only experience minimal filamentation. In other words, the flow does not ``disperse'' the set~$A$.

The setting can be simplified by considering the flow at only discrete time instances. For our purposes, only two time instances are enough, say $t$ and~$t+\tau$. The mapping $T:=\phi^{t,\tau}$ does not need to leave any set in some state space invariant, hence the states at initial and final time can belong to different sets. Formally, the dynamics thus boils down to a mapping~$T:\Omega_1 \to \Omega_2$. We assume that $\Omega_1,\Omega_2$ are measurable spaces, $T$ is a measurable map, and suppress the underlying sigma algebras.

The coherence problem for the map $T$ can now be vaguely stated as the task to find nontrivial subsets $E\subset \Omega_1$ and $F\subset \Omega_2$ such that $T(E) \approx F$ and $E,F$ are relatively `simple' in terms of their geometry and balancedness. The latter can be made precise by requiring that the relation $T(E) \approx F$ persists under slight (random) perturbations. We refer to \cite{froyland2010transport} for further details. The sets $E,F$ are then called a \emph{finite time coherent pair}.
There are precise functional-analytic formulations of this problem available in~\cite{froyland2013analytic,Fro15}. Instead of taking this route, we can discretize the dynamics first and state the coherence problem in the discrete setting directly. This was done in \cite{froyland2010transport} to arrive at a problem that is numerically accessible via matrix analysis.
The same setting arises in situations where a precise observation of the state of the system is not possible and only quantized (discrete) observations are performed.

To get to the discrete setting, we subdivide the subsets $\Omega_{1}$ and $\Omega_{2}$ into collections of mutually disjoint partition elements $\left\{ B_{1},...,B_{n}\right\} $ and $\left\{ C_{1},...,C_{m}\right\} $, respectively. We assume that the initial state $X$ is an $\Omega_1$-valued random variable with law $\mu$; thus $\mu$ is a probability measure supported on~$\Omega_1$. Let $\nu$ denote the pushforward of $\mu$ by $T$, i.e., the law of~$Y \coloneqq T(X)$. We then define the \emph{transition matrix} $P\in\mathbb{R}^{m\times n}$ by
\begin{equation}\label{eq:approx_PF}
	P_{ij}=\frac{\mu(B_j \cap T^{-1}(C_i))}{\mu(B_j)} = \mathbb{P}\left[ Y \in C_i \mid X\in B_j \right].
\end{equation}
Note that $P$ is left stochastic, i.e., $P_{ij}\ge 0$ for all $i,j$, and~$\sum_{i=1}^m P_{ij}=1$ for all~$j$. We further define the discrete distributions at initial and final times by $p\in\mathbb{R}^n$ and $q\in\mathbb{R}^m$, respectively, with
\[
p_j = \mu(B_j), \quad j=1,\ldots,n,\qquad q_i = \nu(C_i), \quad i = 1,\ldots,m.
\]
It follows that the discrete initial distribution is mapped to the discrete final one by the transition matrix,
\begin{equation}
    \label{eq:qPp}
    q = Pp.
\end{equation}
We also assume that $p>0$ and $q>0$, componentwise. If not, the associated partition elements are removed from the sets~$B_j$ and~$C_i$, respectively (and the sets $\Omega_1$ and $\Omega_2$ are restricted accordingly).

The transition matrix $P$ together with the (initial) distribution $p$ characterize a one-step random transition
that jumps from some element $B_j$ of the initial partition to some element $C_i$ of the final partition. This way, if one were to consider a sequence of partitions with diameter converging to zero, the associated sequence of transition matrices $P$ would constitute a \emph{small random perturbation}~\cite{Kif86} of $T$, see~\cite{Fr98}. Thus, formulating the coherence problem for this discrete dynamics (see main text) will automatically deliver coherent sets that are robust with respect to small random perturbations of~$T$. 

Finally, we note that approximating dynamical properties of $T$ through the dis\-cretization~\eqref{eq:approx_PF} is attributed to Ulam~\cite{ulam1960collection}.

\section{Proof of Theorem~\ref{thm:projection}}
\label{app:proof_projection_theorem}

In this section, we present the proof of Theorem~\ref{thm:projection}.
Here, $\gamma\colon [n]\to[\lowrank]$ denotes the assignment corresponding to $\Gamma$ (cf.\ Definition~\ref{def:hard_affiliation_matrix}).

\begin{lemma}
\label{lemma:DBMR_projection_tpi_properties}
The matrix $\tpi \in \mathbb{R}^{n \times n}$ given by \eqref{equ:definition_Pi_and_tpi} is symmetric, $ \tpi^\top = \tpi$, and a projection, $\tpi^2 = \tpi$.	
Consequently, $\Pi$ is a $p^{-1}$-orthogonal projection (for $p^{-1}$-symmetry, see \eqref{eq:p symmetric}).
\end{lemma}
\begin{proof}    
	Symmetry of $\tpi$ follows directly from its definition in \eqref{equ:definition_Pi_and_tpi} and the fact that $\tpi_{ij} \neq 0$ if and only if~$\gamma(i) = \gamma(j)$.	
	In order to show that $\tpi$ is a projection we compute
	\[
	(\tpi^2)_{ij}
	=
	\sum_k \frac{\sqrt{p_i p_k} \, \delta_{\gamma(i)\gamma(k)}}{\sum_l p_l \delta_{\gamma(l)\gamma(k)}}  \frac{\sqrt{p_k p_j} \, \delta_{\gamma(k)\gamma(j)}}{\sum_{l'} p_{l'} \delta_{\gamma(l')\gamma(j)}} 
	\\
	=
	\frac{\sqrt{p_i p_j}\, \delta_{\gamma(i)\gamma(j)}}{\sum_l p_l \delta_{\gamma(l)\gamma(i)}} \sum_k \frac{p_k \, \delta_{\gamma(k)\gamma(j)}}{\sum_{l'} p_{l'} \delta_{\gamma(l')\gamma(j)}} 
	=
	\tpi_{ij},
	\]	
where we use the fact that nonzero terms in the first sum require~$\gamma(i) = \gamma(j) = \gamma(k)$.		
It follows that $\Pi$ is a $p^{-1}$-orthogonal projection:
\begin{align*}
\Pi^2 = (D_p^{1/2}\tpi D_p^{-1/2})^2
&=D_p^{1/2}\tpi^2 D_p^{-1/2} = \Pi, 
\\
\langle u, \Pi v \rangle_{p^{-1}}
=
\langle u, D_p^{-1/2} \tpi D_p^{-1/2} v \rangle_2
&=
\langle D_p^{1/2} \tpi D_p^{-1/2} u, D_p^{-1} v  \rangle_2
=
\langle \Pi u, v \rangle_{p^{-1}},
\qquad
u,v \in \mathbb{R}^n.
\end{align*}	
\end{proof}

\begin{proof}[Proof of Theorem~\ref{thm:projection}]
$\Pi$ is left stochastic by definition and a $p^{-1}$-orthogonal projection by Lemma~\ref{lemma:DBMR_projection_tpi_properties}.
Since $N_{ij} = S P_{ij}p_j$ and $\sum_{i=1}^m N_{ij} = Sp_j$ by~\eqref{eq:N_and_pP_estimators} as well as $\Gamma_{kj} = \delta_{k \gamma(j)}$ by Definition~\ref{def:hard_affiliation_matrix}, equation \eqref{eq:lamda_hat} implies
\begin{equation*}
    {\lambda}_{ik}=\frac{\sum_{{j=1}}^{n}{\Gamma}_{kj}N_{ij}}{\sum_{{i'=1}}^{m}\sum_{{j'=1}}^{n}{\Gamma}_{kj'}N_{i'j'}} = \frac{\sum_{j=1}^n\delta_{k \gamma(j)}P_{ij} p_j}{\sum_{j'=1}^m \delta_{k \gamma(j')}p_{j'}}.
\end{equation*}
Hence,
\begin{equation*}
    (\lambda \Gamma)_{il}
    =
    \sum_{k=1}^\lowrank \left( \frac{\sum_{j=1}^n\delta_{k \gamma(j)}P_{ij} p_j}{\sum_{j'=1}^m \delta_{k \gamma(j')}p_{j'}} \delta_{k \gamma(l)}\right)
    =
    \sum_{j=1}^n P_{ij} \frac{p_j \, \delta_{\gamma(j) \gamma(l)}}{\sum_{j'} p_{j'} \, \delta_{\gamma(l) \gamma(j')}}
    =
    (P \Pi)_{il},
\end{equation*}
proving $\lambda \Gamma = P \Pi$.
The eigenvector properties in \eqref{item:projection_thm_eigenvectors} follow from
\begin{align*}
    \left( \Pi a^{(k)}  \right)_i
    % = \sum_{j=1}^n \Pi_{ij} a_j^{(k)}
    &=
    \sum_{j=1}^n
    \frac{p_i \delta_{\gamma(i) \gamma(j)}}{\sum_{l} p_l \delta_{\gamma(l) \gamma(i)}}\, 
    p_j \delta_{\gamma(j)k}
    =
    p_i \delta_{\gamma(i)k}	 \sum_{j=1}^n
    \frac{p_j \, \delta_{\gamma(i) \gamma(j)}}{\sum_{l} p_l \delta_{\gamma(l) \gamma(i)}}
    =
    a^{(k)}_i,
    \\
    \left( \Pi p  \right)_i    
    &=
    \sum_{j=1}^n
    \frac{p_i \delta_{\gamma(i) \gamma(j)}}{\sum_{l} p_l \delta_{\gamma(l) \gamma(i)}}\, p_j
    =
    p_i \sum_{j=1}^n
    \frac{p_j \, \delta_{\gamma(i) \gamma(j)}}{\sum_{l} p_l \delta_{\gamma(l) \gamma(i)}}
    =
    p_i,
\end{align*}
where we use the fact that nonzero terms in the first sums  require~$\gamma(i) = \gamma(j)$. Since $\Pi$ is a projection, any of its eigenvalues can either be $0$ or $1$. For \eqref{eq:eigenspace}, it is hence sufficient to show that if $\langle b, a^{(k)}\rangle_{p^{-1}} =0$ for all $k \in \mathrm{Ran} \, \gamma$, then necessarily $\Pi b = 0$. This follows by noting that 
\begin{equation}
\label{eq:b orthogonal}
\langle b, a^{(k)} \rangle_{p^{-1}} = \sum_{i=1}^n b_i \delta_{\gamma(i) k}, 
\end{equation}
so that 
\begin{equation}
\nonumber
\sum_{j=1}^n \Pi_{ij} b_j = \frac{p_i \sum_{j=1}^n \delta_{\gamma(i)\gamma(j)}b_j}{\sum_{l=1}^n p_l \delta_{\gamma(l) \gamma(i)}} = 0
\end{equation}
if \eqref{eq:b orthogonal} is satisfied for all $k \in \mathrm{Ran} \, \gamma$.

The disjointness of the supports of the vectors $a^{(k)}$ follows from their definition (and $\Gamma$ being a hard affiliation matrix). Item is now a direct consequence of \eqref{item:projection_thm_eigenvectors}.
Item \eqref{item:projection_thm_DBMR_structure_preserving} follows directly from \eqref{item:projection_thm_eigenvectors} and $\Lambda = P \Pi$, $\Lambda p = P\Pi p = Pp = q$.
\end{proof}

\section{Pinsker's inequalities for the (weighted) $\ell^{2}$ norm}
\label{section:Pinsker_l_2}

The classical formulation of Pinsker's inequality \cite[Lemma~2.5]{tsybakov2009nonparametric} bounds the squared $\ell^{1}$ norm (or, equivalently, the squared total variation norm) of the difference of two probability vectors $u,v \in \bR^{m}$ by the Kullback--Leibler divergence,
\begin{equation}
\label{equ:ell_1_Pinsker}
\norm{u-v}_{1}^{2} \leq 2 \, \KLD{u}{v}.
\end{equation}
In this section, we derive a similar result for the (possibly weighted) $\ell^{2}$ norm in place of the $\ell^{1}$ norm.
While this can easily be achieved by applying the inequality $\norm{x}_{2} \leq \norm{x}_{1}$, $x\in\bR^{m}$, our aim is to obtain bounds that are as sharp as possible.
For this purpose, we use the concepts of balancedness and $q$-weighted balancedness from Definition \ref{def:balancedness} and state four versions of Pinsker's inequality in Proposition \ref{prop:l2_Pinsker} below which are particularly sharp in cases where either
\begin{enumerate}[(a)]
\item the difference $u-v$ has high balancedness, or
\item we require a bound of the $q$-weighted $\ell^{2}$ norm and $u-v$ has high $q$-balancedness, or
\item the vector $u$ has high balancedness and the difference $u-v$ is small, or
\item we require a bound of the $q$-weighted $\ell^{2}$ norm, $u$ has high $q$-balancedness and the difference $u-v$ is small,
\end{enumerate}
respectively.

\begin{lemma}
\label{lemma:technical_inequality_log_Taylor}
For any $x\in\bR$ with $x > -1$,% $\log(1+x) \leq x - x^{2}/2 + x^{3}/3$.
\[
\log(1+x)
\leq
x - \frac{x^{2}}{2} + \frac{x^{3}}{3}.
\]
\end{lemma}

\begin{proof}
The function $f(x) = x - x^{2}/2 + x^{3}/3 - \log(1+x)$ satisfies $f(0) = 0$ and, since
$(1+x) f'(x) = x^{3}$,
\[
f'(x)
\begin{cases}
\geq 0 & \text{if } x\geq 0,
\\
\leq 0 & \text{if } -1 < x < 0,
\end{cases}
\]
proving the claim by the fundamental theorem of calculus.
\end{proof}

The following result is utilized in Theorem~\ref{thm:Frob_DBMR_bound}.

\begin{proposition}
\label{prop:l2_Pinsker}
Let $u,v,q \in \bR^{m}$, $m\in\bN$, be probability vectors such that $q>0$ (componentwise) and
$\alpha(u,v) \coloneqq \frac{2}{3} \max_{i} \frac{|u_{i}-v_{i}|}{u_{i}} \in [0,\infty]$.\footnote{Recall the convention $\tfrac{0}{0} = 0$.}
Then
\begin{enumerate}[(a)]
\item
\label{item:l2_Pinsker_Hoelder}
$\displaystyle
\| u-v \|_{2}^{2}
\leq
\frac{2\, \KLD{u}{v}}{m\, \fB(u-v)},
% 2 \, \frac{\norm{u-v}_{\infty}}{\norm{u-v}_{1}} \, \KLD{u}{v},
$
\item
\label{item:weighted_l2_Pinsker_Hoelder}
$\displaystyle
\sum_{i=1}^{m} \frac{|u_{i} - v_{i}|^{2}}{q_{i}}
\leq
\frac{ 2\, \KLD{u}{v} }{\fB_{q}(u-v)},
$
\item
\label{item:l2_Pinsker_Taylor}
$\displaystyle
\| u-v \|_{2}^{2}
\leq
\frac{2 \, \KLD{u}{v} }{m\, \fB(u) (1 - \alpha(u,v))}$,
\hfill
%\hspace{10em}
if $\alpha(u,v) < 1$.
\item
\label{item:weighted_l2_Pinsker_Taylor}
$\displaystyle
\sum_{i=1}^{m} \frac{|u_{i} - v_{i}|^{2}}{q_{i}}
\leq
% \frac{2 \norm{\big(\frac{u_{i}}{q_{i}}\big)_{i}}_{\infty} }{1 - \alpha(u,v)}\, 
\frac{2 \, \KLD{u}{v}}{\fB_{q}(u) (1 - \alpha(u,v))}$,
\hfill
%\hspace{10em}
if $\alpha(u,v) < 1$.
\end{enumerate}	
\end{proposition}

\begin{proof}
Let $\eta \coloneqq v-u$.
If $\eta = 0$, there is nothing to show.
Otherwise, Hölder's inequality and Pinsker's inequality \eqref{equ:ell_1_Pinsker} yield
\begin{alignat*}{4}
\norm{\eta}_{2}^{2}
&\leq
\norm{\eta}_{\infty} \norm{\eta}_{1}
&&=
\frac{\norm{\eta}_{\infty}}{\norm{\eta}_{1}} \norm{\eta}_{1}^{2}
&&\leq
\frac{2\, \KLD{u}{v}}{m\, \fB(u-v)},
\\
\sum_{i=1}^{m} \frac{|\eta_{i}|^{2}}{q_{i}}
&\leq
\max_{i} \frac{|\eta_{i}|}{q_{i}} \norm{\eta}_{1}
&&=
\frac{\max_{i} \frac{|\eta_{i}|}{q_{i}}}{\norm{\eta}_{1}} \norm{\eta}_{1}^{2}
&&\leq
\frac{ 2\, \KLD{u}{v} }{\fB_{q}(u-v)},
\end{alignat*}
proving \eqref{item:l2_Pinsker_Hoelder} and \eqref{item:weighted_l2_Pinsker_Hoelder}.
In order to show \eqref{item:l2_Pinsker_Taylor} and \eqref{item:weighted_l2_Pinsker_Taylor}, first note that, if $v_{i} = 0$ and $u_{i} \neq 0$ for some $i$, then $\KLD{u}{v}=\infty$ and there is also nothing to show.
On the other hand, if $u_{i} = 0$ and $v_{i} \neq 0$ for some $i$, then the condition $\alpha(u,v) < 1$ is violated.
Finally, if $u_{i} = 0$ and $v_{i} = 0$ for some $i$, then we can reduce the dimension $m$ by one and work with $[m]\setminus\{i\}$ without changing any of the quantities involved.
Hence, we can assume $u_{i} \neq 0 \neq v_{i}$ for each $i\in [m]$.
Now, since $\sum_{i} \eta_{i} = 0$ and assuming that $\alpha(u,v) < 1$, Lemma \ref{lemma:technical_inequality_log_Taylor} implies
\begin{align*}
\KLD{u}{v}
&=
- \sum_{i=1}^{m} u_{i} \log\bigg(\frac{u_{i}+\eta_{i}}{u_{i}}\bigg)
\\
&\geq
\sum_{i=1}^{m} u_{i} \bigg( - \frac{\eta_{i}}{u_{i}}  + \frac{\eta_{i}^{2}}{2u_{i}^{2}} - \frac{\eta_{i}^{3}}{3u_{i}^{3}}\bigg)
\\
&=
\sum_{i=1}^{m} \frac{\eta_{i}^{2}}{2u_{i}} \bigg( 1 - \frac{2 \eta_{i}}{3u_{i}}\bigg)
\\
&\geq
\frac{1 - \alpha(u,v)}{2 \norm{u}_{\infty}}\, \|u-v\|_{2}^{2},
\end{align*}
proving \eqref{item:l2_Pinsker_Taylor}.
The proof of \eqref{item:weighted_l2_Pinsker_Taylor} goes similarly with a slight modification of the last step:
\[
\KLD{u}{v}
\geq
\sum_{i=1}^{m} \frac{\eta_{i}^{2}}{q_{i}} \frac{q_{i}}{2u_{i}} \bigg( 1 - \frac{2 \eta_{i}}{3u_{i}}\bigg)
\geq
\frac{1 - \alpha(u,v)}{2 \max_{i} \frac{u_{i}}{q_{i}}}\, \sum_{i=1}^{m} \frac{|u_{i} - v_{i}|^{2}}{q_{i}}.
\]
\end{proof}

\section{Courant--Fischer Theorem for Singular Values}

The common formulation of the Courant--Fischer (or min-max) theorem is stated for the eigenvalues of quadratic matrices.
However, in section~\ref{sec:coh_svd} we require a version for singular values of an arbitrary matrix, which is a well-known consequence.
Let us state and prove the precise version that we are going to use:
\begin{theorem}
\label{thm:Courant_Fischer_singular_values_version}
Let $n_{1},n_{2} \in \bN$ and $M \in \bR^{n_{1}\times n_{2}}$ be an arbitrary matrix with $\mathrm{rank} (M) \coloneqq s \leq \min(n_{1},n_{2})$ and ordered positive singular values $\sigma_{1} \geq \sigma_{2} \geq \cdots \geq \sigma_{s} > 0$.
Further, let $M = U \Sigma V^{\top}$ be a singular value decomposition of $M$ with $\Sigma = \textrm{diag}(\sigma_{1},\dots,\sigma_{s})$ and with $U \in \bR^{n_{1}\times s}$, $V \in \bR^{n_{2}\times s}$ having orthonormal columns.
Then, denoting by $\fW_{k}$ the set of $k$-dimensional subspaces of $\bR^{n_{2}}$,
\begin{align}
\label{equ:min_max_thm_singular_values}
\max_{W \in \fW_{k}} \min_{x \in W,\, \norm{x}_{2} = 1} \norm{Mx}_{2}
&=
\sigma_{k},
&&
k \in [s],
\\
\label{equ:min_max_thm_singular_values_our_version}
\max_{(e_{1},\dots,e_{\lowrank}) \text{ orthonormal}} \sum_{k=1}^{\lowrank} \norm{M e_{k}}_{2}
&=
\sum_{k=1}^{\lowrank} \sigma_{k},
&&
\lowrank \in [s],
\end{align}
where  \eqref{equ:min_max_thm_singular_values} is maximized by $W_{k}^{\star} = \mathrm{span}(V_{\bullet 1},\dots, V_{\bullet k})$. (with the inner minimization problem solved by $x = V_{\bullet k}$), while \eqref{equ:min_max_thm_singular_values_our_version} is maximized by the right singular vectors $e_{1} = V_{\bullet 1},\dots,e_{\lowrank} = V_{\bullet \lowrank}$.
\end{theorem}

\begin{proof}
First note that, for each $k \in [s]$, $\sigma_{k} = \sqrt{\lambda_{k}}$, where $\lambda_{1} \geq \cdots \geq \lambda_{s} > 0$ are the positive eigenvalues of $M^{\top} M$ and that
\[
\norm{Mx}_{2}^{2}
=
\langle Mx , Mx \rangle_{2}
=
\langle x , M^{\top} Mx \rangle_{2}
\leq 
\norm{x}_{2} \norm{M^{\top} Mx}_{2},
\]
with equality if and only if $x$ and $M^{\top} Mx$ are collinear, i.e.\ whenever $x$ is an eigenvector of $M^{\top} M$.
Hence, the inequality ``$\leq$'' in \eqref{equ:min_max_thm_singular_values} follows directly from the classical Courant--Fischer theorem \cite[Theorem 4.2.6]{HoJo13}.
To see the equality for $W = W_{k}^{\star}$, consider any $x = \sum_{j=1}^{k} \alpha_{j} V_{\bullet j} \in W_{k}^{\star}$, $\alpha \in \bR^{k}$, and observe that
\[
\norm{Mx}_{2}^{2}
=
\norm{ \sum_{j=1}^{k} \alpha_{j} \sigma_{j} U_{\bullet j} }_{2}^{2}
=
\sum_{j=1}^{k} (\alpha_{j} \sigma_{j})^{2}
\geq
\sigma_{k}^{2} \sum_{j=1}^{k} \alpha_{j}^{2}
=
\sigma_{k}^{2} \norm{x}_{2}^{2},
\]
with equality for $x = V_{\bullet k}$.

To see \eqref{equ:min_max_thm_singular_values_our_version}, note that a recursive application of \eqref{equ:min_max_thm_singular_values} implies the inequality ``$\leq$'' in \eqref{equ:min_max_thm_singular_values_our_version}, while the choice $e_{k} = V_{\bullet k}$ satisfies equality:
\[
\sum_{k=1}^{\lowrank} \norm{M V_{\bullet k}}_{2}
=
\sum_{k=1}^{\lowrank} \norm{\sigma_{k} U_{\bullet k}}_{2}
=
\sum_{k=1}^{\lowrank} \sigma_{k}.
\]
\end{proof}

%%%
\section{Collective analysis of DBMR runs}
\label{sec:DBMR_stat}

To support the the investigations in section~\ref{sec:numerical_examples}, here we provide a brief analysis involving multiple DBMR runs. Due to the non-concavity of the objective $\hat\ell$, DBMR can settle in different local maxima of $\hat\ell$, depending on the initialization~$\Gamma^{(0)}$.

Results of $100$ runs of DBMR are presented in Figures~\ref{fig:spectrum_100_runs_1}  (Example~\ref{example:1_three_coherent_sets}) and~\ref{fig:spectrum_100_runs_2} (Ex\-am\-ple~\ref{example:piecewise_expanding_interval_map}). The top row of images compare the first $5$ singular values of $\tLam$ for every converged pair $(\lambda,\Gamma)$ (blue crosses) with the singular values of the full model~$\tP$ (red dots).
Since $\mathrm{rank}\,\tLam = 3$, we do not show its 4th and 5th singular values, which are always zero. The bottom row of images present, for every converged pair $(\lambda,\Gamma)$, the DBMR objective $\hat\ell(\lambda,\Gamma)$ versus the degree of coherence~$\cC_3(\Lambda)$. In the bottom panels a colorbar indicates the number of solutions in the histogram bins. In each of the bottom rows, one panel is showing the results as a scatter plot instead of a histogram, for better visual experience.

\begin{figure}[htbp]
	\centering
	\begin{subfigure}[b]{0.32\textwidth}
		\centering
		\includegraphics[width=\textwidth]{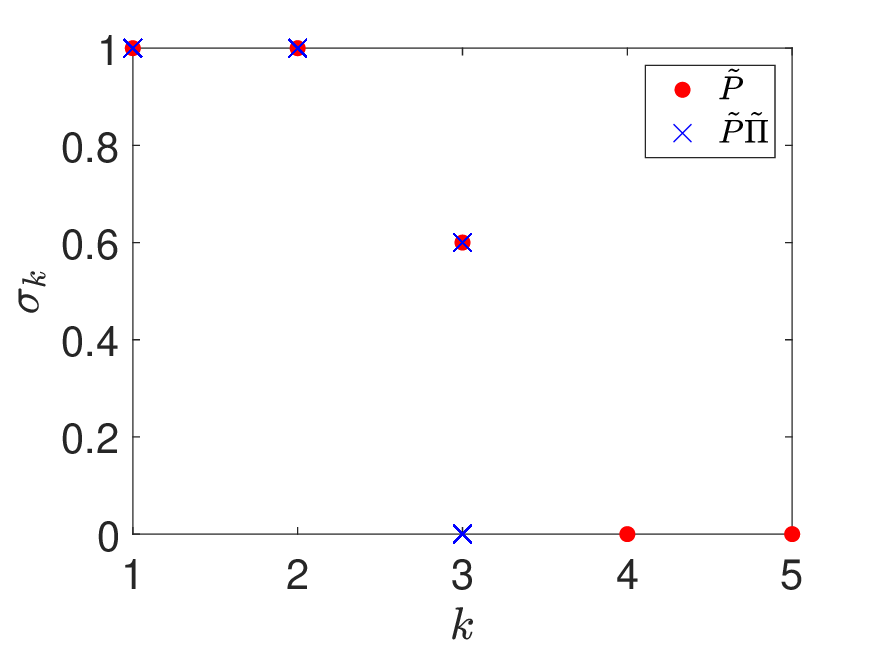}
	\end{subfigure}
	\begin{subfigure}[b]{0.32\textwidth}
		\centering
		\includegraphics[width=\textwidth]{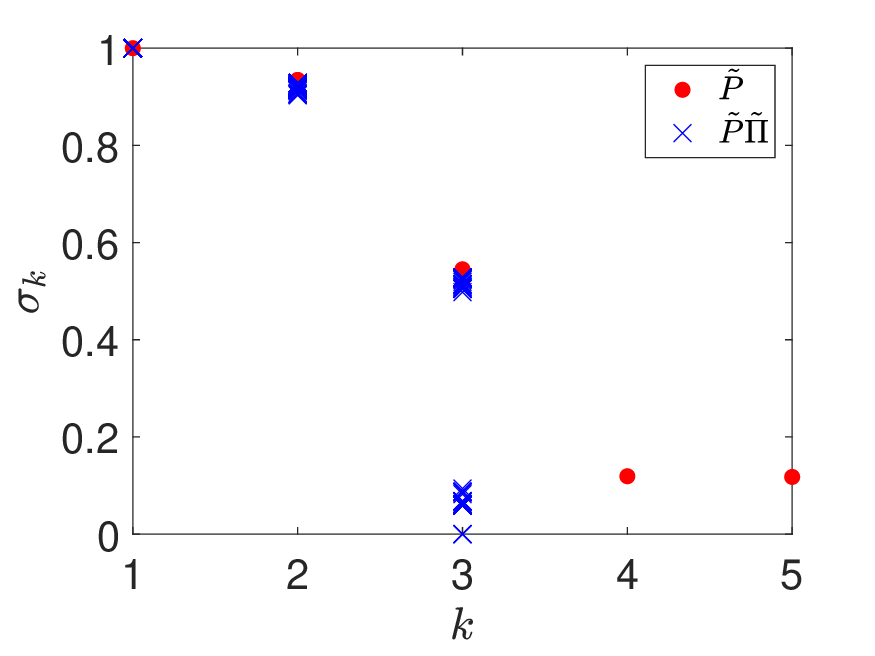}
	\end{subfigure}     
        \begin{subfigure}[b]{0.32\textwidth}
		\centering
		\includegraphics[width=\textwidth]{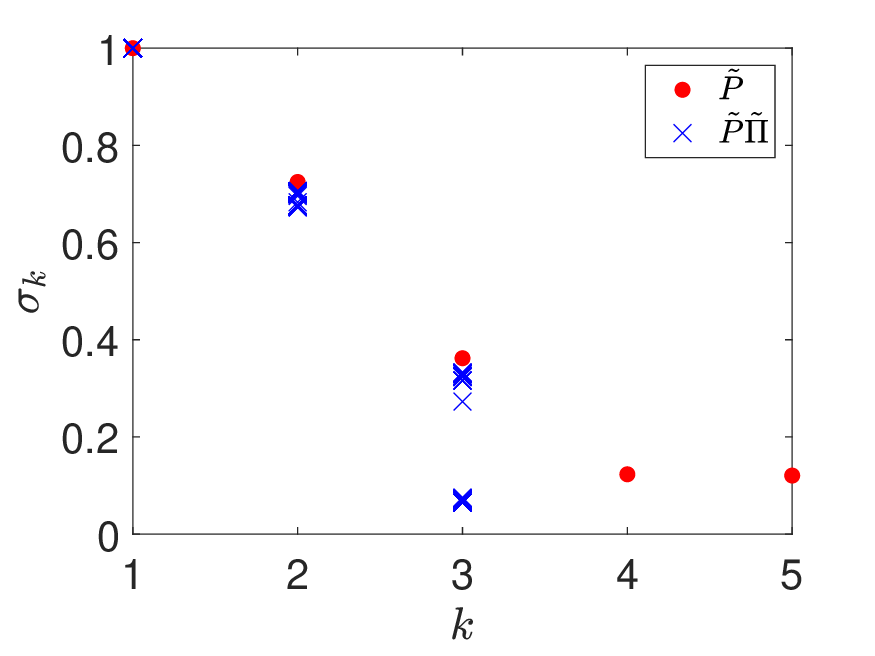}
	\end{subfigure}
    \vfill
	\begin{subfigure}[b]{0.32\textwidth}
		\centering
		\includegraphics[width=\textwidth]{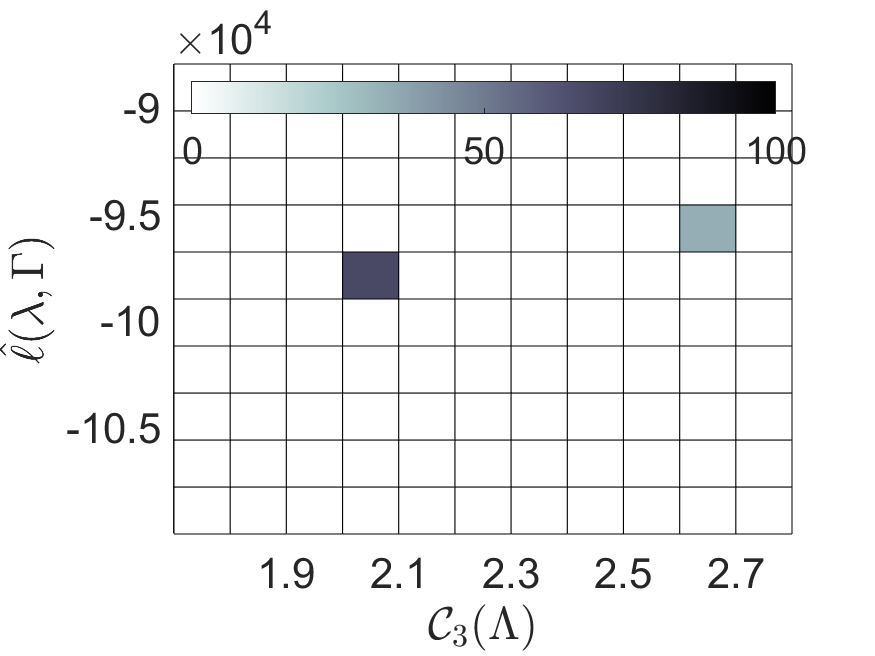}
	\end{subfigure} 
        \begin{subfigure}[b]{0.32\textwidth}
		\centering
		\includegraphics[width=\textwidth]{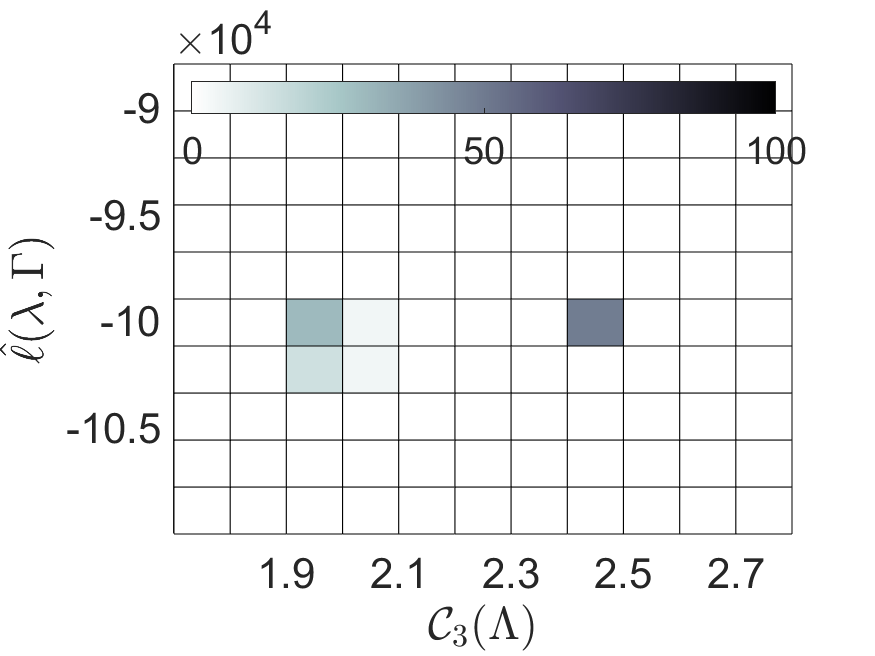}
	\end{subfigure}
	\begin{subfigure}[b]{0.32\textwidth}
		\centering
		\includegraphics[width=\textwidth]{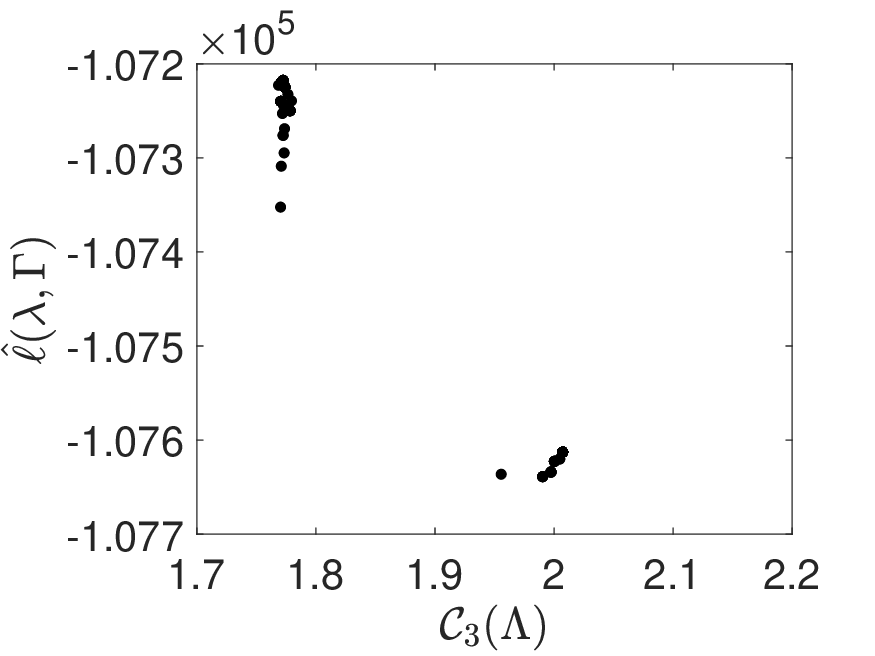}
	\end{subfigure}
	\caption{\textit{Top:} First $5$ singular values of $\tP$ and $\tLam = \tP\tpi$ for the transition matrix of Example~\ref{example:1_three_coherent_sets}. \textit{Bottom:} Likelihood bound $\hat\ell(\lambda,\Gamma)$ in dependence of the degree of coherence~$\cC_{3}(\Lambda)$. Results are shown for $100$ runs of DBMR with $\lowrank=3$ latent states. \textit{Left:} unperturbed; \textit{center:} slightly perturbed \textit{right:} strongly perturbed.
 }
	\label{fig:spectrum_100_runs_1}
\end{figure}

\begin{figure}[htbp]
	\centering
	\begin{subfigure}[b]{0.32\textwidth}
		\centering
		\includegraphics[width=\textwidth]{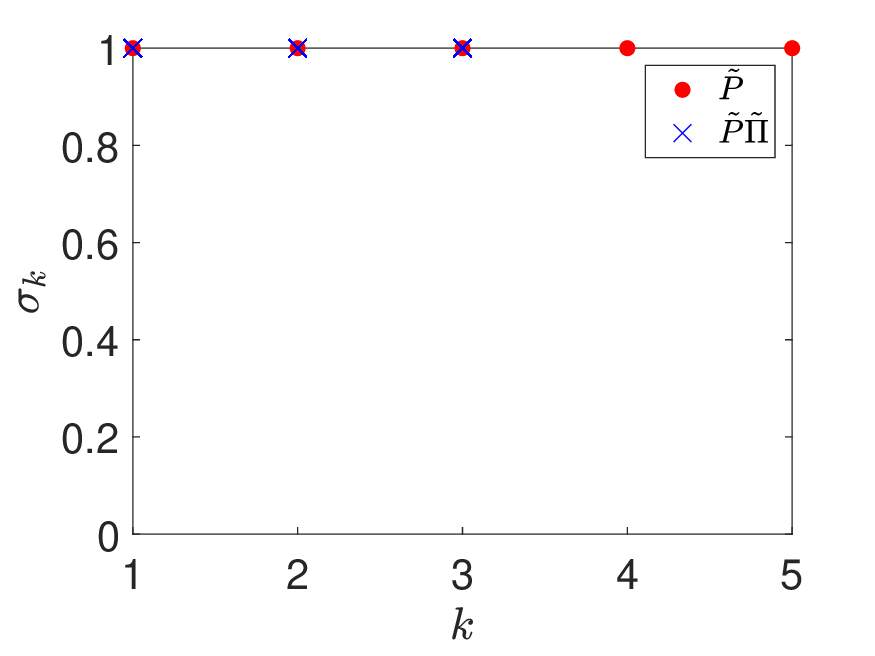}
	\end{subfigure}
 \hfill
    \begin{subfigure}[b]{0.32\textwidth}
		\centering
		\includegraphics[width=\textwidth]{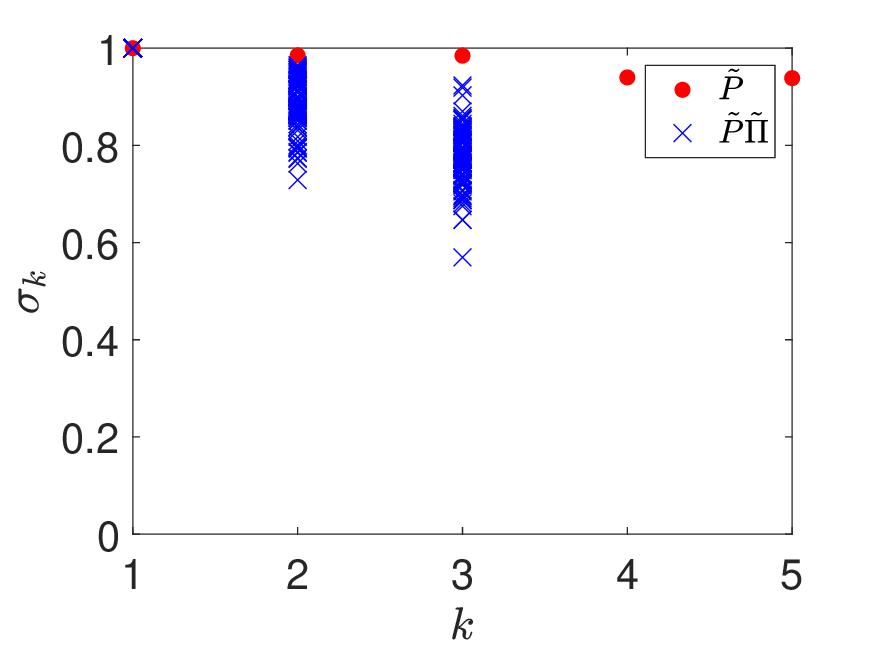}
	\end{subfigure}
 \hfill
    \begin{subfigure}[b]{0.32\textwidth}
		\centering
		\includegraphics[width=\textwidth]{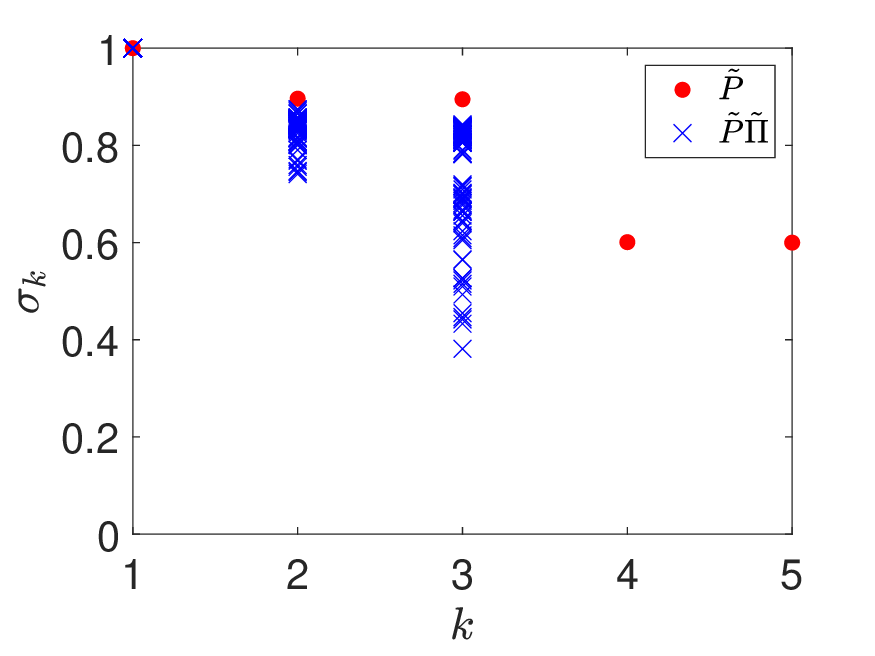}
	\end{subfigure}
 \vfill
	\begin{subfigure}[b]{0.32\textwidth}
		\centering
		\includegraphics[width=\textwidth]{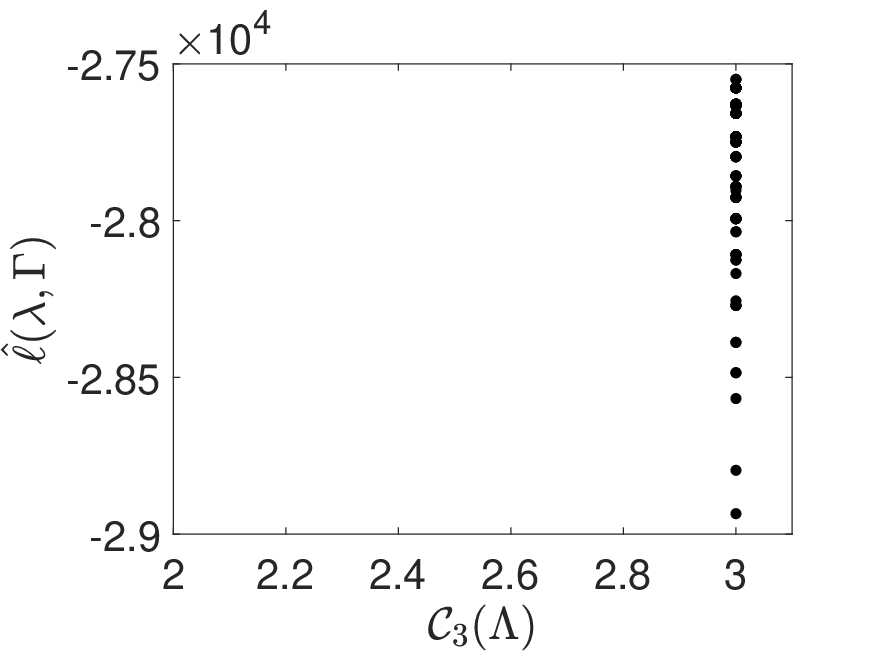}
	\end{subfigure} 
    \hfill    
	\begin{subfigure}[b]{0.32\textwidth}
		\centering
		\includegraphics[width=\textwidth]{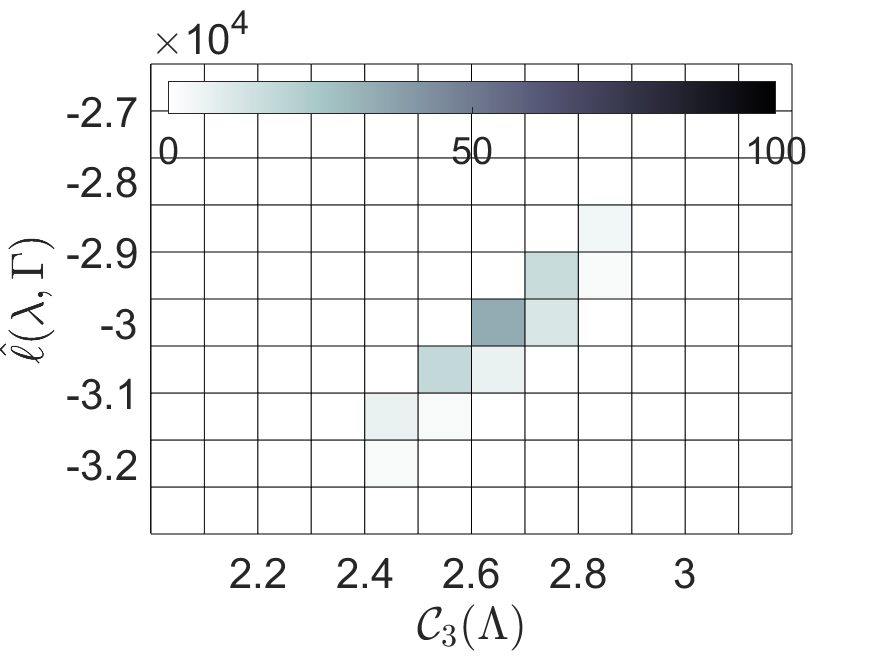}
	\end{subfigure} 
 \hfill
	\begin{subfigure}[b]{0.32\textwidth}
		\centering
		\includegraphics[width=\textwidth]{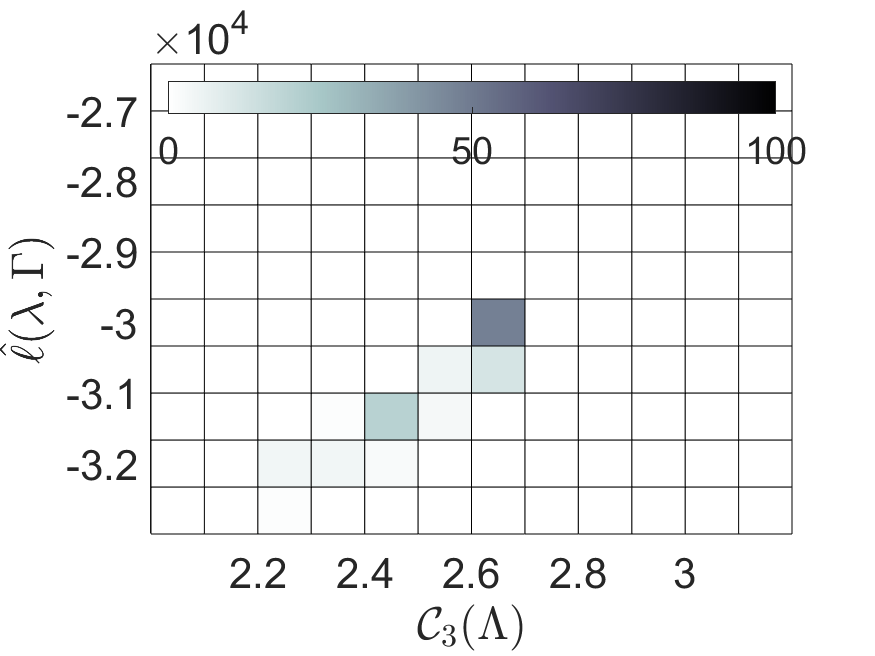}
	\end{subfigure} 
	\caption{\textit{Top:} First $5$ singular values of $\tP$ and $\tLam = \tP\tpi$ for the transition matrix of Example~\ref{example:piecewise_expanding_interval_map}. \textit{Bottom:} Likelihood bound $\hat\ell(\lambda,\Gamma)$ in dependence of the degree of coherence~$\cC_{3}(\Lambda)$. Results are shown for $100$ runs of DBMR with $\lowrank=3$ latent states. \textit{Left:} unperturbed; \textit{center:} slightly perturbed \textit{right:} strongly perturbed.}
	\label{fig:spectrum_100_runs_2}
\end{figure}

We observe that in the unperturbed case DBMR recovers $\sigma_2=1$ perfectly for all runs, but $\sigma_3=0.6$ is recovered only in $60\%$ of the runs, and in the remaining runs it converges to a (degenerate) model with effectively $\lowrank=2$ latent states; i.e., $\sigma_3(\tLam)=0$. The bottom left panel of Figure~\ref{fig:spectrum_100_runs_1} shows that the degenerate results are suboptimal, in the sense that the corresponding iterations get stuck in a suboptimal local maximum. As the perturbation in the data is increased, $\sigma_2$ is close to optimal and we start to get results between the previous two extreme cases of $\sigma_3(\tLam)$, and the associated degrees of coherence $\cC_3(\tLam)$ spread out from the previous two values somewhat.
For large perturbation, this process continues, and we see clustering of the objectives $(\cC,\hat\ell)$ around $(1.75, -1.072\cdot 10^5)$ and $(2,-1.076\cdot 10^5)$. That is, the objectives work against one another.
We thus note that while increasing the perturbation improved the success rate of the DBMR runs finding a global maximum, it also turned the harmonious objectives into mildly conflicting ones.

\begin{figure}[htb]
    \centering
    \begin{subfigure}[b]{0.49\textwidth}
		\centering
		\includegraphics[width=\textwidth]{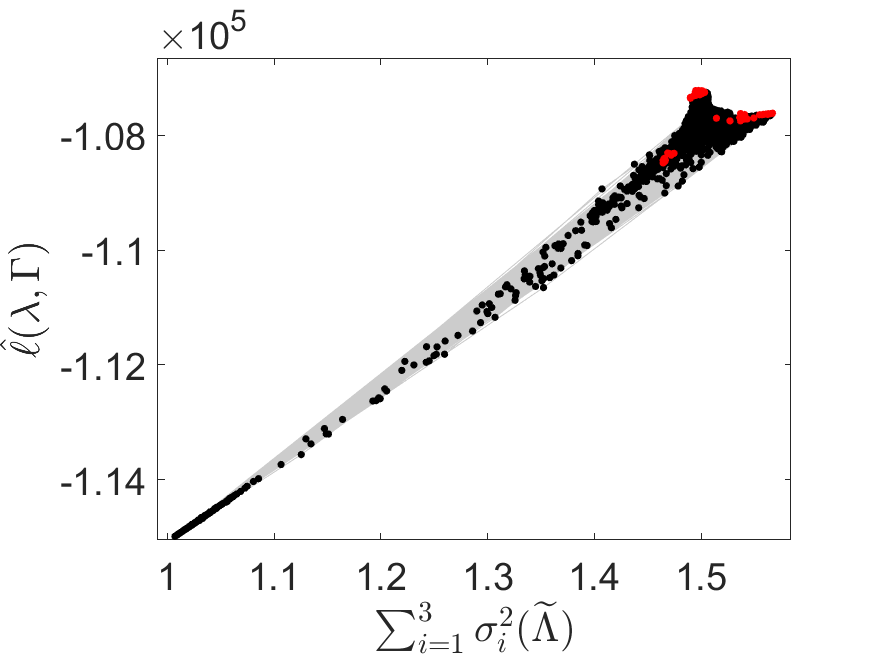}
	\end{subfigure}
 \hfill
    \begin{subfigure}[b]{0.49\textwidth}
		\centering
		\includegraphics[width=\textwidth]{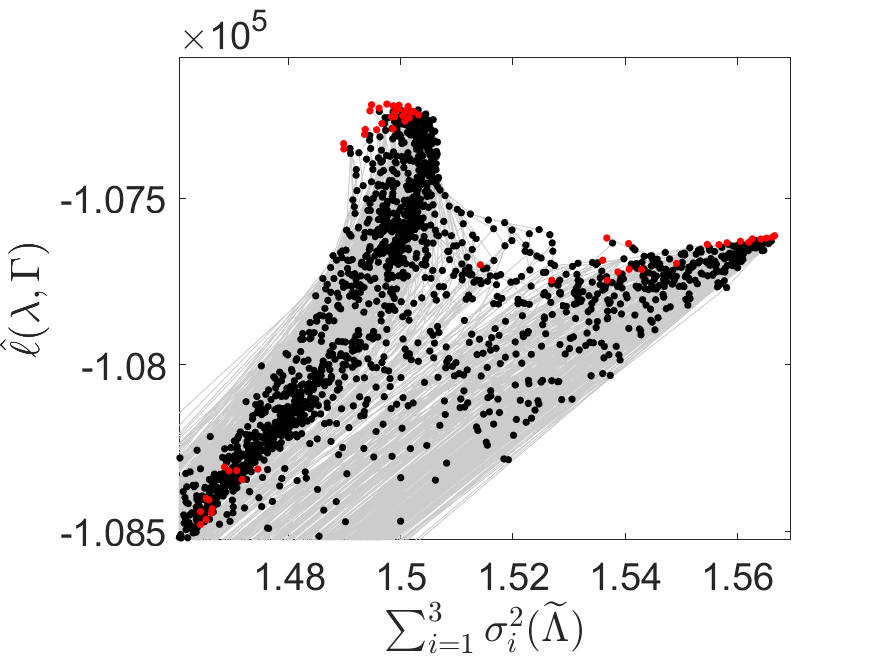}
	\end{subfigure}
    \caption{The two objectives $\hat\ell(\lambda,\Gamma)$ and $\|\tLam\|_F^2$ for all iterates of 1000 randomly initialized DBMR runs in the strongly perturbed case of Example~\ref{example:1_three_coherent_sets}. Grey lines connect data points for successive DBMR iterates, red dots indicate the endpoints of the 1000 runs (local optima of the DBMR objective~\eqref{eq:DBMR_objective}). The right panel is a close-up of the region with the highest values of both objectives, corresponding to the top right corner of the left panel.}
    \label{fig:DBMRtraj}
\end{figure}

Figure~\ref{fig:spectrum_100_runs_2} (Example~\ref{example:piecewise_expanding_interval_map}) shows that for the unperturbed case all DBMR runs converge to coherence-optimal partitions.
However, for the perturbed cases, many local optima trap the runs.
As discussed in Example~\ref{example:piecewise_expanding_interval_map}, we attribute these local minima to the many coherent sets of different sizes present in this system. In this example we observe no conflict between the optimization criteria~$\cC$ and~$\hat\ell$.

To get a more comprehensive picture for Example~\ref{example:1_three_coherent_sets} with large perturbation, where the conflict between the optimization criteria arises, we consider 1000 new DBMR runs. Every run is initialized, as before, with an affiliation matrix $\Gamma^{(0)}$ of which every column is a uniform i.i.d.\ sample of one of the three standard unit vectors.
This time, instead of considering only the converged DBMR solutions, we keep all iterates of every run, that is, whole ``DBMR trajectories'': $6674$ pairs of $(\lambda,\Gamma)$ in total.
For all these, we depict $\hat\ell(\lambda,\Gamma)$ and $\sum_{i=1}^3\sigma_i(\tLam)^2 = \|\tLam\|_F^2$, where the latter equality is due to~$\mathrm{rank}(\tLam) \le \lowrank = 3$.
By Corollary~\ref{cor:relating_objectives} this is an equivalent objective to~$\|\tP - \tLam\|_F^2$, since $\|\tP\|_F^2$ depends only on the data and not on DBMR iterates. The results are shown in the left-hand panel of Figure~\ref{fig:DBMRtraj}, with a close-up on the region with the conflicting optima on the right. In the left-hand panel all initial points of DBMR runs satisfy~$\|\tLam\|_F^2\le 1.1$. We observe that, although there is some ``spread'' during the DBMR iterations and in particular in the local DBMR optima, the correlation between the objectives is quite high for this set of matrices.
We also observe a third cluster of local optima around $\hat{\ell}(\lambda,\Gamma) \approx -1.084 \cdot 10^5$ which was not found by the previous 100 runs, cf.\ 	Figure~\ref{fig:spectrum_100_runs_1} (bottom right panel).
It seems to correspond to degenerate local minima essentially belonging to a coherent 2-partition, as can be seen by the corresponding DBMR transition matrix $\Lambda$, depicted in Figure~\ref{fig:modelcase_otheropt} (right).
This figure illustrates three DBMR transition matrices $\Lambda$, each corresponding to one DBMR optimum from the three red clusters in Figure~\ref{fig:DBMRtraj}.
Its left panel is identical to the bottom right panel Figure~\ref{fig:CSI_1} with the highest DBMR objective value, while the center panel with objective value around~$-1.076\cdot 10^5$ corresponds to a clustering that is closer to the coherence-optimal partition (Figure~\ref{fig:CSI_1} middle right).

\begin{figure}
    \centering
 \begin{subfigure}[b]{0.32\textwidth}
		\centering
		\includegraphics[width=\textwidth]{Capital_Lambda_d_10.eps}
	\end{subfigure}
    \hfill
    \begin{subfigure}[b]{0.32\textwidth}
		\centering
		\includegraphics[width=\textwidth]{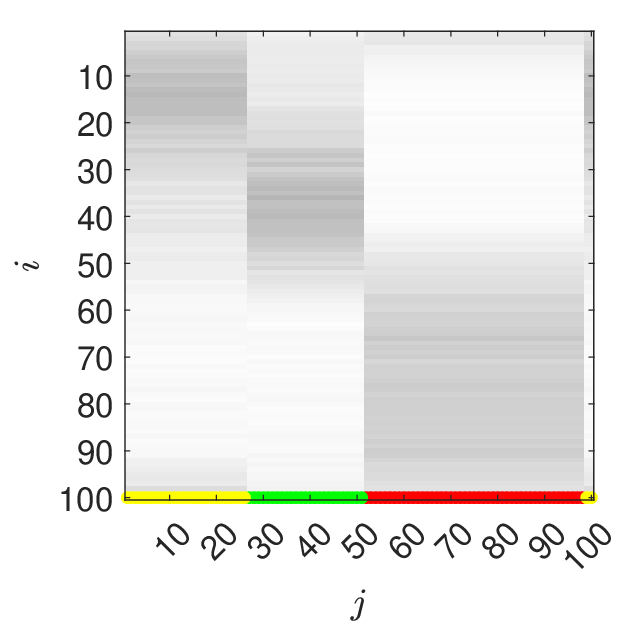}
	\end{subfigure}
    \hfill
    \begin{subfigure}[b]{0.32\textwidth}
		\centering
		\includegraphics[width=\textwidth]{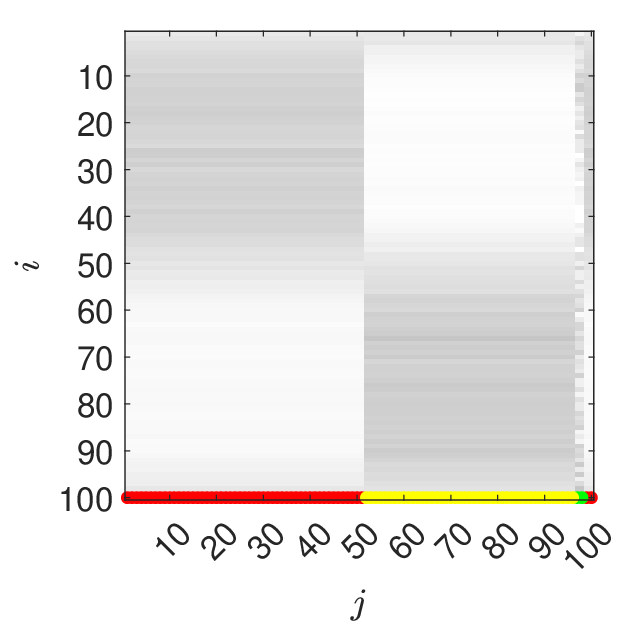}
	\end{subfigure}
    \caption{Three DBMR output transition matrices $\Lambda = \lambda\Gamma$ for optimal (left: $\hat\ell(\lambda,\Gamma) \approx -1.072\cdot 10^5$) and suboptimal (center: $\hat\ell(\lambda,\Gamma) \approx -1.076 \cdot 10^5$, right: $\hat\ell(\lambda,\Gamma) \approx -1.084 \cdot 10^5$) local maxima in Example~\ref{example:1_three_coherent_sets} with large perturbation, each corresponding to one of the three red clusters in Figure~\ref{fig:DBMRtraj}.
    The coloring of the bottom line indicates the partition given by the affiliation matrix~$\Gamma$.}
    \label{fig:modelcase_otheropt}
\end{figure}

%%%%%%%%%%%%%%%%%%%%%%%%%%%%%%%%%%%%%%% Bib %%%%%
\newcommand{\etalchar}[1]{$^{#1}$}

%\bibliographystyle{myalpha}

%\small
% \bibliographystyle{unsrt}
%\nocite{*}
%\bibliography{bibliography}

\begin{thebibliography}{dLGDLR08}

\bibitem[ABB{\etalchar{+}}17]{aref2017}
H.~Aref, J.~R. Blake, M.~Budi{\v{s}}i{\'c}, S.~S. Cardoso, J.~H. Cartwright,
  H.~J. Clercx, K.~El~Omari, U.~Feudel, R.~Golestanian, E.~Gouillart, G.~F. van
  Heijst, T.~S. Krasnopolskaya, Y.~L. Guer, R.~S. MacKay, V.~V. Meleshko,
  G.~Metcalfe, I.~Mezić, A.~P.~S. de~Moura, O.~Piro, M.~F.~M. Speetjens,
  R.~Sturman, J.-L. Thiffeault, and I.~Tuval.
\newblock Frontiers of chaotic advection.
\newblock {\em Reviews of Modern Physics}, 89(2):025007, 2017.

\bibitem[Are84]{aref1984}
H.~Aref.
\newblock Stirring by chaotic advection.
\newblock {\em Journal of Fluid Mechanics}, 143:1--21, 1984.

\bibitem[BK17]{banisch2017understanding}
R.~Banisch and P.~Koltai.
\newblock Understanding the geometry of transport: diffusion maps for
  {L}agrangian trajectory data unravel coherent sets.
\newblock {\em Chaos}, 27(3):035804, 16, 2017.

\bibitem[BKPG19]{BaKoPG19}
R.~Banisch, P.~Koltai, and K.~Padberg-Gehle.
\newblock Network measures of mixing.
\newblock {\em Chaos: An Interdisciplinary Journal of Nonlinear Science},
  29(6):063125, 2019.

\bibitem[Den17]{denner2017coherent}
A.~Denner.
\newblock {\em Coherent structures and transfer operators}.
\newblock PhD thesis, Technische Universit{\"a}t M{\"u}nchen, 2017.

\bibitem[dLGDLR08]{de2008shrinking}
D.~M. de~Lachapelle, D.~Gfeller, and P.~De~Los~Rios.
\newblock Shrinking matrices while preserving their eigenpairs with application
  to the spectral coarse graining of graphs.
\newblock {\em Submitted to SIAM Journal on Matrix Analysis and Applications},
  2008.

\bibitem[DLP06]{ding2006nonnegative}
C.~Ding, T.~Li, and W.~Peng.
\newblock Nonnegative matrix factorization and probabilistic latent semantic
  indexing: {E}quivalence chi-square statistic, and a hybrid method.
\newblock In {\em AAAI}, volume~42, pages 137--43, 2006.

\bibitem[DLPP06]{ding2006orthogonal}
C.~Ding, T.~Li, W.~Peng, and H.~Park.
\newblock Orthogonal nonnegative matrix tri-factorizations for clustering.
\newblock In {\em Proceedings of the 12th ACM SIGKDD international conference
  on Knowledge discovery and data mining}, pages 126--135, 2006.

\bibitem[DW05]{deuflhard2005robust}
P.~Deuflhard and M.~Weber.
\newblock Robust {P}erron cluster analysis in conformation dynamics.
\newblock {\em Linear Algebra Appl.}, 398:161--184, 2005.

\bibitem[FBD09]{fevotte2009nonnegative}
C.~F{\'e}votte, N.~Bertin, and J.-L. Durrieu.
\newblock Nonnegative matrix factorization with the {I}takura-{S}aito
  divergence: With application to music analysis.
\newblock {\em Neural computation}, 21(3):793--830, 2009.

\bibitem[FHRVS15]{froyland2015studying}
G.~Froyland, C.~Horenkamp, V.~Rossi, and E.~Van~Sebille.
\newblock Studying an {A}gulhas ring's long-term pathway and decay with
  finite-time coherent sets.
\newblock {\em Chaos: An Interdisciplinary Journal of Nonlinear Science},
  25(8), 2015.

\bibitem[FI11]{fevotte2011algorithms}
C.~F{\'e}votte and J.~Idier.
\newblock Algorithms for nonnegative matrix factorization with the
  $\beta$-divergence.
\newblock {\em Neural computation}, 23(9):2421--2456, 2011.

\bibitem[FLS10]{FrLlSa10}
G.~Froyland, S.~Lloyd, and N.~Santitissadeekorn.
\newblock Coherent sets for nonautonomous dynamical systems.
\newblock {\em Phys. D}, 239(16):1527--1541, 2010.

\bibitem[FP09]{froyland_padberg_09}
G.~Froyland and K.~Padberg.
\newblock Almost-invariant sets and invariant manifolds -- connecting
  probabilistic and geometric descriptions of coherent structures in flows.
\newblock {\em Phys. D}, 238:1507--1523, 2009.

\bibitem[FPG14]{froyland2014almost}
G.~Froyland and K.~Padberg-Gehle.
\newblock Almost-invariant and finite-time coherent sets: directionality,
  duration, and diffusion.
\newblock In {\em Ergodic theory, open dynamics, and coherent structures},
  pages 171--216. Springer, 2014.

\bibitem[Fro98]{Fr98}
G.~Froyland.
\newblock Approximating physical invariant measures of mixing dynamical systems
  in higher dimensions.
\newblock {\em Nonlinear Anal.}, 32(7):831--860, 1998.

\bibitem[Fro13]{froyland2013analytic}
G.~Froyland.
\newblock An analytic framework for identifying finite-time coherent sets in
  time-dependent dynamical systems.
\newblock {\em Phys. D}, 250:1--19, 2013.

\bibitem[Fro15]{Fro15}
G.~Froyland.
\newblock Dynamic isoperimetry and the geometry of {L}agrangian coherent
  structures.
\newblock {\em Nonlinearity}, 28(10):3587--3622, 2015.

\bibitem[FRS19]{froyland2019sparse}
G.~Froyland, C.~P. Rock, and K.~Sakellariou.
\newblock Sparse eigenbasis approximation: {M}ultiple feature extraction across
  spatiotemporal scales with application to coherent set identification.
\newblock {\em Commun. Nonlinear Sci. Numer. Simul.}, 77:81--107, 2019.

\bibitem[FSM10]{froyland2010transport}
G.~Froyland, N.~Santitissadeekorn, and A.~Monahan.
\newblock Transport in time-dependent dynamical systems: finite-time coherent
  sets.
\newblock {\em Chaos}, 20(4):043116, 10, 2010.

\bibitem[GH17]{gerber2017toward}
S.~Gerber and I.~Horenko.
\newblock Toward a direct and scalable identification of reduced models for
  categorical processes.
\newblock {\em Proceedings of the National Academy of Sciences},
  114(19):4863--4868, 2017.

\bibitem[Gil20]{gillis2020nonnegative}
N.~Gillis.
\newblock {\em Nonnegative Matrix Factorization}.
\newblock Society for Industrial and Applied Mathematics, 2020.

\bibitem[GK78]{gohberg1978introduction}
I.~Gohberg and M.~G. Kre{\u\i}n.
\newblock {\em Introduction to the theory of linear nonselfadjoint operators},
  volume~18.
\newblock American Mathematical Soc., 1978.

\bibitem[GONH18]{gerber2018scalable}
S.~Gerber, S.~Olsson, F.~No{\'e}, and I.~Horenko.
\newblock A scalable approach to the computation of invariant measures for
  high-dimensional {M}arkovian systems.
\newblock {\em Scientific reports}, 8(1):1796, 2018.

\bibitem[HBV13]{haller2013coherent}
G.~Haller and F.~J. Beron-Vera.
\newblock Coherent {L}agrangian vortices: The black holes of turbulence.
\newblock {\em Journal of Fluid Mechanics}, 731:R4, 2013.

\bibitem[HE15]{HsEu15}
T.~Hsing and R.~Eubank.
\newblock {\em Theoretical foundations of functional data analysis, with an
  introduction to linear operators}.
\newblock John Wiley \& Sons, 2015.

\bibitem[Hec98]{heckerman1998tutorial}
D.~Heckerman.
\newblock {\em A tutorial on learning with {B}ayesian networks}.
\newblock Springer, 1998.

\bibitem[HJ13]{HoJo13}
R.~A. Horn and C.~R. Johnson.
\newblock {\em Matrix analysis}.
\newblock Cambridge University Press, 2 edition, 2013.

\bibitem[Hof99]{hofmann1999probabilistic}
T.~Hofmann.
\newblock Probabilistic latent semantic indexing.
\newblock In {\em Proceedings of the 22nd annual international ACM SIGIR
  conference on Research and development in information retrieval}, pages
  50--57, 1999.

\bibitem[Hof01]{hofmann2001unsupervised}
T.~Hofmann.
\newblock Unsupervised learning by probabilistic latent semantic analysis.
\newblock {\em Machine learning}, 42(1-2):177--196, 2001.

\bibitem[Hot36]{Hot36}
H.~Hotelling.
\newblock Relations between two sets of variates.
\newblock {\em Biometrika}, 28(3-4):321--377, 12 1936.

\bibitem[HP98]{haller1998finite}
G.~Haller and A.~C. Poje.
\newblock Finite time transport in aperiodic flows.
\newblock {\em Phys. D}, 119(3):352--380, 1998.

\bibitem[HS05]{huisinga2005metastability}
W.~Huisinga and B.~Schmidt.
\newblock Metastability and dominant eigenvalues of transfer operators.
\newblock {\em Lecture Notes in Computational Science and Engineering}, 49:167,
  2005.

\bibitem[KCS16]{koltai2016metastability}
P.~Koltai, G.~Ciccotti, and C.~Sch{\"u}tte.
\newblock On metastability and {M}arkov state models for non-stationary
  molecular dynamics.
\newblock {\em The Journal of Chemical Physics}, 145(17):174103, 2016.

\bibitem[KHMN19]{klus2019kernel}
S.~Klus, B.~E. Husic, M.~Mollenhauer, and F.~No\'{e}.
\newblock Kernel methods for detecting coherent structures in dynamical data.
\newblock {\em Chaos}, 29(12):123112, 15, 2019.

\bibitem[Kif86]{Kif86}
Y.~Kifer.
\newblock General random perturbations of hyperbolic and expanding
  transformations.
\newblock {\em Journal D'Analyse Math\'ematique}, 47:111--150, 1986.

\bibitem[KWNS18]{KWNS18}
P.~Koltai, H.~Wu, F.~No\'e, and C.~Sch\"utte.
\newblock Optimal data-driven estimation of generalized {M}arkov state models
  for non-equilibrium dynamics.
\newblock {\em Computation}, 6(1), 2018.

\bibitem[LD14]{li2014nonnegative}
T.~Li and C.~Ding.
\newblock Nonnegative matrix factorizations for clustering: A survey.
\newblock In C.~C. Aggarwal and C.~K. Reddy, editors, {\em Data Clustering:
  Algorithms and Applications}, Data Mining and Knowledge Discovery Series,
  chapter~7, pages 149--176. Chapman and Hall/CRC, 1. edition, 2014.

\bibitem[LS99]{lee1999learning}
D.~D. Lee and H.~S. Seung.
\newblock Learning the parts of objects by non-negative matrix factorization.
\newblock {\em Nature}, 401(6755):788--791, 1999.

\bibitem[LS00]{lee2000algorithms}
D.~Lee and H.~S. Seung.
\newblock Algorithms for non-negative matrix factorization.
\newblock {\em Advances in neural information processing systems}, 13, 2000.

\bibitem[LSZL20]{lu2020community}
H.~Lu, X.~Sang, Q.~Zhao, and J.~Lu.
\newblock Community detection algorithm based on nonnegative matrix
  factorization and pairwise constraints.
\newblock {\em Physica A: Statistical Mechanics and its Applications},
  545:123491, 2020.

\bibitem[OBA24]{ortiz2022community}
M.~Ortiz-Bouza and S.~Aviyente.
\newblock Community detection in multiplex networks based on orthogonal
  nonnegative matrix tri-factorization.
\newblock {\em IEEE Access}, 12:6423--6436, 2024.

\bibitem[OW05]{ordentlich2005distribution}
E.~Ordentlich and M.~J. Weinberger.
\newblock A distribution dependent refinement of {P}insker's inequality.
\newblock {\em IEEE Transactions on Information Theory}, 51(5):1836--1840,
  2005.

\bibitem[PGAG14]{pompili2014two}
F.~Pompili, N.~Gillis, P.-A. Absil, and F.~Glineur.
\newblock Two algorithms for orthogonal nonnegative matrix factorization with
  application to clustering.
\newblock {\em Neurocomputing}, 141:15--25, 2014.

\bibitem[RKLW90]{romkedar_etal_1990}
V.~Rom-Kedar, A.~Leonard, and S.~Wiggins.
\newblock An analytical study of transport, mixing and chaos in an unsteady
  vortical flow.
\newblock {\em J. Fluid Mech.}, 214:347--394, 1990.

\bibitem[RW13]{roblitz2013fuzzy}
S.~R{\"o}blitz and M.~Weber.
\newblock Fuzzy spectral clustering by {P}{C}{C}{A}+: {A}pplication to {M}arkov
  state models and data classification.
\newblock {\em Advances in Data Analysis and Classification}, 7(2):147--179,
  2013.

\bibitem[SG08]{singh2008unified}
A.~P. Singh and G.~J. Gordon.
\newblock A unified view of matrix factorization models.
\newblock In {\em Machine Learning and Knowledge Discovery in Databases:
  European Conference, ECML PKDD 2008, Antwerp, Belgium, September 15-19, 2008,
  Proceedings, Part II 19}, pages 358--373. Springer Berlin Heidelberg, 2008.

\bibitem[SLM05]{shadden2005definition}
S.~C. Shadden, F.~Lekien, and J.~E. Marsden.
\newblock Definition and properties of lagrangian coherent structures from
  finite-time {L}yapunov exponents in two-dimensional aperiodic flows.
\newblock {\em Physica D: Nonlinear Phenomena}, 212(3-4):271--304, 2005.

\bibitem[SRS{\etalchar{+}}08]{shashanka2008probabilistic}
M.~Shashanka, B.~Raj, P.~Smaragdis, et~al.
\newblock Probabilistic latent variable models as nonnegative factorizations.
\newblock {\em Computational intelligence and neuroscience}, 2008, 2008.

\bibitem[SS13]{schutte2013metastability}
C.~Sch{\"u}tte and M.~Sarich.
\newblock {\em Metastability and Markov State Models in Molecular Dynamics},
  volume~24.
\newblock American Mathematical Soc., 2013.

\bibitem[Tsy09]{tsybakov2009nonparametric}
A.~B. Tsybakov.
\newblock {\em Introduction to nonparametric estimation}.
\newblock Springer Series in Statistics. Springer, New York, 2009.
\newblock Revised and extended from the 2004 French original, Translated by
  Vladimir Zaiats.

\bibitem[UHZ{\etalchar{+}}16]{udell2016generalized}
M.~Udell, C.~Horn, R.~Zadeh, S.~Boyd, et~al.
\newblock Generalized low rank models.
\newblock {\em Foundations and Trends{\textregistered} in Machine Learning},
  9(1):1--118, 2016.

\bibitem[Ula60]{ulam1960collection}
S.~M. Ulam.
\newblock {\em A collection of mathematical problems}, volume~8.
\newblock Interscience Publishers, 1960.

\bibitem[Was04]{wasserman2004all}
L.~Wasserman.
\newblock {\em All of statistics: a concise course in statistical inference},
  volume~26.
\newblock Springer, 2004.

\bibitem[WKZ{\etalchar{+}}18]{wu2018nonnegative}
W.~Wu, S.~Kwong, Y.~Zhou, Y.~Jia, and W.~Gao.
\newblock Nonnegative matrix factorization with mixed hypergraph regularization
  for community detection.
\newblock {\em Information Sciences}, 435:263--281, 2018.

\bibitem[WN20]{wu2020variational}
H.~Wu and F.~No{\'e}.
\newblock Variational approach for learning {M}arkov processes from time series
  data.
\newblock {\em Journal of Nonlinear Science}, 30(1):23--66, 2020.

\bibitem[WZ12]{wang2012nonnegative}
Y.-X. Wang and Y.-J. Zhang.
\newblock Nonnegative matrix factorization: {A} comprehensive review.
\newblock {\em IEEE Transactions on knowledge and data engineering},
  25(6):1336--1353, 2012.

\bibitem[YL13]{yang2013overlapping}
J.~Yang and J.~Leskovec.
\newblock Overlapping community detection at scale: {A} nonnegative matrix
  factorization approach.
\newblock In {\em Proceedings of the sixth ACM international conference on Web
  search and data mining}, pages 587--596, 2013.

\end{thebibliography}

\end{document}